\documentclass[accepted]{uai2026} 
                        
\usepackage[american]{babel}
\usepackage{makecell}
\usepackage{rotating}
\usepackage{natbib} 
\usepackage{mathtools} 
\usepackage{booktabs} 
\usepackage{tikz} 

\usepackage[textsize=tiny]{todonotes}
\usepackage{multirow}

\usepackage{ascmac}
\usepackage{float}
\usepackage{perpage}
\MakeSorted{figure}
\MakeSorted{table}

\usepackage{url}
\usepackage{natbib}
\usepackage{chapterbib}

\usepackage{color}
\usepackage{tikz}
\tikzset{%
mynode/.style={circle,minimum width=.5ex, fill=none,draw}, 
myfillnode/.style={circle,minimum width=.5ex, fill=lightgray,draw}, 
}

\usepackage{pgfplots}
\pgfplotsset{compat=newest}

\usepackage{amssymb}
\usepackage{natbib}

\newcommand{\indep}{\mathrel{\perp\!\!\!\perp}}
\usepackage{amsmath}               
\usepackage{lscape}
\usepackage{algorithm}
\usepackage{color}
\usepackage{tikz}
\usepackage{amsmath,amsthm}
\theoremstyle{plain}
\newtheorem{theorem}{Theorem}
\newtheorem{definition}{Definition}

\newtheorem{lemma}{Lemma}
\newtheorem{proposition}{Proposition}
\newtheorem{corollary}{Corollary}
\newtheorem{remark}{Remark}

\newtheorem*{remark*}{Remark}
\newtheorem*{example}{Example}
\newtheorem*{corollary*}{Corollary}
\usepackage{multirow}
\usepackage{comment}
\usepackage{here}
\allowdisplaybreaks[4]
\usepackage{bbm}
\usepackage{caption}
\usepackage{bbding}
\usepackage{afterpage}

\usepackage{enumitem}

\usepackage{mathrsfs}

\newcommand{\jin}[1]{\textcolor{blue}{[[#1]]}}

\newcommand{\yuta}[1]{\textcolor{red}{#1}}

\newcommand{\naoya}[1]{\textcolor{purple}{#1}}
\usepackage{soul}

\usepackage{float}

\usepackage{amsmath}
\usepackage{amsthm}
\usepackage{amssymb}
\usepackage{mathtools}
\usepackage{thmtools}

\usepackage{stackengine}
\def\defeq{\mathrel{\ensurestackMath{\stackon[1pt]{=}{\scriptscriptstyle\Delta}}}}

\usepackage{xcolor}

\usepackage{makecell}
\usepackage{graphicx}
\usepackage{arydshln}

\newcommand{\Var}{\mathrm{Var}}
\newcommand{\Cov}{\mathrm{Cov}}
\newcommand{\Cor}{\mathrm{Cor}}

\usepackage{algorithm}
\usepackage{algorithmic}

\usepackage{enumitem}
\newcommand{\hide}[1]{}

\hypersetup{hidelinks}

\title{Identification and Bounding of Central Moments of Causal Effects\\ Using Marginal Moments Information}

\author[1,*]{\href{naoya.hashimoto@mbzuai.ac.ae}{Naoya Hashimoto}}
\author[1,*]{\href{yuta.kawakami@mbzuai.ac.ae}{Yuta Kawakami}}
\author[1]{\href{jin.tian@mbzuai.ac.ae}{Jin Tian}}
\affil[1]{%
Mohamed bin Zayed University of Artificial Intelligence (MBZUAI)\\
Abu Dhabi, UAE
}

\newif\ifuaiEqualContribPrinted
\makeatletter
\ifUAI@accepted
  \g@addto@macro\@thanks{%
    \ifuaiEqualContribPrinted\else
      \global\uaiEqualContribPrintedtrue
      \footnotetext[1]{Equal contribution.}%
    \fi
  }%
\fi
\makeatother

\begin{document}
\maketitle

\begin{abstract}
Evaluating the causal effect of a treatment on an outcome is a central objective in causal inference. While the average causal effect  summarizes the mean impact of treatment, the central moments of the individual causal effect (ICE) characterize the shape of the ICE distribution, thereby revealing the extent and structure of treatment effect heterogeneity across individuals. This paper investigates the identification and bounding of the central moments of the ICE using only the marginal central moments of each potential outcome (PO). Compared with existing approaches that require knowledge of the full marginal distributions of the POs, marginal moment information is often substantially easier to obtain in empirical applications. Finally, we illustrate the practical relevance of our  results through two empirical case studies.
\end{abstract}

\section{Introduction}\label{sec:intro}

Evaluating the causal effect of a treatment ($X$) on an outcome ($Y$) from observed data is a central objective in causal inference.
Researchers often focus on the average causal effect (ACE) $\mathbb{E}[Y_1-Y_0]$, where $Y_x$ denotes the potential outcome (PO) of an intervention $do(X=x)$. 
However, the ACE ignores heterogeneity in causal effects.
To reveal the heterogeneity of causal effects across individuals, the distribution of individual causal effects (ICE) $Y_1-Y_0$ has become an increasingly important focus in causal inference.
\citet{Fan2010} provided  sharp bounds for the distribution function of ICE, i.e., $\mathbb{P}(Y_1-Y_0<y)$.
\citet{Post2025} provided an identification result for the distribution function of ICE under the assumption, named "Independent Effect Deviation" (IED), that the ICE $Y_1-Y_0$ is independent of the baseline $Y_0$, i.e., $(Y_1-Y_0) \indep Y_0$.
Recently, \citet{Kawakami2025_moments} instead studied 
the central moments of ICE, $\overline{\mu}^{(m)}\defeq \mathbb{E}[(Y_1-Y_0-\mathbb{E}[Y_1-Y_0])^m]$,  
to characterize the shape of the ICE distribution. 
They provided identification and bounding results for $\overline{\mu}^{(m)}$. 
{Summary statistics derived from the central moments—including variance, skewness, and kurtosis—are widely used in statistics. In particular, variance captures the extent of treatment effect heterogeneity; skewness indicates asymmetry in treatment response, helping identify whether a subset of individuals benefits (or is harmed) disproportionately; and kurtosis reflects tail behavior, which is important in applications where extreme responses (e.g., adverse effects) are of concern.}

These previous studies \citep{Post2023, Post2025, Kawakami2025_moments} assume access to the full marginal distributions of each PO, namely $\mathbb{P}(Y_1<y_1)$ and $\mathbb{P}(Y_0<y_0)$.
In many empirical settings, however, only summary statistics—such as mean, variance, skewness, or kurtosis—are reported, while the underlying distributions and individual-level data are unavailable. 
This may occur due to privacy or confidentiality restrictions, 
or in historically accumulated studies that report only aggregated statistics rather than the original datasets.  
In such settings, the marginal central moments of each PO,  
$\overline{\sigma}^{(k)}_1\defeq\mathbb{E}[(Y_1-\mathbb{E}[Y_1])^k]$ and
$\overline{\sigma}^{(k)}_0\defeq\mathbb{E}[(Y_0-\mathbb{E}[Y_0])^k]$, 
are typically more accessible in practice, even though they contain  less information than the full marginal distributions.

This paper provides methods for evaluating the central moments of ICE $\overline{\mu}^{(m)}$ from the marginal central moments  $(\overline{\sigma}^{(k)}_1,\overline{\sigma}^{(k)}_0)$. $\overline{\mu}^{(m)}$ characterize the shape of the distribution of ICE and reveal the heterogeneity of causal effects across individuals. They can be used to compute key statistical measures of causal effects heterogeneity such as standard deviation (SD) $\sqrt{\overline{\mu}^{(2)}}$, skewness ${\overline{\mu}^{(3)}}/{(\overline{\mu}^{(2)})^{3/2}}$, and kurtosis ${\overline{\mu}^{(4)}}/{(\overline{\mu}^{(2)})^{2}}$. SD captures dispersion: a large SD indicates substantial heterogeneity, with ICEs potentially far from the ACE, whereas a small SD implies that ICEs cluster tightly around the ACE. Skewness measures distributional asymmetry; positive skewness corresponds to a longer right tail, and negative skewness to a longer left tail. Kurtosis reflects tail thickness and concentration, with higher values signaling 
a greater likelihood of extreme causal effects.

We study identifying and bounding the central moments of ICE $\overline{\mu}^{(m)}$ using \emph{only} the marginal central moments $(\overline{\sigma}^{(k)}_1,\overline{\sigma}^{(k)}_0)$.
In contrast to prior approaches, our framework does not require access to the full marginal distributions of $Y_1$ and $Y_0$; it relies solely on their marginal central moments.
Our theoretical contributions are as follows: 
\begin{itemize}
\item We establish  identification results for the central moments of ICE using only marginal central moment information from $Y_1$ and $Y_0$ under the IED condition. 
\item We derive bounds for the central moments of ICE using only marginal central moment information from $Y_1$ and $Y_0$. 
Table~\ref{tab:tab1} summarizes our results for bounding $\overline{\mu}^{(m)}$ using  $(\overline{\sigma}^{(k)}_1,\overline{\sigma}^{(k)}_0)$ for a single $k$. All bounds reported in the table are sharp. 
{In addition, Theorems~\ref{cor:bounding_mu2_intersection}, \ref{cor4}, and \ref{cor5}, together with Corollary~\ref{cor-mum}, provide bounds for $\overline{\mu}^{(m)}$ using multiple $(\overline{\sigma}^{(k)}_1,\overline{\sigma}^{(k)}_0)$ by intersecting the single-$k$ bounds in Table~\ref{tab:tab1}; these bounds are generally not sharp.}
\end{itemize}

To our knowledge, these identification and bounding results are novel. The problem of evaluating central moments of the ICE using only marginal central moments of $Y_1$ and $Y_0$ has not previously been studied in a systematic manner.

Finally, we show the practical relevance of our theoretical  results through two empirical case studies.

\begin{table}[tb]
\centering
\caption{Summary table of the theorems in this paper for bounding $\overline{\mu}^{(m)}$ using  $(\overline{\sigma}^{(k)}_1,\overline{\sigma}^{(k)}_0)$ for a single $k$. 
{All bounds reported in this table are sharp.}
{Nontrivial} bounds for even $m$ mean that the resulting bounds are tighter than $[0,\infty]$, whereas {nontrivial} bounds for odd $m$ mean that the resulting bounds are tighter than $[-\infty,\infty]$.
LB denotes lower bound, and UB denotes upper bound.
}
\label{tab:tab1}
\scalebox{1}{
\begin{tabular}{c|ccc}
    \hline   
    Thm. & Target $m$ 
    & Available $k$
    & {Nontrivial} bounds?\\
    \hline
    \hline   
    
   \ref{thm:bound_of_var}  & 2 & 2 & $\checkmark$ \\
      \ref{thm:mu2_from_3} & 2 & 3 & $\times$\\ 
     \ref{thm:mu2_from_4} & 2 & 4 & $\checkmark$ (only UB)\\
      \hdashline
      \ref{thm:mu3_from_2} & 3 & 2 & $\times$ \\
      \ref{thm:mu3_from_3} & 3 & 3 & $\times$ \\
    \ref{thm:mu3_from_4} & 3 & 4 & $\checkmark$ \\

      \hdashline
      \ref{thm:mu4_from_2} & 4 & 2 & $\checkmark$ (only LB) \\
     \ref{thm:mu4_from_3} & 4 & 3 & $\times$ \\
      \ref{thm:mu4_from_4}  & 4 & 4 & $\checkmark$\\
      \hline
    \ref{thm:minkowski_bound} & Even $m$ & $m$ & $\checkmark$ \\
    \ref{thm:even_m_single_k} & Even $m$ & Even $k< m$ & $\checkmark$ (only LB) \\
    \ref{thm:even_m_single_k} & Even $m$ & Odd $k< m$ & $\times$ \\
  \ref{thm:odd_unbounded_from_single} & Odd $m$ & $k\leq m$ & $\times$ \\
    \hline
\end{tabular}
}
\end{table}

\section{Motivating Examples}
\label{sec2}

We begin by motivating the problem of bounding the central moments of the ICE using the marginal central moments of each PO.
In many experimental studies, datasets containing individual-level information are not publicly released due to privacy or confidentiality concerns.  
By contrast, moment information—such as the variance, skewness, or kurtosis—of treatment and control groups in randomized controlled trials (RCTs) is frequently reported as basic summary statistics.

For example, \citet{Moseley2002} conducted a randomized, placebo-controlled trial of arthroscopic surgery for osteoarthritis of the knee.
A total of 180 patients with osteoarthritis of the knee were randomly assigned to arthroscopic d\'ebridement ($X=1$), arthroscopic lavage, or placebo surgery ($X=0$) (we chose two of the three groups for our analysis). 
\citet{Moseley2002} reported the Physical Functioning Scale (PFS) score at 24 months (outcome $Y$), which records the time in seconds that a patient requires to walk 30 m 
and to climb up and down a flight of stairs as quickly as possible, where longer times indicate poorer functioning.
The article reports the mean and standard deviation (SD) of PFS for each group, given in Table \ref{tab:moseley2002-pfs}.

\citet{Spirtovic2023} reported an experimental study on the influence of a 12-week preventive program of mixed aerobics on body composition in healthy adult women.
A total of 64 
women were randomly divided into two groups and completed both the initial and final measurements.
The experimental group ($X=1$) took part in the exercise program of mixed aerobics, whereas the control group ($X=0$) did not participate in any form of organized physical activity.
Muscle mass percentage (MM\%) in the final state (outcome $Y$) is calculated and expressed as a percentage of total body mass.
The study reports the mean, SD, skewness, and kurtosis for MM\% in each group, given in Table \ref{tab:spirtovic_mm_moments}.

Because the authors of the two studies did not make the underlying dataset publicly available or provide the full distribution functions for each group, it is impossible to evaluate the central moments of the ICE 
using the methods based on marginal distributions of each group in 
\citep{Post2023, Post2025, Kawakami2025_moments}.

Additionally, several experimental studies (e.g., \citet{Hebert1999, vanderBerghe2001, Rivers2001, Nice2009}) reported the means and standard deviations for each group. 
Other  studies (e.g., \citet{Lucena2020, Goble2021, Ulacs2022, Ozdemir2024, Stefanescu2024}) provided higher-order moment information, such as skewness and kurtosis, even though the underlying individual-level datasets are not publicly available. 
{These studies compared the differences in marginal summary statistics between the two groups (control and treatment). Our results in this paper enable re-analysis of these studies to infer properties of the \emph{distribution} of causal effects using only the reported marginal moment information.}

\begin{table}[t]
\centering
\caption{Summary statistics reported by \citep{Moseley2002}. SD: standard deviation. 
}
\label{tab:moseley2002-pfs}
\scalebox{1}{
\begin{tabular}{lccc}
\hline
Group & Sample Size &  Mean & SD \\
\hline\hline
\makecell[c]{Placebo surgery\\ ($X=0$)} & 44 & 47.7 & 12.0 \\
\hdashline
\makecell[c]{Arthroscopic\\ d\'ebridement ($X=1$)} & 44 & 52.6 & 16.4 \\
\hline
\end{tabular}
}
\end{table}

\begin{table*}[t]
\centering
\caption{Summary statistics
reported by \citep{Spirtovic2023}. 
SD: standard deviation, 
Skew: skewness, Kurt: kurtosis, 
2nd M: second central moments, 
3rd M: third central moments, 
4th M: fourth central moments.
(We calculated 2nd M, 3rd M, and 4th M from their SD, Skew, and Kurt.)
}
\label{tab:spirtovic_mm_moments}
\scalebox{1}{
\begin{tabular}{lrrrrr|rrr}
\hline
Group & Sample Size & Mean & SD & Skew & Kurt & 2nd M & 3rd M & 4th M \\
\hline \hline
Control ($X=0$) & 30 & 45.469 & 4.573 & 0.648  & 3.652   & 20.912 & 61.970 & 1597.113 \\
Experimental ($X=1$) & 34 & 44.125 & 5.262 & 0.915  & 3.273    & 27.689 & 133.313 & 2509.282 \\
\hline
\end{tabular}
}
\end{table*}

\section{Notation and Background}\label{sec:notation_and_backgraound}

We represent each variable with a capital letter $(X)$ and its realized value with a lowercase letter $(x)$. 
We denote
$\mathbb{E}[Y]$ as the expectation of $Y$.
Let $\overline{\mathbb{R}}=[-\infty,\infty]=\mathbb{R}\cup \{\pm\infty\}$ be the extended real line.

We use the language of structural causal models (SCM) as our basic semantic and inferential framework \citep{Pearl09}. We denote the potential outcome (PO) $Y$ under intervention $do({x})$ by $Y_{{x}}$. 
The individual causal effect (ICE)  is defined as $Y_1-Y_0$ \citep{Holland1986}, and the central moments of ICE are defined as follows \citep{Kawakami2025_moments}: 
\begin{definition}[The central moments of ICE]
The $m$-th central moment of ICE is defined as 
\begin{equation}
\begin{aligned}
\overline{\mu}^{(m)}\defeq \mathbb{E}[(Y_1-Y_0-\mathbb{E}[Y_1-Y_0])^m].
\end{aligned}
\end{equation}
\end{definition}
\citet{Post2025} showed that, under the "Independent Effect Deviation" (IED) assumption\footnote{IED assumes that the individual causal effect $Y_1 - Y_0$ is independent of the baseline outcome $Y_0$. 
Intuitively, this means that how much a unit benefits from treatment is unrelated to its baseline level.
This assumption may be plausible in settings where the mechanisms generating the treatment effect are
independent of those determining the baseline outcome.
Conversely, the assumption is less plausible in scenarios where treatment effects  depend on baseline levels—for example, when individuals with higher (or lower) baseline outcomes tend to benefit more (or less) from the treatment.}, $(Y_1-Y_0) \indep Y_0$, the distribution of ICE is identified from the marginal distributions of $Y_1$ and  $Y_0$ .
\citet{Kawakami2025_moments} derived identification results for the central moments $\overline{\mu}^{(m)}$ of ICE in terms of the marginals of $Y_1$ and $Y_0$ under a monotonicity assumption, and in the absence of monotonicity, established bounds for $\overline{\mu}^{(m)}$ in terms of the marginals of $Y_1$ and $Y_0$. In this paper, we study the identification and bounding of $\overline{\mu}^{(m)}$ using only the central moments of the marginals of $Y_1$ and $Y_0$, without assuming full knowledge of their marginal distributions.

\section{Identification of central moments of ICE under IED condition}
\label{subsec:Identification_under_IED}
We study the evaluation of the central moments of ICE, $\overline{\mu}^{(m)}$, through 
the marginal central moments of each PO, 
\begin{equation}
\begin{aligned}
&\overline{\sigma}^{(k)}_1\defeq\mathbb{E}[(Y_1-\mathbb{E}[Y_1])^k],\\
&\overline{\sigma}^{(k)}_0\defeq\mathbb{E}[(Y_0-\mathbb{E}[Y_0])^k].
\end{aligned}
\end{equation}
In RCTs, $\overline{\sigma}^{(k)}_1$ and $\overline{\sigma}^{(k)}_0$ correspond to the $k$-th central moment of the outcome for the treatment and control groups, respectively. 
 In observational studies with no confounders, i.e., $Y_x \indep X$, $\overline{\sigma}^{(k)}_1$ and $\overline{\sigma}^{(k)}_0$ correspond to the $k$-th central moment of the outcome for subgroups conditioning on $X=1$ and $X=0$, respectively. In observational studies with observed covariates $C$ and no unmeasured confounding, i.e., $Y_x \indep X | C$, 
 the quantities $\overline{\sigma}^{(k)}_1$ and $\overline{\sigma}^{(k)}_0$ can be identified via covariate adjustment. For discrete $C$,
\begin{equation}
\label{eq:covariate_adjustment_moments}
\begin{aligned}
\overline{\sigma}^{(k)}_1
&=
\sum_{c \in \Omega_C}
\mathbb{E}[(Y-\mathbb{E}[Y_1])^k | X=1,C=c]\mathbb{P}(C=c),\\
\overline{\sigma}^{(k)}_0
&=
\sum_{c \in \Omega_C}
\mathbb{E}[(Y-\mathbb{E}[Y_0])^k | X=0,C=c]\mathbb{P}(C=c),
\end{aligned}
\end{equation}
where
\begin{equation}
\label{eq:covariate_adjustment_means}
\begin{aligned}
\mathbb{E}[Y_1]
&=
\sum_{c\in \Omega_C}
\mathbb{E}[Y | X=1,C=c]\mathbb{P}(C=c),\\
\mathbb{E}[Y_0]
&=
\sum_{c\in \Omega_C}
\mathbb{E}[Y | X=0,C=c]\mathbb{P}(C=c).
\end{aligned}
\end{equation}
For continuous $C$, the sums are replaced by integrals.

The same arguments can be applied within strata of observed covariates. If the marginal central moments of $Y_x | C=c$ are available, the identification and bounding results below apply after replacing $Y_x$ by $Y_x | C=c$. The resulting stratum-specific moments describe the distribution of the ICE within the subpopulation $C=c$, while the conditional average causal effect $\mathbb{E}[Y_1-Y_0 | C=c]$ describes only the corresponding mean.

In this section, we discuss the identification of the central moments of the ICE.
\citet{Post2025} provided identification of the distribution of the ICE using the marginal distributions of each PO. 
In contrast, we provide identification of the central moments of the ICE using only the central moments of each PO.

{\bf Identification of the variance $\overline{\mu}^{(2)}$ of ICE.}
We first discuss the identification of the variance 
of ICE using the variance of each PO.
The variance of ICE quantifies the degree of dispersion in ICE and serves as a key measure of causal effect heterogeneity \citep{Kawakami2025_moments}.
A large variance indicates substantial heterogeneity in causal effects, whereas a small variance indicates relative homogeneity.

\citet{Post2025} showed that the variance of ICE is identifiable under the IED condition. 
We restate their results formally in the following proposition.
\begin{proposition}
\label{prop2}
Assume $\overline{\sigma}^{(2)}_1$ and $\overline{\sigma}^{(2)}_0$ exist and are available. If the IED condition $(Y_1-Y_0) \indep Y_0$ holds, then 
the variance of ICE  is identified as 
\begin{equation}
\begin{aligned}
&\overline{\mu}^{(2)}=\overline{\sigma}^{(2)}_1-\overline{\sigma}^{(2)}_0.
\end{aligned}
\end{equation}    
\end{proposition}
The variance of ICE is given by the difference between the variances of the two POs. 

{\bf Remark.}
Under the IED, $\overline{\sigma}^{(2)}_1-\overline{\sigma}^{(2)}_0\geq 0$ holds by Proposition \ref{prop2}, which provides a testable condition to falsify the IED assumption. 


\citet{Kawakami2025_moments} identified the central moments of ICE from the marginal distributions of $Y_1$ and $Y_0$ under a monotonicity assumption (their Assumption~2). However, $\overline{\mu}^{(2)}$ is not identifiable from $(\overline{\sigma}_1^{(2)},\overline{\sigma}_0^{(2)})$ alone, even under this monotonicity assumption. A counterexample is provided in Appendix \ref{sec_app:counterexample_assump2}.

{\bf Identification of the higher-order central moments of the ICE.}
Next, we 
provide an identification theorem for 
the higher-order central moments of the ICE through the central moments of each PO $(\overline{\sigma}^{(k)}_1,\overline{\sigma}^{(k)}_0)$, 
which is a generalization of Proposition \ref{prop2}. 
\begin{restatable}{theorem}{Identification}
\label{thm:identification_under_independence}
Assume $(\overline{\sigma}^{(k)}_1,\overline{\sigma}^{(k)}_0)$, {$k=1,2,\dots,m-3,m-2,m$}, exist and are available. 
If the IED condition $(Y_1-Y_0) \indep Y_0$ holds,
then
the  $m$-th central moment of ICE  is recursively identified as 
\begin{equation}
\label{eq:identification_under_independence}
\begin{aligned}
\overline{\mu}^{(m)}& =\overline{\sigma}^{(m)}_1-\overline{\sigma}^{(m)}_0 -\sum_{\ell=2}^{m-2} \binom{m}{\ell} \overline{\mu}^{(m-\ell)}\overline{\sigma}^{(\ell)}_0.
\end{aligned}
\end{equation}
\end{restatable}
Theorem~\ref{thm:identification_under_independence} shows that the central moments of ICE are identified from marginal central moments of each PO under IED, without requiring the complete  knowledge of the marginal distributions used in \citep{Post2025}. 

We have the following important corollaries, as $\overline{\mu}^{(3)}$ and $\overline{\mu}^{(4)}$ are used to compute skewness and kurtosis: 
\begin{corollary}
If $(\overline{\sigma}^{(3)}_1,\overline{\sigma}^{(3)}_0)$ exist and are available,
and the IED condition $(Y_1-Y_0) \indep Y_0$ holds,
then $\overline{\mu}^{(3)}$ is identified as follows:
\begin{equation}
\begin{aligned}
    \overline{\mu}^{(3)}& =\overline{\sigma}^{(3)}_1-\overline{\sigma}^{(3)}_0 .
\end{aligned}
\end{equation}
\end{corollary}
Then the skewness of the ICE, i.e., $\overline{\mu}^{(3)} / (\overline{\mu}^{(2)})^{3/2}$, is identified from $(\overline{\sigma}^{(k)}_1,\overline{\sigma}^{(k)}_0)$,  $k=2, 3$, under IED. 
\begin{corollary}
If $(\overline{\sigma}^{(2)}_1,\overline{\sigma}^{(2)}_0)$ and $(\overline{\sigma}^{(4)}_1,\overline{\sigma}^{(4)}_0)$ exist and are available,
and the IED condition $(Y_1-Y_0) \indep Y_0$ holds,
then $\overline{\mu}^{(4)}$ is identified as follows:
\begin{equation}
\overline{\mu}^{(4)}
=\overline{\sigma}^{(4)}_1-\overline{\sigma}^{(4)}_0 -6(\overline{\sigma}^{(2)}_1-\overline{\sigma}^{(2)}_0)\overline{\sigma}^{(2)}_0.
\end{equation}
\end{corollary}
Then the kurtosis of the ICE, i.e., $\overline{\mu}^{(4)} / (\overline{\mu}^{(2)})^{2}$, is identified from $(\overline{\sigma}^{(k)}_1,\overline{\sigma}^{(k)}_0)$,  $k=2, 4$, under IED.

\section{Bounding 2nd, 3rd, and 4th central moments of ICE}
\label{sec6}

We next study the bounding of the central moments $\overline{\mu}^{(m)}$ of ICE  
from the marginal central moments $(\overline{\sigma}^{(k)}_1,\overline{\sigma}^{(k)}_0)$  without imposing the IED condition. 
The marginal means $\mathbb{E}[Y_1]$ and $\mathbb{E}[Y_0]$ are also assumed to be available; however, they do not provide additional power for bounding  $\overline{\mu}^{(m)}$.
In practice, interest in the moments of causal effects primarily concerns the variance, skewness, and kurtosis. These can be calculated through the second, third, and fourth central moments of the ICE, i.e., $\overline{\mu}^{(2)}$, $\overline{\mu}^{(3)}$, and $\overline{\mu}^{(4)}$, which are the focus of this section.


{\bf Sharpness.} 
The notion of sharpness varies across  contexts \citep{Zhang2025bounds, Jonzon2025}.
We derive bounds for $\overline{\mu}^{(m)}$ in the form of a closed interval $[l,u]$, where $l,u \in \overline{\mathbb{R}}$, given available marginal central moments $(\overline{\sigma}^{(k)}_1,\overline{\sigma}^{(k)}_0)$. 
We say that $[l,u]$ is a 
\emph{sharp}
bound for $\overline{\mu}^{(m)}$ if there exist no strictly tighter closed bounds for $\overline{\mu}^{(m)}$ under the given moment information.
Our sharpness  means that the closed interval cannot be further tightened without additional information (e.g., the marginal distributions of $Y_1$ and $Y_0$) or assumptions.

Note that we restrict the bounds to a closed interval.
Some researchers instead define the bounds using an open interval.
However, the interpretation is essentially unchanged whether the interval is written as $[l, u]$ or $(l, u)$, since the distinction concerns only the inclusion or exclusion of the boundary points and does not affect the substantive properties under consideration.

\subsection{Bounds of $\overline{\mu}^{(2)}$}
\label{subsec:Bounding_of_mu2(1)}

We first study bounds for the variance of ICE by analyzing the consequences of violating the IED condition. 
We denote the correlation of $Y_1-Y_0$ and $Y_0$ as $\lambda$: 
\begin{equation}
\lambda\defeq \frac{\mathbb{E}[(Y_1-Y_0-\mathbb{E}[Y_1-Y_0])(Y_0-\mathbb{E}[Y_0])]}{\sqrt{\Var(Y_1-Y_0)}\sqrt{\Var(Y_0)}}.
\end{equation}
We define $\lambda=0$ if $\Var(Y_1-Y_0)\Var(Y_0)=0$. 
When the IED condition holds, we have $\lambda = 0$. 
A positive correlation implies that units with a higher baseline outcome $Y_0$ exhibit a larger ICE, whereas units with a lower baseline outcome exhibit a smaller ICE; conversely, a negative correlation implies that higher baseline outcomes are associated with smaller ICEs and lower baseline outcomes with larger ICEs.

We examine how the variance of the ICE is influenced by the parameter $\lambda$.
If the correlation of $Y_1-Y_0$ and $Y_0$ is fixed, $\overline{\mu}^{(2)}$ is determined as follows:
\begin{restatable}{theorem}{SensAn}
\label{thm:sensitivity_analysis_of_var}

(1). Assume $\overline{\sigma}_1^{(2)} > \overline{\sigma}_0^{(2)}$.
If $\overline{\sigma}_0^{(2)}=0$, then $\lambda=0$ and $\overline{\mu}^{(2)}=\overline{\sigma}_1^{(2)}$.
If $\overline{\sigma}_0^{(2)}>0$, then for any $-1\leq \lambda\leq 1$, $\overline{\mu}^{(2)}$ is uniquely determined as
\begin{equation}
    \overline{\mu}^{(2)} = \left(-\sqrt{\overline{\sigma}^{(2)}_0}\lambda + \sqrt{\overline{\sigma}^{(2)}_1-(1-\lambda^2)\overline{\sigma}^{(2)}_0}\right)^2.
\end{equation}

(2). Assume $\overline{\sigma}^{(2)}_1 = \overline{\sigma}^{(2)}_0$.
Then $\lambda \leq 0$, and $\overline{\mu}^{(2)}$ is uniquely determined as
\begin{equation}
\overline{\mu}^{(2)}
= 4\lambda^2\,\overline{\sigma}^{(2)}_0 .
\end{equation}

(3). Assume $\overline{\sigma}_1^{(2)} < \overline{\sigma}_0^{(2)}$.
We have $-1\leq \lambda\leq -\sqrt{1-\overline{\sigma}^{(2)}_1/\overline{\sigma}^{(2)}_0}$, and
\begin{equation}
\begin{aligned}
&\overline{\mu}^{(2)} \in \Bigg\{\left(-\sqrt{\overline{\sigma}^{(2)}_0}\lambda - \sqrt{\overline{\sigma}^{(2)}_1-(1-\lambda^2)\overline{\sigma}^{(2)}_0}\right)^2,\\
&\left(-\sqrt{\overline{\sigma}^{(2)}_0}\lambda + \sqrt{\overline{\sigma}^{(2)}_1-(1-\lambda^2)\overline{\sigma}^{(2)}_0}\right)^2\Bigg\}.
\end{aligned}
\end{equation}
\end{restatable}
The number of possible values that $\overline{\mu}^{(2)}$ can take given $\lambda$ depends on the relative magnitudes of $\overline{\sigma}_0^{(2)}$ and $\overline{\sigma}_1^{(2)}$.
In addition, when $\overline{\sigma}_1^{(2)} \leq \overline{\sigma}_0^{(2)}$, the feasible range of $\lambda$ is restricted.

By Theorem \ref{thm:sensitivity_analysis_of_var}, we obtain the following bound on $\overline{\mu}^{(2)}$:
\begin{restatable}{theorem}{SecondfromS}
\label{thm:bound_of_var}
Assume that $(\overline{\sigma}^{(2)}_1,\overline{\sigma}^{(2)}_0)$ exist and are available.  
$\overline{\mu}^{(2)}$ is bounded as follows:
\begin{equation}
\label{eq:bound_of_var}
\Big(\sqrt{\overline{\sigma}^{(2)}_1}-\sqrt{\overline{\sigma}^{(2)}_0}\Big)^2
\leq \overline{\mu}^{(2)}\leq 
\Big(\sqrt{\overline{\sigma}^{(2)}_1}+\sqrt{\overline{\sigma}^{(2)}_0}\Big)^2.
\end{equation}
If only $(\overline{\sigma}^{(2)}_1,\overline{\sigma}^{(2)}_0)$ are available, then this bound is sharp.

\end{restatable}

If we know whether $\lambda\geq0$ or $\lambda\leq 0$ a priori, we can obtain tighter bounds:
\begin{restatable}{theorem}{VarSignBound}
\label{theorem:bound_of_var_sign}
Assume  $(\overline{\sigma}_1^{(2)},\overline{\sigma}_0^{(2)})$ exist and are available.

(1). 
If $\lambda\geq 0$, then we have $\overline{\sigma}_1^{(2)}\geq \overline{\sigma}_0^{(2)}$.  
If $\overline{\sigma}_1^{(2)}=\overline{\sigma}_0^{(2)}$, then $\overline{\mu}^{(2)}=0$; 
if $\overline{\sigma}_1^{(2)}>\overline{\sigma}_0^{(2)}$, we have
\begin{equation}
\label{eq:bound_of_var_lambda_ge0}
\Big(\sqrt{\overline{\sigma}_1^{(2)}}-\sqrt{\overline{\sigma}_0^{(2)}}\Big)^2
\leq
\overline{\mu}^{(2)}
\leq
\overline{\sigma}_1^{(2)}-\overline{\sigma}_0^{(2)}.
\end{equation}

(2). 
Assume $\lambda\leq 0$. If  $\overline{\sigma}_1^{(2)}>\overline{\sigma}_0^{(2)}$, we have
\begin{equation}
\label{eq:bound_of_var_lambda_le0}
\overline{\sigma}_1^{(2)}-\overline{\sigma}_0^{(2)}
\leq
\overline{\mu}^{(2)}
\leq
\Big(\sqrt{\overline{\sigma}_1^{(2)}}+\sqrt{\overline{\sigma}_0^{(2)}}\Big)^2.
\end{equation}
If $\overline{\sigma}_1^{(2)}\leq \overline{\sigma}_0^{(2)}$, Eq.~\eqref{eq:bound_of_var} holds.

These bounds are sharp. 

\end{restatable}

Can $(\overline{\sigma}_1^{(3)},\overline{\sigma}_0^{(3)})$ or
$(\overline{\sigma}_1^{(4)},\overline{\sigma}_0^{(4)})$, which  may be available in practice, provide information about $\overline{\mu}^{(2)}$?


\begin{restatable}{theorem}{SecondfromT}
\label{thm:mu2_from_3}
Assume that
$(\overline{\sigma}^{(3)}_1,\overline{\sigma}^{(3)}_0)$ exist and are available.
Then we have 
$0\leq \overline{\mu}^{(2)} \leq \infty$.
This bound is sharp if only $(\overline{\sigma}^{(3)}_1,\overline{\sigma}^{(3)}_0)$ are available.
\end{restatable}
$(\overline{\sigma}^{(3)}_1,\overline{\sigma}_0^{(3)})$ alone does not provide a nontrivial bound for $\overline{\mu}^{(2)}$.
On the other hand, $(\overline{\sigma}^{(4)}_1, \overline{\sigma}_0^{(4)})$ provides a nontrivial upper bound as follows:
\begin{restatable}{theorem}{SecondfromF}
\label{thm:mu2_from_4}
Assume that $(\overline{\sigma}^{(4)}_1,\overline{\sigma}^{(4)}_0)$ exist and are available. Then we have
\begin{equation}
0\leq \overline{\mu}^{(2)}\leq \Big(\sqrt[4]{\overline{\sigma}_1^{(4)}}+\sqrt[4]{\overline{\sigma}_0^{(4)}}\Big)^2.
\end{equation}
This bound is sharp if only $(\overline{\sigma}^{(4)}_1,\overline{\sigma}^{(4)}_0)$ are available.
\end{restatable}

Furthermore, 
by taking the intersection of the bounds in Theorems~\ref{thm:bound_of_var},~\ref{thm:mu2_from_3}, and \ref{thm:mu2_from_4},
we have the following results:

\begin{restatable}{theorem}{SecondfromComb}
\label{cor:bounding_mu2_intersection}
(1). Assume that  $(\overline{\sigma}_1^{(2)}$,\,$\overline{\sigma}_0^{(2)})$ and 
$(\overline{\sigma}_1^{(3)}$,\,$\overline{\sigma}_0^{(3)})$ exist and are available. Then we have
\begin{equation}
\label{eq:mu2_from_23}
    \Big(\sqrt{\overline{\sigma}^{(2)}_1}-\sqrt{\overline{\sigma}^{(2)}_0}\Big)^2\leq \overline{\mu}^{(2)}
    \leq \Big(\sqrt{\overline{\sigma}^{(2)}_1}+\sqrt{\overline{\sigma}^{(2)}_0}\Big)^2.
\end{equation}
{If only $(\overline{\sigma}_1^{(2)}$,\,$\overline{\sigma}_0^{(2)})$ and  $(\overline{\sigma}_1^{(3)}$,\,$\overline{\sigma}_0^{(3)})$ are available, then} this bound is sharp.

(2). Assume that $(\overline{\sigma}_1^{(2)}$,\,$\overline{\sigma}_0^{(2)})$ and $(\overline{\sigma}_1^{(4)}$,\,$\overline{\sigma}_0^{(4)})$ exist and are available, we have 
\begin{equation}
\label{eq:mu2_from_24}
    \Big(\sqrt{\overline{\sigma}^{(2)}_1}-\sqrt{\overline{\sigma}^{(2)}_0}\Big)^2  \leq \overline{\mu}^{(2)}
    \leq \Big(\sqrt{\overline{\sigma}^{(2)}_1}+\sqrt{\overline{\sigma}^{(2)}_0}\Big)^2.
\end{equation}
{If only $(\overline{\sigma}_1^{(2)}$,\,$\overline{\sigma}_0^{(2)})$ and $(\overline{\sigma}_1^{(4)}$,\,$\overline{\sigma}_0^{(4)})$ are available, then} this bound is sharp.

(3). Assume that $(\overline{\sigma}_1^{(3)}$,\,$\overline{\sigma}_0^{(3)})$ and $(\overline{\sigma}_1^{(4)}$,\,$\overline{\sigma}_0^{(4)})$ exist and are available. Then we have
\begin{equation}
\label{eq:mu2_from_34_intersection}
    0  \leq \overline{\mu}^{(2)} \leq \Big(\sqrt[4]{\overline{\sigma}_1^{(4)}}+\sqrt[4]{\overline{\sigma}_0^{(4)}}\Big)^2. 
\end{equation} 

(4). Assume that $(\overline{\sigma}_1^{(2)}$,\,$\overline{\sigma}_0^{(2)})$,
$(\overline{\sigma}_1^{(3)}$,\,$\overline{\sigma}_0^{(3)})$,
and $(\overline{\sigma}_1^{(4)}$,\,$\overline{\sigma}_0^{(4)})$ exist and are available. 
Then we have
\begin{equation}
\label{eq:bound_of_var_with_234}
\Big(\sqrt{\overline{\sigma}^{(2)}_1}-\sqrt{\overline{\sigma}^{(2)}_0}\Big)^2  \leq \overline{\mu}^{(2)}
\leq \Big(\sqrt{\overline{\sigma}^{(2)}_1}+\sqrt{\overline{\sigma}^{(2)}_0}\Big)^2.
\end{equation}
\end{restatable}

In general, taking the intersection of sharp bounds does not preserve the sharpness.
Indeed, the bounds in (3) and (4) of Theorem \ref{cor:bounding_mu2_intersection} are not sharp: we provide their counterexamples in Appendix \ref{sec_app:intersection_non-sharp}. 
As an exception, the bounds in (1) and (2) are sharp, which means that, given $(\overline{\sigma}_1^{(2)}$,\,$\overline{\sigma}_0^{(2)})$, additionally knowing either $(\overline{\sigma}_1^{(3)}$,\,$\overline{\sigma}_0^{(3)})$ alone or $(\overline{\sigma}_1^{(4)}$,\,$\overline{\sigma}_0^{(4)})$ alone does not further tighten the bound for $\overline{\mu}^{(2)}$. On the other hand, jointly knowing both $(\overline{\sigma}_1^{(3)}$,\,$\overline{\sigma}_0^{(3)})$ and $(\overline{\sigma}_1^{(4)}$,\,$\overline{\sigma}_0^{(4)})$, in addition to $(\overline{\sigma}_1^{(2)}$,\,$\overline{\sigma}_0^{(2)})$,  provides additional identifying information about  $\overline{\mu}^{(2)}$.

The upper bound on $\overline{\mu}^{(2)}$ can be translated into a bound on the probability that the ICE has the opposite sign from the ACE. Let $\tau\coloneqq \mathbb{E}[Y_1-Y_0]$ and suppose that $\overline{\mu}^{(2)}\leq U_2$. If $\tau>0$, then by Cantelli's inequality,
\begin{equation}
\label{eq:cantelli_positive_ace}
\mathbb{P}(Y_1-Y_0\leq 0)\leq\frac{U_2}{U_2+\tau^2}.
\end{equation}
If $\tau<0$, then
\begin{equation}
\label{eq:cantelli_negative_ace}
\mathbb{P}(Y_1-Y_0\geq 0)\leq\frac{U_2}{U_2+\tau^2}.
\end{equation}

\subsection{Bounds of $\overline{\mu}^{(3)}$
}
\label{subsec:Bounding_of_mu3}


We next study the bounding of $\overline{\mu}^{(3)}$ using $(\overline{\sigma}_1^{(2)}$,\,$\overline{\sigma}_0^{(2)})$,
$(\overline{\sigma}_1^{(3)}$,\,$\overline{\sigma}_0^{(3)})$, or
$(\overline{\sigma}_1^{(4)}$,\,$\overline{\sigma}_0^{(4)})$.
First we consider the cases where only one of these is available.
\begin{restatable}{theorem}{ThirdfromS}
\label{thm:mu3_from_2}
Assume  $(\overline{\sigma}_1^{(2)},\overline{\sigma}_0^{(2)})$ exist and are available.

(1). If $\overline{\sigma}_1^{(2)}=\overline{\sigma}_0^{(2)}=0$, we have $\overline{\mu}^{(3)}=0$.

(2). If $\overline{\sigma}_1^{(2)}>0$ or $\overline{\sigma}_0^{(2)}>0$, we have $-\infty \leq \overline{\mu}^{(3)}\leq \infty$. 
If only $(\overline{\sigma}_1^{(2)},\overline{\sigma}_0^{(2)})$ are available, then this bound is sharp.

\end{restatable}

\begin{restatable}{theorem}{ThirdfromT}
\label{thm:mu3_from_3}
Assume that  $(\overline{\sigma}_1^{(3)},\overline{\sigma}_0^{(3)})$ exist and are available.
Then we have
$-\infty \leq \overline{\mu}^{(3)} \leq \infty$.
If only $(\overline{\sigma}_1^{(3)},\overline{\sigma}_0^{(3)})$ are available, then 
this bound is sharp.
\end{restatable}
By Theorems~\ref{thm:mu3_from_2} and \ref{thm:mu3_from_3},  $(\overline{\sigma}_1^{(2)}$,\,$\overline{\sigma}_0^{(2)})$ or 
$(\overline{\sigma}_1^{(3)}$,\,$\overline{\sigma}_0^{(3)})$ alone does not help to obtain a nontrivial bound for  $\overline{\mu}^{(3)}$. 
The condition $\overline{\sigma}_1^{(2)}=\overline{\sigma}_0^{(2)}=0$ corresponds to an extreme case where both $Y_1$ and $Y_0$ are 
constant with probability one.

On the other hand, $(\overline{\sigma}_1^{(4)},\overline{\sigma}_0^{(4)})$ is helpful in bounding $\overline{\mu}^{(3)}$:
\begin{restatable}{theorem}{ThirdfromF}
\label{thm:mu3_from_4}
Assume that $(\overline{\sigma}_1^{(4)},\overline{\sigma}_0^{(4)})$ exist and are available.
Then we have
\begin{equation}\label{eq:mu3_bounds_4only_newextra}
\begin{aligned}
&-\frac{\sqrt{2}}{\sqrt[4]{27}}\Bigl(\sqrt[4]{\overline{\sigma}_1^{(4)}}+\sqrt[4]{\overline{\sigma}_0^{(4)}}\Bigr)^3
\leq \overline{\mu}^{(3)} \\
&\hspace{2cm}\leq \frac{\sqrt{2}}{\sqrt[4]{27}}\Bigl(\sqrt[4]{\overline{\sigma}_1^{(4)}}+\sqrt[4]{\overline{\sigma}_0^{(4)}}\Bigr)^3.
\end{aligned}
\end{equation}
If only $(\overline{\sigma}_1^{(4)},\overline{\sigma}_0^{(4)})$ are available, then 
this bound is sharp.
\end{restatable}

Moreover, if $(\overline{\sigma}_1^{(k)},\overline{\sigma}_0^{(k)})$ ($k=2,3,4$) are jointly available,
by taking the intersection of the bounds in Theorems~\ref{thm:mu3_from_2}-\ref{thm:mu3_from_4},
we have the following results:

\begin{restatable}{theorem}{ThirdfromComb}
\label{cor4}
(1). Assume $(\overline{\sigma}_1^{(2)}$,\,$\overline{\sigma}_0^{(2)})$ and 
$(\overline{\sigma}_1^{(3)}$,\,$\overline{\sigma}_0^{(3)})$ exist and are available.  

(i).
If $\overline{\sigma}_0^{(2)}>0$ and $\overline{\sigma}_1^{(2)}>0$, we have 
$-\infty \leq \overline{\mu}^{(3)} \leq \infty$.
If only $(\overline{\sigma}_1^{(2)}$,\,$\overline{\sigma}_0^{(2)})$ and  $(\overline{\sigma}_1^{(3)}$,\,$\overline{\sigma}_0^{(3)})$ are available, then this bound is sharp.

(ii).
If $\overline{\sigma}_1^{(2)}=\overline{\sigma}_0^{(2)}=0$, we have $\overline{\mu}^{(3)}=0$. Otherwise, if $\overline{\sigma}_1^{(2)}=0$, we have $\overline{\mu}^{(3)}=-\overline{\sigma}_0^{(3)}$; if $\overline{\sigma}_0^{(2)}=0$, we have $\overline{\mu}^{(3)}=\overline{\sigma}_1^{(3)}$.

(2) Assume that $(\overline{\sigma}_1^{(4)}, \overline{\sigma}_0^{(4)})$ exist and are available, and that at least one of $(\overline{\sigma}_1^{(2)}, \overline{\sigma}_0^{(2)})$ and $(\overline{\sigma}_1^{(3)}, \overline{\sigma}_0^{(3)})$ exist and are available.
Then we have
\begin{align}
\label{eq:bound_of_mu3_with_4+23}
&-\frac{\sqrt{2}}{\sqrt[4]{27}}\Bigl(\sqrt[4]{\overline{\sigma}_1^{(4)}}+\sqrt[4]{\overline{\sigma}_0^{(4)}}\Bigr)^3  \leq \overline{\mu}^{(3)}  \nonumber\\
&\hspace{2cm}\leq \frac{\sqrt{2}}{\sqrt[4]{27}}\Big(\sqrt[4]{\overline{\sigma}_1^{(4)}}+\sqrt[4]{\overline{\sigma}_0^{(4)}}\Big)^3.  
\end{align}
\end{restatable}

The bound in (2) of Theorem \ref{cor4} is not sharp: we provide a counterexample in Appendix \ref{sec_app:intersection_non-sharp}. 
On the other hand, the bound in (1) of Theorem \ref{cor4} is sharp, which means that jointly knowing $(\overline{\sigma}_1^{(2)}$,\,$\overline{\sigma}_0^{(2)})$ and 
$(\overline{\sigma}_1^{(3)}$,\,$\overline{\sigma}_0^{(3)})$  does not provide a nontrivial bound for  $\overline{\mu}^{(3)}$ (except in extreme cases). However, once $(\overline{\sigma}_1^{(4)},\overline{\sigma}_0^{(4)})$ is given, both $(\overline{\sigma}_1^{(2)}$,\,$\overline{\sigma}_0^{(2)})$ and 
$(\overline{\sigma}_1^{(3)}$,\,$\overline{\sigma}_0^{(3)})$ provide additional identifying information about   $\overline{\mu}^{(3)}$.

{\bf Bounding of skewness.}
To obtain bounds on the skewness of ICE, i.e., $\overline{\mu}^{(3)} / (\overline{\mu}^{(2)})^{3/2}$, we first compute bounds on $\overline{\mu}^{(k)}$ $(k=2,3)$.
Then, we compute the upper bound for the skewness as the ratio of the upper bound of $\overline{\mu}^{(3)}$ to the lower bound of $\overline{\mu}^{(2)}$ raised to the power $3/2$, and the lower bound is the negative of this value. 
Note that the obtained bounds may not be sharp. 


\subsection{Bounds of $\overline{\mu}^{(4)}$ 
}
\label{subsec:Bounding_of_mu4}


We consider the bounding of $\overline{\mu}^{(4)}$ using $(\overline{\sigma}_1^{(2)}$,\,$\overline{\sigma}_0^{(2)})$,
$(\overline{\sigma}_1^{(3)}$,\,$\overline{\sigma}_0^{(3)})$, or
$(\overline{\sigma}_1^{(4)}$,\,$\overline{\sigma}_0^{(4)})$.

($\overline{\sigma}_1^{(2)}$, $\overline{\sigma}_0^{(2)}$) are  helpful in getting a lower bound for $\overline{\mu}^{(4)}$: 
\begin{restatable}{theorem}{FourthfromS}
\label{thm:mu4_from_2}
Assume  $(\overline{\sigma}_1^{(2)},\overline{\sigma}_0^{(2)})$ exist and are available.

(1). If $\overline{\sigma}_1^{(2)}=\overline{\sigma}_0^{(2)}=0$, we have $\overline{\mu}^{(4)}=0$.

(2). If $\overline{\sigma}_1^{(2)}>0$ or $\overline{\sigma}_0^{(2)}>0$, we have \begin{equation}
\Big(\sqrt{\overline{\sigma}_1^{(2)}}-\sqrt{\overline{\sigma}_0^{(2)}}\Big)^4
\leq \overline{\mu}^{(4)} \leq \infty.
\end{equation}
If only $(\overline{\sigma}_1^{(2)},\overline{\sigma}_0^{(2)})$ are available, then this bound is sharp.

\end{restatable}

On the other hand, ($\overline{\sigma}_1^{(3)}$, $\overline{\sigma}_0^{(3)})$ are not helpful in obtaining a nontrivial bound for $\overline{\mu}^{(4)}$:
\begin{restatable}{theorem}{FourthfromT}
\label{thm:mu4_from_3}
Assume that 
$(\overline{\sigma}^{(3)}_1,\overline{\sigma}^{(3)}_0)$ exist and are available.
Then we have
$0\leq \overline{\mu}^{(4)} \leq \infty$.
If only $(\overline{\sigma}^{(3)}_1,\overline{\sigma}^{(3)}_0)$ are available, then 
this bound is sharp.
\end{restatable}

When ($\overline{\sigma}_1^{(4)}$, $\overline{\sigma}_0^{(4)}$) are available, we can bound $\overline{\mu}^{(4)}$ as: 

\begin{restatable}{theorem}{FourthfromF}
\label{thm:mu4_from_4}
Assume that $(\overline{\sigma}^{(4)}_1,\overline{\sigma}^{(4)}_0)$ exist and are available.
Then we have 
\begin{equation}
\label{eq:mu4_from_4}
\Big(\sqrt[4]{\overline{\sigma}^{(4)}_1}-\sqrt[4]{\overline{\sigma}^{(4)}_0}\Big)^4
\leq
\overline{\mu}^{(4)}
\leq
\Big(\sqrt[4]{\overline{\sigma}^{(4)}_1}+\sqrt[4]{\overline{\sigma}^{(4)}_0}\Big)^4.
\end{equation}
If only $(\overline{\sigma}^{(4)}_1,\overline{\sigma}^{(4)}_0)$ are available, then 
this bound is sharp.
\end{restatable}

If $(\overline{\sigma}_1^{(k)},\overline{\sigma}_0^{(k)})$ ($k=2,3,4$) are jointly available,
by taking the intersection of the bounds in Theorems~\ref{thm:mu4_from_2}-\ref{thm:mu4_from_4},
we have: 


\begin{restatable}{theorem}{FourthfromComb}
\label{cor5}
(1). Assume $(\overline{\sigma}_1^{(2)},\overline{\sigma}_0^{(2)})$
and $(\overline{\sigma}_1^{(3)},\overline{\sigma}_0^{(3)})$ exist and are available.

(i). If $\overline{\sigma}_1^{(2)}=\overline{\sigma}_0^{(2)}=0$, then
$\overline{\mu}^{(4)}=0$.

(ii). If $\overline{\sigma}_1^{(2)}>0$ and $\overline{\sigma}_0^{(2)}>0$, then
\begin{equation}
\left(\sqrt{\overline{\sigma}_1^{(2)}}-\sqrt{\overline{\sigma}_0^{(2)}}\right)^4
\leq \overline{\mu}^{(4)} \leq \infty .
\end{equation}

(iii). If $\overline{\sigma}_1^{(2)}>0$ and $\overline{\sigma}_0^{(2)}=0$, then $\overline{\sigma}_0^{(3)}=0$, and
\begin{equation}
(\overline{\sigma}_1^{(2)})^2
+
\frac{(\overline{\sigma}_1^{(3)})^2}{\overline{\sigma}_1^{(2)}}
\leq
\overline{\mu}^{(4)}
\leq
\infty .
\end{equation}

(iv). If $\overline{\sigma}_1^{(2)}=0$ and $\overline{\sigma}_0^{(2)}>0$, then $\overline{\sigma}_1^{(3)}=0$, and
\begin{equation}
(\overline{\sigma}_0^{(2)})^2
+
\frac{(\overline{\sigma}_0^{(3)})^2}{\overline{\sigma}_0^{(2)}}
\leq
\overline{\mu}^{(4)}
\leq
\infty .
\end{equation}

If only $(\overline{\sigma}_1^{(2)}$,\,$\overline{\sigma}_0^{(2)})$ and 
$(\overline{\sigma}_1^{(3)}$,\,$\overline{\sigma}_0^{(3)})$ are available, then 
these bounds are sharp.

(2). Assume $(\overline{\sigma}_1^{(2)}$,\,$\overline{\sigma}_0^{(2)})$ and $(\overline{\sigma}_1^{(4)}$,\,$\overline{\sigma}_0^{(4)})$ exist and are available.
Then we have
\begin{align}
\label{eq:mu4_from_24_intersection}
&\max\left\{
\Big(\sqrt{\overline{\sigma}_1^{(2)}}-\sqrt{\overline{\sigma}_0^{(2)}}\Big)^4,   \Big(\sqrt[4]{\overline{\sigma}_1^{(4)}}-\sqrt[4]{\overline{\sigma}_0^{(4)}}\Big)^4 
\right\} \nonumber\\
& \hspace{2cm}\leq \overline{\mu}^{(4)} \leq \Big(\sqrt[4]{\overline{\sigma}_1^{(4)}}+\sqrt[4]{\overline{\sigma}_0^{(4)}}\Big)^4.  
\end{align}

(3). Assume $(\overline{\sigma}_1^{(3)}$,\,$\overline{\sigma}_0^{(3)})$ and $(\overline{\sigma}_1^{(4)}$,\,$\overline{\sigma}_0^{(4)})$ exist and are available.
Then we have
\begin{equation}
\label{eq:mu4_from_34_intersection}
    \Big(\sqrt[4]{\overline{\sigma}^{(4)}_1}-\sqrt[4]{\overline{\sigma}^{(4)}_0}\Big)^4 
    \leq
    \overline{\mu}^{(4)}
    \leq
    \Big(\sqrt[4]{\overline{\sigma}^{(4)}_1}+\sqrt[4]{\overline{\sigma}^{(4)}_0}\Big)^4.    
\end{equation}

(4). Assume $(\overline{\sigma}_1^{(2)}$,\,$\overline{\sigma}_0^{(2)})$,
$(\overline{\sigma}_1^{(3)}$,\,$\overline{\sigma}_0^{(3)})$,
and $(\overline{\sigma}_1^{(4)}$,\,$\overline{\sigma}_0^{(4)})$ exist and are available.
Then we have
\begin{align}
\label{eq:mu4_from_234_intersection}
&\max\left\{
\Big(\sqrt{\overline{\sigma}_1^{(2)}}-\sqrt{\overline{\sigma}_0^{(2)}}\Big)^4,   \Big(\sqrt[4]{\overline{\sigma}_1^{(4)}}-\sqrt[4]{\overline{\sigma}_0^{(4)}}\Big)^4 
\right\}  \nonumber\\
& \hspace{2cm}\leq \overline{\mu}^{(4)} \leq \Big(\sqrt[4]{\overline{\sigma}_1^{(4)}}+\sqrt[4]{\overline{\sigma}_0^{(4)}}\Big)^4.  
\end{align}
\end{restatable}

The bounds in (2)-(4) of Theorem~\ref{cor5} are not sharp: we provide counterexamples in Appendix \ref{sec_app:intersection_non-sharp}.
On the other hand, the bounds in (1) of Theorem~\ref{cor5} are sharp. 
{Therefore, when $\overline{\sigma}_1^{(2)}\overline{\sigma}_0^{(2)}>0$, additionally knowing $(\overline{\sigma}_1^{(3)},\overline{\sigma}_0^{(3)})$ does not further tighten the bound for $\overline{\mu}^{(4)}$ obtained from $(\overline{\sigma}_1^{(2)},\overline{\sigma}_0^{(2)})$ alone.}
However, given  $(\overline{\sigma}_1^{(4)},\overline{\sigma}_0^{(4)})$, both $(\overline{\sigma}_1^{(2)}$,\,$\overline{\sigma}_0^{(2)})$ and 
$(\overline{\sigma}_1^{(3)}$,\,$\overline{\sigma}_0^{(3)})$ provide additional identifying information about   $\overline{\mu}^{(4)}$.

{\bf Bounding of kurtosis.}
To obtain bounds on the kurtosis of ICE, $\overline{\mu}^{(4)} / (\overline{\mu}^{(2)})^{2}$, we first compute bounds on $\overline{\mu}^{(k)}$ $(k=2,4)$. 
Then, we compute the upper (resp. lower) bound for the kurtosis as given by the ratio of the upper (resp. lower) bound of $\overline{\mu}^{(4)}$ to the square of the lower (resp. upper) bound of $\overline{\mu}^{(2)}$.
The bounds may not be sharp. 

{\bf Remark.}
If we have access to variance, skewness and kurtosis, then we can compute 3rd and 4th central moments. In the settings where we have only marginal skewness and kurtosis of POs but not marginal variance, we are not able to derive nontrivial bounds for $\overline{\mu}^{(2)}$, $\overline{\mu}^{(3)}$, or $\overline{\mu}^{(4)}$, unless we have range information about  $\overline{\sigma}_1^{(2)}$ and $\overline{\sigma}_0^{(2)}$. A discussion is given in  Appendix \ref{sec:appendix_skew_kurt_only}.

\section{Bounding higher central moments of ICE}

In this section, we study the bounding of higher orders of the central moments $\overline{\mu}^{(m)}$ of the ICE 
for arbitrary $m$.

\subsection{Bounds of even-order 
$\overline{\mu}^{(m)}$
}
\label{subsec:Bounding_of_even-order_moments}

As a generalization of Theorems \ref{thm:bound_of_var} and \ref{thm:mu4_from_4}, 
we have bounds of $\overline{\mu}^{(m)}$ for arbitrary even $m$ if $(\overline{\sigma}_1^{(m)},\overline{\sigma}_0^{(m)})$ are available.

\begin{restatable}{theorem}{Minkowskibound}
\label{thm:minkowski_bound}
Let $m\geq 2$ be an even integer.
Assume $(\overline{\sigma}^{(m)}_1,\overline{\sigma}^{(m)}_0)$ exist and are available.
Then 
we have
\begin{equation}
\label{eq:minkowski_bound}
\begin{aligned}
&\Big(\sqrt[m]{\overline{\sigma}^{(m)}_1}-\sqrt[m]{\overline{\sigma}^{(m)}_0}\Big)^m
\leq
\overline{\mu}^{(m)}\\
&\hspace{2cm}\leq
\Big(\sqrt[m]{\overline{\sigma}^{(m)}_1}+\sqrt[m]{\overline{\sigma}^{(m)}_0}\Big)^m.
\end{aligned}
\end{equation}
If only $(\overline{\sigma}^{(m)}_1,\overline{\sigma}^{(m)}_0)$ are available, then this bound is sharp.

\end{restatable}

As a generalization of Theorems \ref{thm:mu4_from_2} and~\ref{thm:mu4_from_3},
when $(\overline{\sigma}_1^{(k)},\overline{\sigma}_0^{(k)})$ for some $k<m$ are available, we  have: 
\begin{restatable}{theorem}{EvenfromSingle}
\label{thm:even_m_single_k}
Let $m\geq 2$ be an even integer. Assume that $(\overline{\sigma}_1^{(k)},\overline{\sigma}_0^{(k)})$ exist and are available for a single $k< m$.

(1). If $k$ is odd, we have $0\leq \overline{\mu}^{(m)} \leq \infty$.

(2). If $k$ is even, we have the following:

(i). If $\overline{\sigma}_1^{(k)}=\overline{\sigma}_0^{(k)}=0$, we have $\overline{\mu}^{(m)}=0$.

(ii). If $\overline{\sigma}_1^{(k)}>0$ or $\overline{\sigma}_0^{(k)}>0$, we have
\begin{equation}
\Big(\sqrt[k]{\overline{\sigma}_1^{(k)}}-\sqrt[k]{\overline{\sigma}_0^{(k)}}\Big)^m \leq \overline{\mu}^{(m)}\leq \infty .
\end{equation}
If only $(\overline{\sigma}_1^{(k)},\overline{\sigma}_0^{(k)})$ are available, then these bounds are sharp.

\end{restatable}

\subsection{Bounds of odd-order 
$\overline{\mu}^{(m)}$ 
}
\label{subsec:Bounding_of_odd-order_moments}

As a generalization of Theorems \ref{thm:mu3_from_2} and \ref{thm:mu3_from_3}, we have: 
\begin{restatable}{theorem}{OddUnboundedSingle}
\label{thm:odd_unbounded_from_single}
Let $m\geq 3$ be an odd number. 
Assume that $(\overline{\sigma}_1^{(k)},\overline{\sigma}_0^{(k)})$ exist and are available for a single $k\leq m$.

(1). If $k$ is even and $\overline{\sigma}_1^{(k)}=\overline{\sigma}_0^{(k)}=0$, we have $\overline{\mu}^{(m)}=0$.

(2). Otherwise, we have $-\infty \leq \overline{\mu}^{(m)} \leq \infty$.
If only $(\overline{\sigma}_1^{(k)},\overline{\sigma}_0^{(k)})$ are available, then this bound is sharp. 
\end{restatable}

Theorem~\ref{thm:odd_unbounded_from_single} means that, when $m$ is odd, knowing $(\overline{\sigma}_1^{(k)},\overline{\sigma}_0^{(k)})$ for a single $k\leq m$ does not provide a nontrivial bound for $\overline{\mu}^{(m)}$ (except in extreme cases).  In particular,  unlike the cases where $m$ is even, we cannot obtain a nontrivial bound for $\overline{\mu}^{(m)}$ even if $(\overline{\sigma}_1^{(m)},\overline{\sigma}_0^{(m)})$ are available. 

\subsection{Bounds of $\overline{\mu}^{(m)}$ given multiple $(\overline{\sigma}_1^{(k)},\overline{\sigma}_0^{(k)})$}
\label{subsec:combining_several}

We obtain bounds for $\overline{\mu}^{(m)}$ when multiple $(\overline{\sigma}_1^{(k)},\overline{\sigma}_0^{(k)})$ are jointly available by 
intersecting bounds in Theorems~\ref{thm:minkowski_bound}--\ref{thm:odd_unbounded_from_single}. 

\begin{restatable}{corollary}{BoundingMum}\label{cor-mum}
Assume  $(\overline{\sigma}_1^{(k)},\overline{\sigma}_0^{(k)})$ exist and are available for all $k\leq m$, and  
$\overline{\sigma}_1^{(k)}+\overline{\sigma}_0^{(k)}>0$ for all even $k$. We have: 

(1). If $m$ is even, we have
\begin{equation}
\begin{aligned}
&\max_{k \leq m,\,k \,\text{ is even}} \Big(\sqrt[k]{\overline{\sigma}_1^{(k)}}-\sqrt[k]{\overline{\sigma}_0^{(k)}}\Big)^m \\
&\qquad \qquad\leq 
\overline{\mu}^{(m)} \leq \Big(\sqrt[m]{\overline{\sigma}_1^{(m)}}+\sqrt[m]{\overline{\sigma}_0^{(m)}}\Big)^m .
\end{aligned}
\end{equation}

(2). If $m$ is odd, we have $-\infty \leq \overline{\mu}^{(m)} \leq \infty$.


\end{restatable}

These bounds are obtained by simply intersecting the individual  bounds and are not sharp.


\paragraph{Remark.}
We do not investigate the general setting in which  $(\overline{\sigma}_1^{(k)},\overline{\sigma}_0^{(k)})$ are available for $k> m$, of which Theorem \ref{thm:mu3_from_4} is  a special instance. The setting is very challenging.





\section{Case Studies}


In this section, we analyze the two empirical case studies described in Section~\ref{sec2}.
To support our findings in finite-sample settings, we also  conducted simulated numerical experiments, the results of which are reported in Appendix~\ref{sec_app:numerical_experiments}.

\subsection{Case Study in [Moseley et al. 2002]}


{\bf Summary of the study.}
We analyze the study by \citet{Moseley2002} introduced in Section \ref{sec2}.  
The mean and standard deviation (SD) of the PFS score ($Y$) for the placebo surgery group ($X=0$) and the arthroscopic d\'ebridement group ($X=1$)
are shown in Table \ref{tab:moseley2002-pfs}. ACE is reported as $4.9$, which means arthroscopic débridement increases PFS.

{\bf Results.}
Figure~\ref{fig1} presents the variance of the ICE $\overline{\mu}^{(2)}$ as a function of $\lambda$, as given in Theorem~\ref{thm:sensitivity_analysis_of_var} case (1). 
$\overline{\mu}^{(2)}$ is monotonically decreasing in $\lambda$.
This implies that $\overline{\mu}^{(2)}$ decreases as the correlation between the PO under placebo surgery and the ICE  increases.
The lower bound of $\overline{\mu}^{(2)}$ is $19.360$, and the upper bound  is $806.560$.
If researchers are willing to assume that $\lambda \geq 0$ (i.e., a positive correlation between $Y_0$ and $Y_1-Y_0$), we obtain a tighter bound from $19.360$ to $124.960$.
The estimated variance of the ICE under the IED condition ($\lambda=0$) is $124.960$.
 We also obtain a lower bound of $5.762\times 10^{-4}$ for the kurtosis of the ICE. 
Note that we cannot compute confidence intervals of the estimates from the reported point estimates in \citep{Moseley2002}.

Our results indicate relatively large heterogeneity in the ICE (the variance is substantial relative to the ACE). 
These findings further highlight the risk of arthroscopic débridement, as there are a substantial number of patients whose ICE differs markedly from the ACE,  including potentially adverse effects.



\begin{figure}[tb]
\centering
\input{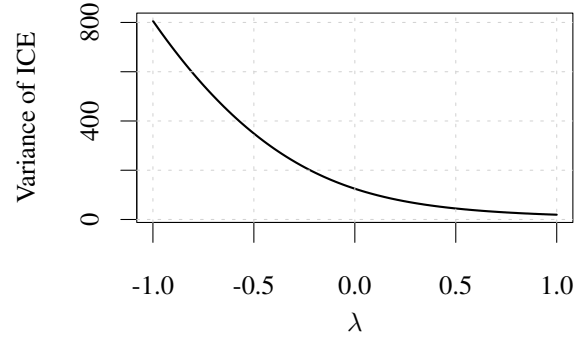}
\caption{The variance of the ICE as a function of $\lambda$ for the case study in \citep{Moseley2002}. 
}
\label{fig1}
\end{figure}

\subsection{Case Study in [Špirtovic et al. 2023]}

{\bf Summary of the study.}
We analyze the study by \citep{Spirtovic2023} introduced in Section \ref{sec2}. The  mean, SD, skewness, and kurtosis of MM\% ($Y$) in the control group ($X=0$) and the experimental group ($X=1$) are shown in Table \ref{tab:spirtovic_mm_moments}.
ACE is reported as $-1.344$.

{\bf Results.}
Table \ref{tab:app2} reports the estimated bounds for the variance, skewness, and kurtosis of the ICE, 
together with their estimates under the IED condition.
The bounds for the variance, skewness, and kurtosis are quite wide and therefore less informative.

For this study, the IED condition states that the treatment-induced 
{change} in the outcome is independent of the outcome under the control treatment. This assumption is plausible if the mechanisms generating the treatment effect are distinct from those determining the outcome under control. 
Under the IED condition, the distribution of the ICE exhibits large dispersion relative to the mean, is positively skewed, and is less kurtotic than a Gaussian distribution.
Thus, a substantial number of participants have ICEs that differ markedly from the ACE, a smaller subset of participants may benefit substantially more than the ACE, and extreme outliers are rare.


\begin{table}[tb]
\centering
\caption{Estimated bounds and IED-based estimates of the variance, skewness, and kurtosis of the ICE for the case study in \citep{Spirtovic2023}.}
\label{tab:app2}
\scalebox{1}{
\begin{tabular}{c|cc}
    \hline
   Target  &  & Estimates  \\
   \hline \hline
    \multirow{3}{*}{Variance} & Under IED &  $6.776$  \\
     & UB & $96.727$  \\
     & LB & $0.475$ \\
     \hline
     \multirow{3}{*}{Skewness} &Under IED &  $4.045$  \\
     &  UB & $4563.173$  \\
     & LB & $-4563.173$ \\
     \hline
    \multirow{3}{*}{Kurtosis} &Under IED &  $1.348$  \\
     &  UB & $1.430\times10^5$  \\
     &LB & $3.490\times10^{-5}$ \\
    \hline
\end{tabular}
}
\end{table}


\section{Conclusion}


This paper studies the identification and bounding of the central moments of the ICE using only the marginal central moments of each PO. We establish new identification results  and derive sharp bounds when identification is not attainable.
By relying solely on marginal central moment information, our framework 
broadens the applicability of moment-based analyses of causal effects. In particular, our results enable researchers to conduct heterogeneity analysis 
using historically accumulated studies that report only summary statistics rather than full datasets.

Our framework can be extended to study stratum-specific moments of the ICE if marginal moments are available within strata of observed covariates $C$: our identification and bounding arguments can be applied within each stratum defined by $C=c$. These stratum-specific moments of the ICE describe the heterogeneity within a subpopulation defined by $C=c$, while standard conditional average causal effects (CACE), $\mathbb{E}[Y_1-Y_0|C=c]$, represent the average effects in the subpopulation $C=c$.

The identification results rely on a strong, untestable IED assumption, whose appropriateness cannot be verified directly from observed data, but must instead be assessed based on domain knowledge and the scientific context of the study. In practice, researchers must evaluate whether such an assumption is plausible or acceptable for their application, and thereby determine whether the resulting identification is credible.
The bounds without IED can be wide, especially for higher-order moments. The width of sharp bounds reflects fundamental non-identifiability given the limited information available. A contribution of our work is to explicitly characterize what can and cannot be learned under moment-only information. In this sense, wide bounds are still informative: they quantify the extent to which higher-order features of treatment effect heterogeneity are inherently underdetermined given only moment information without stronger assumptions. Still, information on variance can often be learned and yield meaningful insights. In addition, our bounds from moment-only information can, in some cases, be tighter than those obtained by existing methods that assume access to full marginal distributions. This is illustrated in the numerical experiments in Appendix~\ref{sec_app:numerical_experiments}: in Table~\ref{tab:numericalexp_2}, the lower bounds for  $\overline{\mu}^{(2)}$ and  $\overline{\mu}^{(4)}$ given $(\overline{\sigma}_1^{(2)}, \overline{\sigma}_0^{(2)})$ are substantially sharper than the corresponding bounds by the method in \citep{Kawakami2025_moments}.

A limitation of this study is that we do not establish sharp bounds for $\overline{\mu}^{(m)}$ when multiple  $(\overline{\sigma}_1^{(k)}, \overline{\sigma}_0^{(k)})$ pairs are jointly available. 
Moreover, the bounds derived for the skewness and kurtosis of the ICE are not sharp. 
Deriving sharp bounds in these settings appears to be substantially more challenging and remains a  future research direction.


The bounds derived in this paper may be further tightened under additional structural assumptions such as a causal graph \citep{cozman2000credal, zaffalon2020structural,Zhang2022_Partial,shridharan2023causal,shridharan2023scalable,zaffalon2024efficient,Duarte2024,Arroyo2025,Maiti2025} or monotonicity assumptions on POs \citep{Manski1997,Pearl1999,Tian2000,Mueller2023,Mueller2025,Zhang2025,Hashimoto2026}. 
Assessing how these additional assumptions can sharpen the bounds is an important direction for future research.

\begin{acknowledgements} 
The authors thank the anonymous reviewers for their time and thoughtful comments.
\end{acknowledgements}

\bibliography{uai2026-template}

\newpage
\appendix
\onecolumn


\title{Identification and Bounding of Central Moments of Causal Effects Using Marginal Moments Information (Supplementary Material)}
\maketitle

\section{Counterexamples}

\subsection{
Counterexample mentioned in Section \ref{subsec:Identification_under_IED}
}
\label{sec_app:counterexample_assump2}

{\bf Counterexample for identifiability using $(\overline{\sigma}_1^{(2)},\overline{\sigma}_0^{(2)})$ under Assumption~2 in \citep{Kawakami2025_moments}.} 
We provide two SCMs that satisfy Assumption~2 in \citep{Kawakami2025_moments} and have the same 
$(\overline{\sigma}_1^{(2)},\overline{\sigma}_0^{(2)})$, while $\overline{\mu}^{(2)}$ is different.

Let $U\sim \mathrm{Unif}[-1,1]$ and consider the following two SCMs:
\begin{equation}
\label{eq:scm_assump2_counterexample}
\begin{aligned}
&\text{(A)}: Y\coloneqq U,\\
&\text{(B)}: Y\coloneqq (1-X)U+cXU^3,\qquad c\coloneqq \sqrt{7/3}.
\end{aligned}
\end{equation}

Assumption~2 in \citep{Kawakami2025_moments} requires that, in an SCM with $Y= f_Y(X,U)$, the function $u\mapsto f_Y(x,u)$ is monotonic in $u$
for every $x\in\{0,1\}$ almost surely with respect to $\mathbb{P}_U$.
We remark that both SCM (A) and (B) satisfy this assumption.

In both SCMs we have $\overline{\sigma}_0^{(2)}=\Var(U)=1/3$ and
$\overline{\sigma}_1^{(2)}=\Var(cU^3)=c^2\mathbb{E}[U^6]=(7/3)\cdot(1/7)=1/3$.
However, $\overline{\mu}^{(2)}=\Var(Y_1-Y_0)$ equals $0$ in SCM (A), while in SCM (B),
\begin{equation}
\begin{aligned}
\overline{\mu}^{(2)}
&=\mathbb{E}[(cU^3-U)^2]
=c^2\mathbb{E}[U^6]-2c\mathbb{E}[U^4]+\mathbb{E}[U^2]\\
&=\frac{2}{3}-\frac{2}{5}\sqrt{\frac{7}{3}}.
\end{aligned}
\end{equation}
Therefore, under Assumption~2, identifying $\overline{\mu}^{(2)}$ generally requires more information than
$(\overline{\sigma}_1^{(2)},\overline{\sigma}_0^{(2)})$, e.g., the marginal distributions of each PO.

\subsection{Counterexamples Supporting the Sharpness Discussion 
in Section \ref{sec6}}
\label{sec_app:intersection_non-sharp}

Throughout this section, we put $\widetilde{Y}_x\coloneqq Y_x-\mathbb{E}[Y_x]$ ($x\in\{0,1\}$) and $\widetilde{\Delta}\coloneqq (Y_1-Y_0)-\mathbb{E}[Y_1-Y_0]=\widetilde{Y}_1-\widetilde{Y}_0$.
Then $\overline{\sigma}_x^{(m)}=\mathbb{E}[\widetilde{Y}_x^m]$ and $\overline{\mu}^{(m)}=\mathbb{E}[\widetilde{\Delta}^m]$.

Note that, to show that a proposed closed bound $[L,U]$ is \emph{not} sharp, we need to obtain an example of moment information such that $L'\leq \overline{\mu}^{(m)}\leq U'$ under that condition, where $L'>L$ or $U'<U$.
We remark that to show the bounds themselves are unattainable is \emph{not} sufficient to conclude that they are not sharp. Indeed, even if a bound is not attained, there remains a possibility that arbitrarily close value to the bound is attainable, whence the bound is still sharp. 

In this section, for each non-sharp bound in the main part, we provide an example with $L'\leq \overline{\mu}^{(m)}\leq U'$, where $L'>L$ and/or $U'<U$.

We use the following lemma:

\begin{lemma}
\label{lem:two_point_extremal_34}
(1) Let $Z$ be a random variable with $\mathbb{E}[Z]=0$, $\mathbb{E}[Z^3]=\sqrt{2}$, and $\mathbb{E}[Z^4]=3$.
Then we have $\mathbb{E}[Z^2]=1$, and the distribution of $Z$ is given by
\begin{equation}
\label{eq:zpzm_prob_1}
\mathbb{P}\left(Z=\frac{\sqrt{2}+\sqrt{6}}{2}\right) = \frac{1}{2}-\frac{\sqrt{3}}{6},\qquad
\mathbb{P}\left(Z=\frac{\sqrt{2}-\sqrt{6}}{2}\right) = \frac{1}{2}+\frac{\sqrt{3}}{6}.
\end{equation}

(2) Let $Z$ be a random variable with $\mathbb{E}[Z]=0$, $\mathbb{E}[Z^3]=-\sqrt{2}$, and $\mathbb{E}[Z^4]=3$.
Then we have $\mathbb{E}[Z^2]=1$, and the distribution of $Z$ is given by
\begin{equation}
\label{eq:zpzm_prob_2}
\mathbb{P}\left(Z=-\frac{\sqrt{2}+\sqrt{6}}{2}\right) = \frac{1}{2}-\frac{\sqrt{3}}{6},\qquad
\mathbb{P}\left(Z=-\frac{\sqrt{2}-\sqrt{6}}{2}\right) = \frac{1}{2}+\frac{\sqrt{3}}{6}.
\end{equation}
\end{lemma}

\begin{proof}
(1) Write $v\coloneqq \mathbb{E}[Z^2]$.
Then $\mathbb{E}[Z^3]=\mathbb{E}[Z(Z^2-v)]$.
By the Cauchy-Schwarz inequality, we have 
\begin{equation}\label{eq:cs_key}
2=\mathbb{E}[Z^3]^2 \leq \mathbb{E}[Z^2]\mathbb{E}[(Z^2-v)^2] = v(\mathbb{E}[Z^4]-v^2) = v(3-v^2).
\end{equation}
For $v\geq 0$, $v(3-v^2)-2=-(v-1)^2(v+2)\geq 0$ holds only at $v=1$.
Therefore \eqref{eq:cs_key} forces $v=1$ and equality in the Cauchy-Schwarz inequality.
Hence $Z$ and $Z^2-1$ are linearly dependent almost surely, and $\mathbb{E}[Z(Z^2-1)]=\mathbb{E}[Z^3]=\sqrt{2}$ implies $Z^2-1=\sqrt{2}Z$ almost surely.
The roots of $z^2-\sqrt{2}z-1=0$ are $\dfrac{\sqrt{2}\pm\sqrt{6}}{2}$, and $Z$ takes these values almost surely.
The probabilities in \eqref{eq:zpzm_prob_1} follow from $\mathbb{E}[Z]=0$.
Conversely, assume that \eqref{eq:zpzm_prob_1} holds. Then $\mathbb{E}[Z]=0$.
Moreover, since $Z$ takes values in the set of roots of $z^2-\sqrt{2}z-1=0$, we have $Z^2=\sqrt{2}Z+1$ almost surely, and hence
\begin{equation*}
Z^3=ZZ^2=\sqrt{2}Z^2+Z=3Z+\sqrt{2},\qquad
Z^4=(Z^2)^2=(\sqrt{2}Z+1)^2=4\sqrt{2}Z+3
\end{equation*}
almost surely. Then we obtain $\mathbb{E}[Z^3]=\sqrt{2}$ and $\mathbb{E}[Z^4]=3$.

(2) Write $v\coloneqq \mathbb{E}[Z^2]$. Eq.~\eqref{eq:cs_key} holds, and it forces $v=1$ and equality in the Cauchy-Schwarz inequality.
Hence $Z$ and $Z^2-1$ are linearly dependent almost surely, and $\mathbb{E}[Z(Z^2-1)]=\mathbb{E}[Z^3]=-\sqrt{2}$ implies $Z^2-1=-\sqrt{2}Z$ almost surely.
The roots of $z^2+\sqrt{2}z-1=0$ are $\dfrac{-\sqrt{2}\pm\sqrt{6}}{2}$, and $Z$ takes these values almost surely.
The probabilities in \eqref{eq:zpzm_prob_2} follow from $\mathbb{E}[Z]=0$.
Conversely, assume that \eqref{eq:zpzm_prob_2} holds. Then $\mathbb{E}[Z]=0$.
Moreover, since $Z$ takes values in the set of roots of $z^2+\sqrt{2}z-1=0$, we have $Z^2=-\sqrt{2}Z+1$ almost surely, and hence
\begin{equation}
Z^3=ZZ^2=-\sqrt{2}Z^2+Z=3Z-\sqrt{2},\qquad
Z^4=(Z^2)^2=(-\sqrt{2}Z+1)^2=-4\sqrt{2}Z+3
\end{equation}
almost surely. Then we have $\mathbb{E}[Z^3]=-\sqrt{2}$ and $\mathbb{E}[Z^4]=3$. 
\end{proof}


{\bf Counterexample for sharpness of the bound \eqref{eq:mu2_from_34_intersection}.} 
Let
$\overline{\sigma}_1^{(3)}=\sqrt{2}$, $\overline{\sigma}_0^{(3)}=-\sqrt{2}$, $\overline{\sigma}_1^{(4)}=\overline{\sigma}_0^{(4)}=3$.
Let $(\widetilde{Y}_1,\widetilde{Y}_0)$ be a distribution compatible with it.
By Lemma \ref{lem:two_point_extremal_34}, $\widetilde{Y}_1$ and  $\widetilde{Y}_0$ have a distribution in \eqref{eq:zpzm_prob_1} and \eqref{eq:zpzm_prob_2}, respectively.
Let $x\coloneqq \mathbb{P}\left(\widetilde{Y}_1=\frac{\sqrt{2}+\sqrt{6}}{2},\,\widetilde{Y}_0=-\frac{\sqrt{2}-\sqrt{6}}{2}\right)$.
Then $0\leq x\leq \frac{1}{2}-\frac{\sqrt{3}}{6}$, and
\begin{equation}\label{eq:mu2_linear_x}
\begin{aligned}
\mathbb{E}[\widetilde{Y}_1\widetilde{Y}_0]
&=-\Big(\frac{1}{2}-\frac{\sqrt{3}}{6}-x\Big)\Big(\frac{\sqrt{2}+\sqrt{6}}{2}\Big)^2
-2x\frac{\sqrt{2}+\sqrt{6}}{2}\frac{\sqrt{2}-\sqrt{6}}{2}
-\Big(\frac{1}{2}+\frac{\sqrt{3}}{6}-x\Big)\Big(\frac{\sqrt{2}-\sqrt{6}}{2}\Big)^2\\
&=-\Big(\frac{1}{2}-\frac{\sqrt{3}}{6}-x\Big)(2+\sqrt{3})
+2x
-\Big(\frac{1}{2}+\frac{\sqrt{3}}{6}-x\Big)(2-\sqrt{3})
=-1+6x,\\
\overline{\mu}^{(2)}
&=\mathbb{E}[(\widetilde{Y}_1-\widetilde{Y}_0)^2]
=\mathbb{E}[\widetilde{Y}_1^2]-2\mathbb{E}[\widetilde{Y}_1\widetilde{Y}_0]+\mathbb{E}[\widetilde{Y}_0^2]
=2-2\mathbb{E}[\widetilde{Y}_1\widetilde{Y}_0]=2-2(-1+6x)=4-12x.
\end{aligned}
\end{equation}

Therefore
$2\sqrt{3}-2 \leq \overline{\mu}^{(2)} \leq 4$.
On the other hand, \eqref{eq:mu2_from_34_intersection} gives a bound $0\leq \overline{\mu}^{(2)}\leq 4\sqrt{3}$.
Hence \eqref{eq:mu2_from_34_intersection} is not sharp.

{\bf Counterexample for sharpness of the bound \eqref{eq:bound_of_var_with_234}.} 

\underline{\emph{Lower bound.}}
Let $\overline{\sigma}_1^{(2)}=\overline{\sigma}_0^{(2)}=1$,
$\overline{\sigma}_1^{(3)}=\sqrt{2}$, $\overline{\sigma}_0^{(3)}=-\sqrt{2}$
and $\overline{\sigma}_1^{(4)}=\overline{\sigma}_0^{(4)}=3$.
Let $(\widetilde{Y}_1,\widetilde{Y}_0)$ be a distribution compatible with it.
By Lemma \ref{lem:two_point_extremal_34}, such an $(\widetilde{Y}_1,\widetilde{Y}_0)$ exists, and $\widetilde{Y}_1$ and $\widetilde{Y}_0$ have the distribution in \eqref{eq:zpzm_prob_1} and \eqref{eq:zpzm_prob_2}, respectively.

Eq.~\eqref{eq:bound_of_var_with_234} gives a bound $0\leq \overline{\mu}^{(2)}\leq 4$.
However, we have $2\sqrt{3}-2 \leq \overline{\mu}^{(2)} \leq 4$, so the lower bound in \eqref{eq:bound_of_var_with_234} is not sharp.

\underline{\emph{Upper bound.}}
Let $\overline{\sigma}_1^{(2)}=\overline{\sigma}_0^{(2)}=1$,
$\overline{\sigma}_1^{(3)}=\overline{\sigma}_0^{(3)}=\sqrt{2}$
and $\overline{\sigma}_1^{(4)}=\overline{\sigma}_0^{(4)}=3$.
Let $(\widetilde{Y}_1,\widetilde{Y}_0)$ be a distribution compatible with it.
By Lemma \ref{lem:two_point_extremal_34}, such an $(\widetilde{Y}_1,\widetilde{Y}_0)$ exists, and both $\widetilde{Y}_1$ and $\widetilde{Y}_0$ have distribution in \eqref{eq:zpzm_prob_1}.
Eq.~\eqref{eq:bound_of_var_with_234} gives a bound $0\leq \overline{\mu}^{(2)}\leq 4$.

Let $x\coloneqq \mathbb{P}\left(\widetilde{Y}_1=\frac{\sqrt{2}+\sqrt{6}}{2},\,\widetilde{Y}_0=\frac{\sqrt{2}+\sqrt{6}}{2}\right)$.
Then $0\leq x\leq \frac{1}{2}-\frac{\sqrt{3}}{6}$, and
\begin{equation}
\begin{aligned}
\overline{\mu}^{(2)}
&=\mathbb{E}[\widetilde{\Delta}^2]
=\mathbb{E}[(\widetilde{Y}_1-\widetilde{Y}_0)^2]\\
&=\Big(\frac{\sqrt{2}+\sqrt{6}}{2}-\frac{\sqrt{2}-\sqrt{6}}{2}\Big)^2
\left\{\mathbb{P}\Big(\widetilde{Y}_1=\frac{\sqrt{2}+\sqrt{6}}{2},\widetilde{Y}_0=\frac{\sqrt{2}-\sqrt{6}}{2}\Big)
+\mathbb{P}\Big(\widetilde{Y}_1=\frac{\sqrt{2}-\sqrt{6}}{2},\widetilde{Y}_0=\frac{\sqrt{2}+\sqrt{6}}{2}\Big)\right\}\\
&=6\cdot 2\Big(\frac{1}{2}-\frac{\sqrt{3}}{6}-x\Big)
=12\Big(\frac{1}{2}-\frac{\sqrt{3}}{6}-x\Big)\\
&\leq 12\Big(\frac{1}{2}-\frac{\sqrt{3}}{6}\Big)=6-2\sqrt{3}.
\end{aligned}
\end{equation}

Since $6-2\sqrt{3}<4$, the upper bound in \eqref{eq:bound_of_var_with_234} is not sharp.

{\bf Counterexample for sharpness of the bound \eqref{eq:bound_of_mu3_with_4+23}.}

\underline{\emph{The case $(\overline{\sigma}_1^{(2)},\overline{\sigma}_0^{(2)})$ and $(\overline{\sigma}_1^{(4)},\overline{\sigma}_0^{(4)})$ are available.}}
Let $\overline{\sigma}_1^{(2)}=\overline{\sigma}_1^{(4)}=1$ and $\overline{\sigma}_0^{(2)}=\overline{\sigma}_0^{(4)}=0$.
(These are realized by, e.g., $\widetilde{Y}_0=0$ and $\widetilde{Y}_1$ with $\mathbb{P}(\widetilde{Y}_1=1)=\mathbb{P}(\widetilde{Y}_1=-1)=1/2$.)

Then $\overline{\sigma}_0^{(4)}=0$ implies $\widetilde{Y}_0=0$ almost surely.
Moreover, by the Cauchy-Schwarz inequality,
\begin{equation}
\mathbb{E}[\widetilde{Y}_1^2]^2\leq \mathbb{E}[\widetilde{Y}_1^4].
\end{equation}
Since $\mathbb{E}[\widetilde{Y}_1^2]=\mathbb{E}[\widetilde{Y}_1^4]=1$, equality holds and hence $\widetilde{Y}_1^2=1$ almost surely.
Since $\mathbb{E}[\widetilde{Y}_1]=0$, this implies $\mathbb{P}(\widetilde{Y}_1=1)=\mathbb{P}(\widetilde{Y}_1=-1)=1/2$.
Therefore $\widetilde{\Delta}=\widetilde{Y}_1$ almost surely and $\overline{\mu}^{(3)}=\mathbb{E}[\widetilde{\Delta}^3]=\mathbb{E}[\widetilde{Y}_1^3]=0$.
In contrast, the bound \eqref{eq:bound_of_mu3_with_4+23} reduces to
$-\frac{\sqrt{2}}{\sqrt[4]{27}}\leq \overline{\mu}^{(3)}\leq \frac{\sqrt{2}}{\sqrt[4]{27}}$.
In particular, neither bound in \eqref{eq:bound_of_mu3_with_4+23} is sharp.

\underline{\emph{The case $(\overline{\sigma}_1^{(3)},\overline{\sigma}_0^{(3)})$ and $(\overline{\sigma}_1^{(4)},\overline{\sigma}_0^{(4)})$ are available.}}
Let $\overline{\sigma}_1^{(3)}=0$, $\overline{\sigma}_1^{(4)}=1$, $\overline{\sigma}_0^{(3)}=0$, and $\overline{\sigma}_0^{(4)}=0$.
(These are realized by, e.g., $\widetilde{Y}_0=0$ and $\widetilde{Y}_1$ with $\mathbb{P}(\widetilde{Y}_1=1)=\mathbb{P}(\widetilde{Y}_1=-1)=1/2$.)

Then $\overline{\sigma}_0^{(4)}=0$ implies $\widetilde{Y}_0=0$ almost surely and hence $\widetilde{\Delta}=\widetilde{Y}_1$ almost surely.
Therefore $\overline{\mu}^{(3)}=\mathbb{E}[\widetilde{\Delta}^3]=\mathbb{E}[\widetilde{Y}_1^3]=\overline{\sigma}_1^{(3)}=0$.
In contrast, the bound \eqref{eq:bound_of_mu3_with_4+23} reduces to
$-\frac{\sqrt{2}}{\sqrt[4]{27}}\leq \overline{\mu}^{(3)}\leq \frac{\sqrt{2}}{\sqrt[4]{27}}$.
In particular, neither bound in \eqref{eq:bound_of_mu3_with_4+23} is sharp.

\underline{\emph{The case $(\overline{\sigma}_1^{(2)},\overline{\sigma}_0^{(2)})$, $(\overline{\sigma}_1^{(3)},\overline{\sigma}_0^{(3)})$ and $(\overline{\sigma}_1^{(4)},\overline{\sigma}_0^{(4)})$ are available.}}
Let
$\overline{\sigma}_1^{(2)}=\overline{\sigma}_1^{(4)}=1$, $\overline{\sigma}_1^{(3)}=0$, and $\overline{\sigma}_0^{(2)}=\overline{\sigma}_0^{(3)}=\overline{\sigma}_0^{(4)}=0$.
(These are realized by, e.g., $\widetilde{Y}_0=0$ and $\widetilde{Y}_1$ with $\mathbb{P}(\widetilde{Y}_1=1)=\mathbb{P}(\widetilde{Y}_1=-1)=1/2$.)

Then $\overline{\sigma}_0^{(4)}=0$ implies $\widetilde{Y}_0=0$ almost surely and hence $\widetilde{\Delta}=\widetilde{Y}_1$ almost surely.
Therefore $\overline{\mu}^{(3)}=\mathbb{E}[\widetilde{\Delta}^3]=\mathbb{E}[\widetilde{Y}_1^3]=\overline{\sigma}_1^{(3)}=0$.
In contrast, the bound \eqref{eq:bound_of_mu3_with_4+23} reduces to
$-\frac{\sqrt{2}}{\sqrt[4]{27}}\leq \overline{\mu}^{(3)}\leq \frac{\sqrt{2}}{\sqrt[4]{27}}$.
In particular, neither bound in \eqref{eq:bound_of_mu3_with_4+23} is sharp.

{\bf Counterexample for sharpness of the bound \eqref{eq:mu4_from_34_intersection}.} 

\underline{\emph{Lower bound.}}
Let $\overline{\sigma}_1^{(3)}=\sqrt{2}$, $\overline{\sigma}_0^{(3)}=-\sqrt{2}$ and $\overline{\sigma}_1^{(4)}=\overline{\sigma}_0^{(4)}=3$, and
let $(\widetilde{Y}_1,\widetilde{Y}_0)$ be a distribution compatible with it.
By Lemma \ref{lem:two_point_extremal_34}, $\widetilde{Y}_1$ and  $\widetilde{Y}_0$ have the distribution in \eqref{eq:zpzm_prob_1} and \eqref{eq:zpzm_prob_2}, respectively.
Eq.~\eqref{eq:mu4_from_34_intersection} gives $0\leq \overline{\mu}^{(4)}\leq 48$.

Let $x\coloneqq \mathbb{P}\left(\widetilde{Y}_1=\frac{\sqrt{2}+\sqrt{6}}{2},\,\widetilde{Y}_0=-\frac{\sqrt{2}-\sqrt{6}}{2}\right)$.
Then
\begin{equation}\label{eq:mu4_linear_x_lb}
\begin{aligned}
\overline{\mu}^{(4)}
&=\mathbb{E}[(\widetilde{Y}_1-\widetilde{Y}_0)^4]\\
&=(\sqrt{2}+\sqrt{6})^4\Big(\frac{1}{2}-\frac{\sqrt{3}}{6}-x\Big)
+(\sqrt{2})^4\cdot 2x
+(\sqrt{2}-\sqrt{6})^4\Big(\frac{1}{2}+\frac{\sqrt{3}}{6}-x\Big)\\
&=(112+64\sqrt{3})\Big(\frac{1}{2}-\frac{\sqrt{3}}{6}-x\Big)
+8x
+(112-64\sqrt{3})\Big(\frac{1}{2}+\frac{\sqrt{3}}{6}-x\Big)
=48-216x.
\end{aligned}
\end{equation}

Since $0\leq x\leq \frac{1}{2}-\frac{\sqrt{3}}{6}$, we have $ \overline{\mu}^{(4)} \geq 48-216(\frac{1}{2}-\frac{\sqrt{3}}{6})=36\sqrt{3}-60>0$.
Thus the lower bound in \eqref{eq:mu4_from_34_intersection} is not sharp.

\underline{\emph{Upper bound.}}
Let 
$\overline{\sigma}_1^{(3)}=\overline{\sigma}_0^{(3)}=\sqrt{2}$
and $\overline{\sigma}_1^{(4)}=\overline{\sigma}_0^{(4)}=3$, and 
let $(\widetilde{Y}_1,\widetilde{Y}_0)$ be a distribution compatible with it.
By Lemma \ref{lem:two_point_extremal_34}, such an $(\widetilde{Y}_1,\widetilde{Y}_0)$ exists, and both $\widetilde{Y}_1$ and $\widetilde{Y}_0$ have the distribution in \eqref{eq:zpzm_prob_1}.

Eq.~\eqref{eq:mu4_from_34_intersection} gives $0\leq \overline{\mu}^{(4)}\leq 48$.
With $x\coloneqq \mathbb{P}\left(\widetilde{Y}_1=\frac{\sqrt{2}+\sqrt{6}}{2},\,\widetilde{Y}_0=\frac{\sqrt{2}+\sqrt{6}}{2}\right)$, we have $0\leq x\leq \frac{1}{2}-\frac{\sqrt{3}}{6}$.
Then we have
\begin{equation}\label{eq:mu4_linear_x_ub}
\begin{aligned}
\overline{\mu}^{(4)}
&=\mathbb{E}[(\widetilde{Y}_1-\widetilde{Y}_0)^4]
=(\sqrt{6})^4\mathbb{P}(\widetilde{Y}_1\neq \widetilde{Y}_0)
=36\cdot 2\Big(\frac{1}{2}-\frac{\sqrt{3}}{6}-x\Big)
=72\Big(\frac{1}{2}-\frac{\sqrt{3}}{6}-x\Big)\\
&\leq 72\Big(\frac{1}{2}-\frac{\sqrt{3}}{6}\Big)=36-12\sqrt{3}.
\end{aligned}
\end{equation}

Since $36-12\sqrt{3}<48$, the upper bound in \eqref{eq:mu4_from_34_intersection} is not sharp.

{\bf Counterexample for sharpness of the bounds \eqref{eq:mu4_from_24_intersection} and \eqref{eq:mu4_from_234_intersection}.} 
Let
$\overline{\sigma}_1^{(2)}=1$, $\overline{\sigma}_0^{(2)}=\frac{1}{10}$, $\overline{\sigma}_1^{(4)}=\overline{\sigma}_0^{(4)}=1$.
Let $(\widetilde{Y}_1,\widetilde{Y}_0)$ be any distribution which satisfies $\overline{\sigma}_1^{(2)}=1$, $\overline{\sigma}_0^{(2)}=\frac{1}{10}$, $\overline{\sigma}_1^{(4)}=\overline{\sigma}_0^{(4)}=1$.
Such an $(\widetilde{Y}_1,\widetilde{Y}_0)$ exists, e.g.,  $(\widetilde{Y}_1,\widetilde{Y}_0)$ with $\mathbb{P}(\widetilde{Y}_1=1)=\mathbb{P}(\widetilde{Y}_1=-1)=\frac{1}{2}$ and 
\begin{equation}
\mathbb{P}(\widetilde{Y}_0=0)=\frac{99}{100}, \quad \mathbb{P}(\widetilde{Y}_0=\sqrt{10})=\mathbb{P}(\widetilde{Y}_0=-\sqrt{10})=\frac{1}{200}.
\end{equation}

Since $\overline{\sigma}_1^{(4)}=(\overline{\sigma}_1^{(2)})^2$, 
the equality in the Cauchy-Schwarz inequality $(\overline{\sigma}_1^{(2)})^2=\mathbb{E}[\widetilde{Y}_1^2]^2 \leq \mathbb{E}[(\widetilde{Y}_1^2)^2]\mathbb{E}[1^2]=\overline{\sigma}_1^{(4)}$ holds.
Therefore, we have $\widetilde{Y}_1^2=c$ almost surely for some $c\in \mathbb{R}$, and we obtain $c=1$ since $\overline{\sigma}_1^{(2)}=1$.
Thus $\widetilde{Y}_1^3=\widetilde{Y}_1$ and $\mathbb{E}[\widetilde{Y}_1^2\widetilde{Y}_0^2]=\mathbb{E}[\widetilde{Y}_0^2]=\frac{1}{10}$.

Now we have
\begin{equation}\label{eq:mu4_expand_24}
\overline{\mu}^{(4)}=\mathbb{E}[(\widetilde{Y}_1-\widetilde{Y}_0)^4]=\frac{13}{5}-4\mathbb{E}[\widetilde{Y}_1\widetilde{Y}_0(1+\widetilde{Y}_0^2)].
\end{equation}
By the Cauchy-Schwarz inequality,
\begin{equation}
\big(\mathbb{E}[\widetilde{Y}_1\widetilde{Y}_0(1+\widetilde{Y}_0^2)]\big)^2 \leq \mathbb{E}[\widetilde{Y}_1^2\widetilde{Y}_0^2]\mathbb{E}[(1+\widetilde{Y}_0^2)^2] \\
= \frac{1}{10}(1+2\mathbb{E}[\widetilde{Y}_0^2]+\mathbb{E}[\widetilde{Y}_0^4])=\frac{11}{50}
\end{equation}
holds.
Thus \eqref{eq:mu4_expand_24} implies
\begin{equation}\label{eq:mu4_interval_strict_24}
\frac{13}{5}-4\sqrt{\frac{11}{50}} \leq \overline{\mu}^{(4)} \leq \frac{13}{5}+4\sqrt{\frac{11}{50}}.
\end{equation}

On the other hand, \eqref{eq:mu4_from_24_intersection} gives a bound $(1-\frac{1}{\sqrt{10}})^4 \leq \overline{\mu}^{(4)} \leq 16$.
Since $\sqrt{11/50}<1$ and $11/50<1/4$, \eqref{eq:mu4_interval_strict_24} implies $\overline{\mu}^{(4)}<\frac{33}{5}<16$ and $\overline{\mu}^{(4)}>\frac{3}{5}$.
Moreover, since $\sqrt{10}<4$ we have $(1-\frac{1}{\sqrt{10}})^4<(3/4)^4<\frac{3}{5}$.
Thus neither bound in \eqref{eq:mu4_from_24_intersection} is sharp.

Also, \eqref{eq:mu4_from_234_intersection} coincides with \eqref{eq:mu4_from_24_intersection}, so \eqref{eq:mu4_from_234_intersection} is not sharp either.

\section{Bounding $\overline{\mu}^{(m)}$ using only marginal skewness and kurtosis}
\label{sec:appendix_skew_kurt_only}


If we have access to variance, skewness and kurtosis, then we can compute 3rd and 4th central moments. In some cases, researchers may have access only to the skewness and kurtosis, but not to the marginal central moments $(\overline{\sigma}_1^{(k)},\overline{\sigma}_0^{(k)})$ for $k = 2, 3, 4$.
In this section, we then assume that only skewness $\gamma_x$ and kurtosis $\kappa_x$, $\{(\gamma_x,\kappa_x)\}_{x\in\{0,1\}}$,  exist and are available and none of $\{\overline{\sigma}^{(k)}_x\}_{x\in\{0,1\},k\in\{2,3,4\}}$ are available.
We discuss the evaluation of the central moments $\overline{\mu}^{(m)}$ for $m = 2, 3, 4$ under this setting.
The marginal skewness $\gamma_x$ and kurtosis $\kappa_x$ are given by 
\begin{equation}
\gamma_x \coloneqq \frac{\overline{\sigma}^{(3)}_x}{(\overline{\sigma}^{(2)}_x)^{3/2}},\qquad
\kappa_x \coloneqq \frac{\overline{\sigma}^{(4)}_x}{(\overline{\sigma}^{(2)}_x)^2}.
\end{equation}

If $(\overline{\sigma}_1^{(2)},\overline{\sigma}_0^{(2)})$ are fixed, 
$\overline{\sigma}^{(3)}_x$ and $\overline{\sigma}^{(4)}_x$ are uniquely determined as
\begin{equation}
\overline{\sigma}^{(3)}_x=\gamma_x(\overline{\sigma}^{(2)}_x)^{3/2},\qquad
\overline{\sigma}^{(4)}_x=\kappa_x(\overline{\sigma}^{(2)}_x)^2,
\end{equation}
and we have the following bounds from Theorems \ref{cor:bounding_mu2_intersection}(4), \ref{cor4}(2), and \ref{cor5}(4):

\begin{proposition}
Assume that $\gamma_1,\kappa_1,\gamma_0,\kappa_0$ exist and are available, and $(\overline{\sigma}_1^{(2)},\overline{\sigma}_0^{(2)})$ are fixed.
Then we have

(1). $\Big(\sqrt{\overline{\sigma}^{(2)}_1}-\sqrt{\overline{\sigma}^{(2)}_0}\Big)^2  \leq \overline{\mu}^{(2)}
\leq \Big(\sqrt{\overline{\sigma}^{(2)}_1}+\sqrt{\overline{\sigma}^{(2)}_0}\Big)^2$.

(2). $-\frac{\sqrt{2}}{\sqrt[4]{27}}\Bigl(\sqrt[4]{(\overline{\sigma}_1^{(2)})^2\kappa_1}+\sqrt[4]{(\overline{\sigma}_0^{(2)})^2\kappa_0}\Bigr)^3  \leq \overline{\mu}^{(3)}  \leq \frac{\sqrt{2}}{\sqrt[4]{27}}\Bigl(\sqrt[4]{(\overline{\sigma}_1^{(2)})^2\kappa_1}+\sqrt[4]{(\overline{\sigma}_0^{(2)})^2\kappa_0}\Bigr)^3$.

(3). $\max\Big\{
\big(\sqrt{\overline{\sigma}_1^{(2)}}-\sqrt{\overline{\sigma}_0^{(2)}}\big)^4,  \big(\sqrt[4]{(\overline{\sigma}_1^{(2)})^2\kappa_1}-\sqrt[4]{(\overline{\sigma}_0^{(2)})^2\kappa_0}\big)^4 
\Big\} 
\leq \overline{\mu}^{(4)} \leq \big(\sqrt[4]{(\overline{\sigma}_1^{(2)})^2\kappa_1}+\sqrt[4]{(\overline{\sigma}_0^{(2)})^2\kappa_0}\big)^4$.
\end{proposition}

\begin{proof}
Substituting $\overline{\sigma}^{(3)}_x=\gamma_x(\overline{\sigma}^{(2)}_x)^{3/2}$ and $\overline{\sigma}_x^{(4)}=(\overline{\sigma}_x^{(2)})^2\kappa_x$ into the bounds in Theorems \ref{cor:bounding_mu2_intersection}(4), \ref{cor4}(2), and \ref{cor5}(4), we have the proposition.
\end{proof}

If we know the bounds of $(\overline{\sigma}_1^{(2)},\overline{\sigma}_0^{(2)})$, we have the following bounds:

\begin{corollary}
Assume that $\gamma_1,\kappa_1,\gamma_0,\kappa_0$ exist and are available, and $0\leq l_1^{(2)}\leq \overline
{\sigma}_1^{(2)}\leq u_1^{(2)}$, $0\leq l_0^{(2)}\leq \overline{\sigma}_0^{(2)}\leq u_0^{(2)}$ for some $l_1^{(2)}, u_1^{(2)}, l_0^{(2)}, u_0^{(2)}\in \mathbb{R}$.
Then we have

(1). $\displaystyle{\min_{l_1^{(2)}\leq x_1\leq u_1^{(2)},\,l_0^{(2)}\leq x_0\leq u_0^{(2)}}}\Big(\sqrt{x_1}-\sqrt{x_0}\Big)^2  \leq \overline{\mu}^{(2)}
\leq \Big(\sqrt{u_1^{(2)}}+\sqrt{u_0^{(2)}}\Big)^2$.

(2). $-\frac{\sqrt{2}}{\sqrt[4]{27}}\Bigl(\sqrt[4]{(u_1^{(2)})^2\kappa_1}+\sqrt[4]{(u_0^{(2)})^2\kappa_0}\Bigr)^3  \leq \overline{\mu}^{(3)}  \leq \frac{\sqrt{2}}{\sqrt[4]{27}}\Bigl(\sqrt[4]{(u_1^{(2)})^2\kappa_1}+\sqrt[4]{(u_0^{(2)})^2\kappa_0}\Bigr)^3$.

(3). $\displaystyle{\min_{l_1^{(2)}\leq x_1\leq u_1^{(2)},\,l_0^{(2)}\leq x_0\leq u_0^{(2)}}}\max\Big\{
\big(\sqrt{x_1}-\sqrt{x_0}\big)^4,  \big(\sqrt[4]{(x_1)^2\kappa_1}-\sqrt[4]{(x_0)^2\kappa_0}\big)^4 
\Big\} 
\leq \overline{\mu}^{(4)} \leq \big(\sqrt[4]{(u_1^{(2)})^2\kappa_1}+\sqrt[4]{(u_0^{(2)})^2\kappa_0}\big)^4$.
\end{corollary}

However, if we do not have the bounds of  $(\overline{\sigma}_1^{(2)},\overline{\sigma}_0^{(2)})$, we cannot derive nontrivial bounds by this approach.

\section{Proofs}


\label{sec:proof}

In this section, we provide proofs of the theorems stated in the main part.
Throughout this section, we put $\widetilde{Y}_x\coloneqq Y_x-\mathbb{E}[Y_x]$  ($x=0,1$), $\Delta\coloneqq Y_1-Y_0$ and $\widetilde{\Delta}\coloneqq \Delta-\mathbb{E}[\Delta]$.

\subsection{Lemmas}

In this subsection, we present several lemmas used in the proofs of the theorems.

\begin{lemma}
\label{lem:coupling_to_scm}
Let $(Z_1,Z_0)$ be any pair of real-valued random variables, and let $(c_1,c_0)\in\mathbb{R}^2$ be any constants.
Then there exists an SCM ${\cal M}$ with endogenous variables ${\boldsymbol V}=\{X,Y\}$, where $X \in \{0,1\}$,
such that the induced potential outcomes $(Y_1,Y_0)$ satisfy
\begin{equation}
(Y_1,Y_0)\overset{d}{=}(Z_1+c_1,\ Z_0+c_0).
\end{equation}
\end{lemma}

\begin{proof}
Let ${\boldsymbol U}=\{U_X,U_0,U_1\}$ be exogenous variables such that $(U_0,U_1)$ has the same joint distribution as $(Z_0,Z_1)$,
and $U_X\in\{0,1\}$ is arbitrary.
Define the structural equations by
\begin{equation}
\label{eq:scm_coupling_embedding_general}
\begin{aligned}
X &\coloneqq U_X,\\
Y &\coloneqq (U_0+c_0) + X\{(U_1+c_1)-(U_0+c_0)\}.
\end{aligned}
\end{equation}
Consider the sub-model ${\cal M}_x$ induced by the intervention $do(X=x)$.
In ${\cal M}_x$, \eqref{eq:scm_coupling_embedding_general} yields
\begin{equation}
Y = (U_0+c_0) + x\{(U_1+c_1)-(U_0+c_0)\} = U_x + c_x,
\end{equation}
where $c_x\in\{c_0,c_1\}$ corresponds to $x\in\{0,1\}$.
Therefore $(Y_1,Y_0)=(U_1+c_1,U_0+c_0)$ almost surely, and hence
$(Y_1,Y_0)\overset{d}{=}(Z_1+c_1,Z_0+c_0)$.
\end{proof}

\begin{lemma}
\label{lem:Var_updown_by_Cor}
Let $\overline{\sigma}_1^{(2)},\overline{\sigma}_0^{(2)}\geq 0$, and let $f_\pm$ be a function defined by 
\begin{equation}
    f_\pm (\lambda)=-\sqrt{\overline{\sigma}^{(2)}_0}\lambda \pm  \sqrt{\overline{\sigma}^{(2)}_0 \lambda^2 +(\overline{\sigma}^{(2)}_1-\overline{\sigma}^{(2)}_0)},
\end{equation}
where the domain is the set of $\lambda$ such that $-1\leq \lambda\leq 1$ and $\overline{\sigma}^{(2)}_0 \lambda^2+\overline{\sigma}^{(2)}_1-\overline{\sigma}^{(2)}_0\geq 0$.
Then we have the following:
 

(i) If $\overline{\sigma}^{(2)}_1 \geq \overline{\sigma}^{(2)}_0$, $f_+$ and $f_-$ are both defined in $[-1,1]$ and non-increasing.

(ii) If $\overline{\sigma}^{(2)}_1<\overline{\sigma}^{(2)}_0$, $f_+$ (resp. $f_-$) is defined in $\left[-1,-\sqrt{1-\overline{\sigma}^{(2)}_1/\overline{\sigma}^{(2)}_0}\right]\cup \left[\sqrt{1-\overline{\sigma}^{(2)}_1/\overline{\sigma}^{(2)}_0},1\right]$, decreasing (resp. increasing) if $\lambda<0$ and increasing (resp. decreasing) if $\lambda>0$, within its domain.
\end{lemma}

\begin{proof}
If $\overline{\sigma}_0^{(2)}=0$, $f_\pm(\lambda)=\pm\sqrt{
\overline{\sigma}_1^{(2)}}$ and (i) holds. The assumption of (ii) does not hold when $\overline{\sigma}_0^{(2)}=0$.

In the following, we assume $\overline{\sigma}_0^{(2)}>0$.
For simplicity of notation, we put $S(\lambda)\coloneqq \sqrt{\overline{\sigma}^{(2)}_0 \lambda^2+(\overline{\sigma}^{(2)}_1-\overline{\sigma}^{(2)}_0)}$ and $T(\lambda)\coloneqq\frac{\overline{\sigma}^{(2)}_0 \lambda}{S(\lambda)}$.

\underline{\emph{Domain.}}
If $\overline{\sigma}^{(2)}_1\geq\overline{\sigma}^{(2)}_0$, we have  $\overline{\sigma}^{(2)}_0 \lambda^2+(\overline{\sigma}^{(2)}_1-\overline{\sigma}^{(2)}_0)\geq 0$ for any $\lambda$, hence $f_\pm$ is defined for all $\lambda\in [-1,1]$. 
If $\overline{\sigma}^{(2)}_1<\overline{\sigma}^{(2)}_0$, $\overline{\sigma}^{(2)}_0 \lambda^2+(\overline{\sigma}^{(2)}_1-\overline{\sigma}^{(2)}_0)\geq 0$ holds if and only if $\lambda^2\geq 1-\overline{\sigma}^{(2)}_1/\overline{\sigma}^{(2)}_0$. Therefore $f_\pm$ is defined in $\left[-1,-\sqrt{(1-\overline{\sigma}^{(2)}_1/\overline{\sigma}^{(2)}_0)}\right]\cup \left[\sqrt{(1-\overline{\sigma}^{(2)}_1/\overline{\sigma}^{(2)}_0)},1\right]$.

\underline{\emph{Increase/Decrease.}}
We remark that we have
\begin{equation}
    f'_\pm(\lambda) = -\sqrt{\overline{\sigma}^{(2)}_0}\pm\frac{\overline{\sigma}^{(2)}_0\lambda}{\sqrt{\overline{\sigma}^{(2)}_0 \lambda^2 +(\overline{\sigma}^{(2)}_1-\overline{\sigma}^{(2)}_0)}}.
\end{equation}

(i): If $\overline{\sigma}^{(2)}_1 >\overline{\sigma}^{(2)}_0$, then $T(0)=0$, and if we further assume $\lambda\ne0$, we have $S(\lambda) > \sqrt{\overline{\sigma}^{(2)}_0}|\lambda|$ so
\begin{equation}
    |T(\lambda)|=\dfrac{\overline{\sigma}^{(2)}_0 |\lambda|}{S(\lambda)} < \dfrac{\overline{\sigma}^{(2)}_0|\lambda|}{\sqrt{\overline{\sigma}^{(2)}_0}|\lambda|}=\sqrt{\overline{\sigma}^{(2)}_0}.
\end{equation}
Thus $|T(\lambda)|\leq \sqrt{\overline{\sigma}^{(2)}_0}$ for all $\lambda\in[-1,1]$. Then we have $f'_\pm(\lambda) \leq 0$, therefore $f_+,f_-$ are both non-increasing.

If $\overline{\sigma}^{(2)}_1=\overline{\sigma}^{(2)}_0$, $f_{\pm} (\lambda)=-\sqrt{\overline{\sigma}^{(2)}_0}\lambda\pm \sqrt{\overline{\sigma}^{(2)}_0}|\lambda|$, so we have
\begin{equation}
\label{eq:case_sig1=sig0}
f_+(\lambda)=
\begin{cases} 0 & (\lambda\geq 0) \\ -2\sqrt{\overline{\sigma}^{(2)}_0}\lambda & (\lambda \leq 0) \end{cases},\quad 
f_-(\lambda)=\begin{cases} -2\sqrt{\overline{\sigma}^{(2)}_0}\lambda & (\lambda\geq 0) \\ 0 & (\lambda \leq 0)\end{cases}.
\end{equation}
Then $f_+, f_-$ are both non-increasing.

(ii): If $\overline{\sigma}^{(2)}_1<\overline{\sigma}^{(2)}_0$, we have $|T(\lambda)|>\sqrt{\overline{\sigma}^{(2)}_0}$. Therefore we have
\begin{equation}
    f'_+(\lambda)\begin{cases}
        <0 & (\lambda<0) \\ >0 &(\lambda>0)
    \end{cases},\qquad
    f'_-(\lambda)\begin{cases}
        >0 &(\lambda<0) \\ <0 & (\lambda>0)
    \end{cases}
\end{equation}
(in the domain of $f_\pm$). Hence we have the claim.
\end{proof}

\begin{lemma}
\label{lem:sharp_third_from_fourth_newextra}
Let $Z$ be a random variable with $\mathbb{E}[Z]=0$ and $\mathbb{E}[Z^4]<\infty$. 
Then
\begin{equation}\label{eq:sharp_third_from_fourth_newextra}
\bigl|\mathbb{E}[Z^3]\bigr|
\leq
\frac{\sqrt{2}}{\sqrt[4]{27}}\Bigl(\sqrt[4]{\mathbb{E}[Z^4]}\Bigr)^{3}.
\end{equation}
Moreover, there exists a random variable $Z_0$ with $\mathbb{E}[Z_0]=0$, $\mathbb{E}[Z_0^4]=1$, and $\mathbb{E}[Z_0^3]=\sqrt{2}/\sqrt[4]{27}$. 
\end{lemma}

\begin{proof}
If $Z=0$ almost surely, we have \eqref{eq:sharp_third_from_fourth_newextra}. In the following, we assume $\mathbb{P}(Z=0)<1$.

Let $s\coloneqq \mathbb{E}[Z^2]$ and $\kappa\coloneqq \mathbb{E}[Z^4]$.
By the Cauchy--Schwarz inequality, $s^2\leq \kappa$.
Since $\mathbb{E}[Z]=0$,
$\mathbb{E}[Z^3] = \mathbb{E}\bigl[Z\bigl(Z^2-s\bigr)\bigr]$,
and the Cauchy--Schwarz inequality yields
$\bigl|\mathbb{E}[Z^3]\bigr|^2
\leq \mathbb{E}[Z^2]\cdot \mathbb{E}\bigl[(Z^2-s)^2\bigr]
= s(\kappa-s^2)$.
Hence
\begin{equation}
\bigl|\mathbb{E}[Z^3]\bigr|
\leq
\sqrt{s(\kappa-s^2)}
=
\Bigl(\sqrt[4]{\kappa}\Bigr)^3\sqrt{u(1-u^2)},
\qquad
u\coloneqq \frac{s}{\sqrt{\kappa}}\in[0,1].
\end{equation}
The function $u\mapsto u(1-u^2)$ is maximized over $[0,1]$ at $u=1/\sqrt{3}$, with
\begin{equation}
\max_{u\in[0,1]}\sqrt{u(1-u^2)}
=
\sqrt{\frac{2}{3\sqrt{3}}}
=
\frac{\sqrt{2}}{\sqrt[4]{27}}.
\end{equation}
Substituting this bound gives \eqref{eq:sharp_third_from_fourth_newextra}.

For any random variable $Z$ with probabilities $\mathbb{P}(Z=a)=b/(a+b)$ and $\mathbb{P}(Z=-b)=a/(a+b)$ ($a,b\geq 0,\,a+b>0$), one checks
\begin{equation}
\mathbb{E}[Z]=0,\qquad 
\mathbb{E}[Z^2]=ab, \qquad 
\mathbb{E}[Z^3]=ab(a-b), \qquad
\mathbb{E}[Z^4]=ab(a^2-ab+b^2).
\end{equation}

We define a random variable $Z_0$ by
\begin{equation}
\mathbb{P}(Z_0=a_0)\coloneqq \frac{b_0}{a_0+b_0},
\quad \mathbb{P}(Z_0=-b_0)\coloneqq \frac{a_0}{a_0+b_0}\quad\text{ where }\quad 
a_0\coloneqq \sqrt{\frac{3+2\sqrt{3}}{3}},
\quad b_0\coloneqq \sqrt{\frac{2\sqrt{3}-3}{3}}.
\end{equation}
Then
$\mathbb{E}[Z_0^4] = a_0b_0\left((a_0^2+b_0^2)-a_0b_0\right) = \frac{1}{\sqrt{3}}\Bigl(\frac{4\sqrt{3}}{3}-\frac{1}{\sqrt{3}}\Bigr) = 1$.
Also,
$(a_0-b_0)^2=(a_0^2+b_0^2)-2a_0b_0=\frac{4\sqrt{3}}{3}-\frac{2}{\sqrt{3}}=\frac{2}{\sqrt{3}}$,
so
$\mathbb{E}[Z_0^3] = a_0b_0(a_0-b_0) = \frac{1}{\sqrt{3}}\sqrt{\frac{2}{\sqrt{3}}}
= \frac{\sqrt{2}}{\sqrt[4]{27}}$.
\end{proof}

\begin{lemma}
\label{lem:two_point_23}
For any $v>0$ and any $t\in\mathbb{R}$, there exists a random variable $W$
such that $\mathbb{E}[W]=0$, $\mathbb{E}[W^2]=v$ and $\mathbb{E}[W^3]=t$.
\end{lemma}

\begin{proof}
Let $\kappa \coloneqq t/v^{3/2}\in\mathbb{R}$ and define a function $\phi:(0,1)\to\mathbb{R}$ by
\begin{equation}
\phi(p)\coloneqq \frac{1-2p}{\sqrt{p(1-p)}}.
\end{equation}
Then $\phi$ is continuous,
$\lim_{p\downarrow 0}\phi(p)=+\infty$, and $\lim_{p\uparrow 1}\phi(p)=-\infty$.
Hence there exists $p\in(0,1)$ such that $\phi(p)=\kappa$.
Define $W$ by
\begin{equation}
W=
\begin{cases}
\ \sqrt{v(1-p)/p} & \text{with probability\ } p,\\
-\sqrt{vp/(1-p)} & \text{with probability\ } 1-p.
\end{cases}
\end{equation}
Then we have $\mathbb{E}[W]=0$, $\mathbb{E}[W^2]=v$
and $\mathbb{E}[W^3]=v^{3/2}\phi(p)=t$.
\end{proof}

\begin{lemma}\label{lem:construction_mu4_24}
Assume $\overline{\sigma}_x^{(4)}\ge \bigl(\overline{\sigma}_x^{(2)}\bigr)^2$ for $x\in\{0,1\}$.
There exist random variables $A,B\geq 0$ such that
\begin{equation}\label{eq:AB_moments_24}
\begin{gathered}
\mathbb{E}[A]=\overline{\sigma}_1^{(2)},\qquad
\mathbb{E}[A^2]=\overline{\sigma}_1^{(4)},\qquad
\mathbb{E}[B]=\overline{\sigma}_0^{(2)},\qquad
\mathbb{E}[B^2]=\overline{\sigma}_0^{(4)},
\\
\mathbb{E}[AB]=\overline{\sigma}_1^{(2)}\overline{\sigma}_0^{(2)}
+\sqrt{\overline{\sigma}_1^{(4)}-\bigl(\overline{\sigma}_1^{(2)}\bigr)^2}
\sqrt{\overline{\sigma}_0^{(4)}-\bigl(\overline{\sigma}_0^{(2)}\bigr)^2}
 \eqqcolon \beta_{2,2}^{\mathrm{CS}}.
\end{gathered}
\end{equation}
Moreover, there exist random signs $S$ independent of $(A,B)$ such that, defining $(\widetilde Y_1,\widetilde Y_0):=(S\sqrt{A},S\sqrt{B})$, we have $\mathbb{E}[\widetilde Y_x]=0$ and $\mathbb{E}[\widetilde Y_x^2]=\overline{\sigma}_x^{(2)}$, $\mathbb{E}[\widetilde Y_x^4]=\overline{\sigma}_x^{(4)}$ for $x\in\{0,1\}$.
\end{lemma}

\begin{proof}
By definition of $\overline{\sigma}_x^{(k)}$, we have $\overline{\sigma}_x^{(2)}=0 \iff \overline{\sigma}_x^{(4)}=0$.

If $\overline{\sigma}_1^{(2)}=0$ and $\overline{\sigma}_0^{(2)}>0$, take $A=0$ almost surely. If $\overline{\sigma}_0^{(4)}=(\overline{\sigma}_0^{(2)})^2$, take $B=\overline{\sigma}_0^{(2)}$ almost surely. 
Otherwise, choose $t>0$ such that $t<\frac{\overline{\sigma}_0^{(2)}}{\sqrt{\overline{\sigma}_0^{(4)}-(\overline{\sigma}_0^{(2)})^2}}$, let $Z$ be a random variable with
\begin{equation}\label{eq:Z_law_24}
\mathbb{P}(Z=-t)=\frac{1}{1+t^2},\qquad \mathbb{P}(Z=1/t)=\frac{t^2}{1+t^2}.
\end{equation}
whence $\mathbb{E}[Z]=0$ and $\mathbb{E}[Z^2]=1$.
Then, set $B=\overline{\sigma}_0^{(2)}+\sqrt{\overline{\sigma}_0^{(4)}-(\overline{\sigma}_0^{(2)})^2}Z$. 
Then $A,B\geq0$, $\mathbb{E}[A]=\mathbb{E}[A^2]=0$, $\mathbb{E}[B]=\overline{\sigma}_0^{(2)}$, $\mathbb{E}[B^2]=\overline{\sigma}_0^{(4)}$, and $\mathbb{E}[AB]=0=\beta_{2,2}^{\mathrm{CS}}$. 
The case $\overline{\sigma}_0^{(2)}=0$ and $\overline{\sigma}_1^{(2)}>0$ is symmetric.

If further  $\overline{\sigma}_1^{(2)}=\overline{\sigma}_0^{(2)}=0$, we can take $A=B=0$ almost surely.

Assume now that $\overline{\sigma}_1^{(2)}>0$ and $\overline{\sigma}_0^{(2)}>0$. Take $t\in \mathbb{R}$ with
$0 < t <
\min\Big\{
\frac{\overline{\sigma}_1^{(2)}}{\sqrt{\overline{\sigma}_1^{(4)}-(\overline{\sigma}_1^{(2)})^2}},
\,
\frac{\overline{\sigma}_0^{(2)}}{\sqrt{\overline{\sigma}_0^{(4)}-(\overline{\sigma}_0^{(2)})^2}}
\Big\}$, where we interpret $a/0\coloneqq +\infty$ for $a>0$.
Set $p\coloneqq 1/(1+t^2)$.
Let $Z$ be a random variable with \eqref{eq:Z_law_24}.
Define
\begin{equation}
A\coloneqq \overline{\sigma}_1^{(2)}+\sqrt{\overline{\sigma}_1^{(4)}-\bigl(\overline{\sigma}_1^{(2)}\bigr)^2}\,Z,
\qquad
B\coloneqq \overline{\sigma}_0^{(2)}+\sqrt{\overline{\sigma}_0^{(4)}-\bigl(\overline{\sigma}_0^{(2)}\bigr)^2}\,Z.
\end{equation}
Since $Z\in\{-t,1/t\}$, we have
$A\geq \overline{\sigma}_1^{(2)}-\sqrt{\overline{\sigma}_1^{(4)}-(\overline{\sigma}_1^{(2)})^2}t\geq 0$,
 $B\geq \overline{\sigma}_0^{(2)}-\sqrt{\overline{\sigma}_0^{(4)}-(\overline{\sigma}_0^{(2)})^2}t\geq 0$.
By definition and $\mathbb{E}[Z]=0$, $\mathbb{E}[A]=\overline{\sigma}_1^{(2)}$ and $\mathbb{E}[B]=\overline{\sigma}_0^{(2)}$.
Further, by $\mathbb{E}[Z^2]=1$,
\begin{equation}
\mathbb{E}[A^2]=\bigl(\overline{\sigma}_1^{(2)}\bigr)^2+\bigl(\overline{\sigma}_1^{(4)}-(\overline{\sigma}_1^{(2)})^2\bigr)=\overline{\sigma}_1^{(4)},
\qquad
\mathbb{E}[B^2]=\overline{\sigma}_0^{(4)}.
\end{equation}
Also,
\begin{equation}
\mathbb{E}[AB]
=\overline{\sigma}_1^{(2)}\overline{\sigma}_0^{(2)}
+\sqrt{\overline{\sigma}_1^{(4)}-(\overline{\sigma}_1^{(2)})^2}\,
\sqrt{\overline{\sigma}_0^{(4)}-(\overline{\sigma}_0^{(2)})^2}\,
\mathbb{E}[Z^2]
=\beta_{2,2}^{\mathrm{CS}}.
\end{equation}
Finally, take $S$ independent of $Z$ with $\mathbb{P}(S=1)=\mathbb{P}(S=-1)=1/2$.
Since $(A,B)$ is a function of $Z$, $S$ is independent of $(A,B)$.
Set $(\widetilde Y_1,\widetilde Y_0)=(S\sqrt{A},S\sqrt{B})$.
Then $\mathbb{E}[\widetilde Y_x]=0$ and $\mathbb{E}[\widetilde Y_x^2]=\mathbb{E}[A]$ or $\mathbb{E}[B]$, $\mathbb{E}[\widetilde Y_x^4]=\mathbb{E}[A^2]$ or $\mathbb{E}[B^2]$, as required.
\end{proof}

\begin{lemma}\label{lem:exist_odd_moment}
Let $k\geq 3$ be an odd integer and let $t\in\mathbb{R}$.
There exists a random variable $Z$ such that $\mathbb{E}[Z]=0$ and $\mathbb{E}[Z^k]=t$.
\end{lemma}

\begin{proof}
If $t=0$, take $Z=0$.
Otherwise, define a random variable $T$ by $\mathbb{P}(T=2)=1/3$ and $\mathbb{P}(T=-1)=2/3$.
Then $\mathbb{E}[T]=0$ and $\mathbb{E}[T^k]=(2^k-2)/3$.
Set $Z\coloneqq \sqrt[k]{t/\mathbb{E}[T^k]}T$.
Then $\mathbb{E}[Z]=0$ and $\mathbb{E}[Z^k]=t$.
\end{proof}

\subsection{Proofs}


\Identification*

\begin{proof}
Let $\Delta\coloneqq Y_1-Y_0$.
Since $Y_1-\mathbb{E}[Y_1]=(Y_0-\mathbb{E}[Y_0])+(\Delta-\mathbb{E}[\Delta])$, by the binomial theorem 
\begin{equation}
\overline{\sigma}_1^{(m)}
=\sum_{\ell=0}^m \binom{m}{\ell}\mathbb{E}[(Y_0-\mathbb{E}[Y_0])^{m-\ell}(\Delta-\mathbb{E}[\Delta])^\ell].
\end{equation}
Under the condition $Y_0\indep \Delta$, the expectation factorizes into $\mathbb{E}[(Y_0-\mathbb{E}[Y_0])^{m-\ell}(\Delta-\mathbb{E}[\Delta])^\ell]=\mathbb{E}[(Y_0-\mathbb{E}[Y_0])^{m-\ell}]\mathbb{E}[(\Delta-\mathbb{E}[\Delta])^\ell]=\overline{\sigma}_0^{(m-\ell)}\overline{\mu}^{(\ell)}$.
Using $\overline{\sigma}_0^{(1)}=0$ and $\overline{\mu}^{(1)}=0$, we obtain
\begin{equation}
\overline{\sigma}_1^{(m)}
=\overline{\sigma}_0^{(m)}+\overline{\mu}^{(m)}+\sum_{\ell=2}^{m-2}\binom{m}{\ell}\overline{\sigma}_0^{(m-\ell)}\overline{\mu}^{(\ell)},
\end{equation}
which is equivalent to \eqref{eq:identification_under_independence}.
\end{proof}

\SensAn*

\begin{proof}
If $\overline{\sigma}_0^{(2)}=0$, then $\widetilde{Y}_0=0$ almost surely. Hence $\overline{\mu}^{(2)}=\overline{\sigma}_1^{(2)}$ and $\lambda=0$ by definition. This proves (1) when $\overline{\sigma}_1^{(2)}>\overline{\sigma}_0^{(2)}=0$ and (2) when $\overline{\sigma}_1^{(2)}=\overline{\sigma}_0^{(2)}=0$. 

In the following, assume $\overline{\sigma}_0^{(2)}>0$.
Since
$\mathbb{E}[(Y_1-Y_0-\mathbb{E}[Y_1-Y_0])(Y_0-\mathbb{E}[Y_0])]=\Cov(Y_1-Y_0,Y_0)=\lambda \sqrt{\overline{\mu}^{(2)}}\sqrt{\overline{\sigma}^{(2)}_0}$,
we have 
$\overline{\mu}^{(2)} =\overline{\sigma}^{(2)}_1-\overline{\sigma}^{(2)}_0-2 \lambda\sqrt{\overline{\mu}^{(2)}}\sqrt{\overline{\sigma}^{(2)}_0}$ .
In other words, $\sqrt{\overline{\mu}^{(2)}}$ is a solution to the equation
\begin{equation}
\label{eq:equation_of_var_in_x}
    x^2+2\lambda\sqrt{\overline{\sigma}_0^{(2)}}x-(\overline{\sigma}^{(2)}_1-\overline{\sigma}^{(2)}_0)=0.
\end{equation}

The solutions to this equation are given by $f_\pm(\lambda)$, which are defined in Lemma \ref{lem:Var_updown_by_Cor}. 
Since $\sqrt{\overline{\mu}^{(2)}}\geq 0$, at least one of $f_\pm(\lambda)$ is nonnegative. Hence  $\lambda$ satisfies $f_+(\lambda)\geq 0$. 
We remark that $f_\pm(-1)=\sqrt{\overline{\sigma}_0^{(2)}} \pm\sqrt{\overline{\sigma}_1^{(2)}}$ and $f_\pm(1)= -\sqrt{\overline{\sigma}_0^{(2)}} \pm\sqrt{\overline{\sigma}_1^{(2)}}$.

(1): If $\overline{\sigma}^{(2)}_1>\overline{\sigma}^{(2)}_0$, 
$f_-(\lambda)<0$ for all $\lambda\in [-1,1]$ since $f_-$ is decreasing (Lemma \ref{lem:Var_updown_by_Cor}) and $f_-(-1)<0$. 
Moreover, $f_+(\lambda)\geq 0$ for all $\lambda\in [-1,1]$ since $f_+$ is decreasing and $f_+(1)\geq 0$.
Therefore, $\lambda$ can take any value of $[-1,1]$ and  $\sqrt{\overline{\mu}^{(2)}}$ is equal to $f_+(\lambda)= -\sqrt{\overline{\sigma}^{(2)}_0}\lambda + \sqrt{\overline{\sigma}^{(2)}_1-(1-\lambda^2)\overline{\sigma}^{(2)}_0}$.

(2): Assume $\overline{\sigma}^{(2)}_1=\overline{\sigma}^{(2)}_0$. Then 
$f_\pm(\lambda)=\sqrt{\overline{\sigma}_0^{(2)}}(-\lambda\pm|\lambda|)$.
If $\lambda> 0$, we have $f_+(\lambda)=0,\,f_-(\lambda)<0$ and hence we get $\overline{\mu}^{(2)}=0$. However, $\Var(Y_1-Y_0)=0$ in this case, so by definition we have $\lambda=0$; hence $\lambda>0$ is infeasible.
If $\lambda=0$, we have $\overline{\mu}^{(2)}=0$.
If $\lambda<0$, $f_+(\lambda)=2\sqrt{\overline{\sigma}_0^{(2)}}|\lambda|$ and $f_-(\lambda)=0$. Since $\overline{\mu}^{(2)}=0$ contradicts with $\lambda<0$, we have $\sqrt{\overline{\mu}^{(2)}}=2\sqrt{\overline{\sigma}_0^{(2)}}|\lambda|$ .
By summarizing the above results, we have the claim.

(3): Assume $\overline{\sigma}^{(2)}_1<\overline{\sigma}^{(2)}_0$, and let $s\coloneqq \sqrt{1-\overline{\sigma}^{(2)}_1/\overline{\sigma}^{(2)}_0}$. Both $f_+$ and $f_-$ are defined in $[-1,-s]\cup [s,1]$, and we see
\begin{equation}
\begin{gathered}
\label{eq:value_of_f_pm(1)}
    f_-(-s)
    = \sqrt{\overline{\sigma}^{(2)}_0-\overline{\sigma}^{(2)}_1}  >0, \qquad  f_-(s) = -\sqrt{\overline{\sigma}^{(2)}_0-\overline{\sigma}^{(2)}_1} <0, \\
    f_+(-s) = \sqrt{\overline{\sigma}^{(2)}_0-\overline{\sigma}^{(2)}_1}>0, \qquad
    f_+(1)=\sqrt{\overline{\sigma}^{(2)}_1}-\sqrt{\overline{\sigma}^{(2)}_0}<0.
\end{gathered}
\end{equation}
By Lemma \ref{lem:Var_updown_by_Cor}, $f_\pm(\lambda)<0$ for all $\lambda\in [s,1]$, and $f_\pm(\lambda)>0$ for all $\lambda\in [-1,-s]$.
Therefore, the set of $\lambda$ compatible with $\overline{\sigma}^{(2)}_1<\overline{\sigma}^{(2)}_0$ is $[-1,-s]$, and for $\lambda$ in this range, $\sqrt{\overline{\mu}^{(2)}}$ can be equal to $f_+(\lambda)$ or $f_-(\lambda)$. Hence we have the claim. 
\end{proof}

\SecondfromS*

\begin{proof}
Let $f_\pm$ be as in Lemma \ref{lem:Var_updown_by_Cor}.
Note that $f_\pm$ are continuous in their domain.
To prove the inequality, it suffices to show
\begin{equation}
    \label{eq:bound_of_Root_of_var(lem)}
    \left|\sqrt{\overline{\sigma}^{(2)}_1}-\sqrt{\overline{\sigma}^{(2)}_0} \right|
    \leq \sqrt{\overline{\mu}^{(2)}}
    \leq   \sqrt{\overline{\sigma}^{(2)}_1}+\sqrt{\overline{\sigma}^{(2)}_0}.
\end{equation}

If $\sqrt{\overline{\sigma}^{(2)}_1}>\sqrt{\overline{\sigma}^{(2)}_0}$, we have $\sqrt{\overline{\mu}^{(2)}}\leq \max_\lambda f_+(\lambda)=f_+(-1)=\sqrt{\overline{\sigma}^{(2)}_1}+\sqrt{\overline{\sigma}^{(2)}_0}$
and $\sqrt{\overline{\mu}^{(2)}}\geq \min_\lambda f_+(\lambda)=f_+(1)=\sqrt{\overline{\sigma}^{(2)}_1}-\sqrt{\overline{\sigma}^{(2)}_0}$ by Lemma \ref{lem:Var_updown_by_Cor}. 

If $\sqrt{\overline{\sigma}^{(2)}_1}=\sqrt{\overline{\sigma}^{(2)}_0}$,
we see $0\leq\sqrt{\overline{\mu}^{(2)}}\leq 2\sqrt{\overline{\sigma}_0^{(2)}}$ as a consequence of \eqref{eq:case_sig1=sig0} since $\sqrt{\overline{\mu}^{(2)}}$ is nonnegative. 

If $\sqrt{\overline{\sigma}^{(2)}_1}<\sqrt{\overline{\sigma}^{(2)}_0}$,
we see  
$f_-(-1) =\sqrt{\overline{\sigma}^{(2)}_0}-\sqrt{\overline{\sigma}^{(2)}_1}$
and 
$f_+(-1) =\sqrt{\overline{\sigma}^{(2)}_0}+\sqrt{\overline{\sigma}^{(2)}_1}$.
Combining these with \eqref{eq:value_of_f_pm(1)}, we see that the range of $\sqrt{\overline{\mu}^{(2)}}$ is $\left[\sqrt{\overline{\sigma}^{(2)}_0}-\sqrt{\overline{\sigma}^{(2)}_1}, \sqrt{\overline{\sigma}^{(2)}_0-\overline{\sigma}^{(2)}_1}\right] \cup \left[\sqrt{\overline{\sigma}^{(2)}_0-\overline{\sigma}^{(2)}_1},\sqrt{\overline{\sigma}^{(2)}_0}+\sqrt{\overline{\sigma}^{(2)}_1}\right]
=\left[\sqrt{\overline{\sigma}^{(2)}_0}-\sqrt{\overline{\sigma}^{(2)}_1},\sqrt{\overline{\sigma}^{(2)}_0}+\sqrt{\overline{\sigma}^{(2)}_1}\right]$ by Lemma \ref{lem:Var_updown_by_Cor}, which implies \eqref{eq:bound_of_Root_of_var(lem)}.

Take a random variable $S$ with $\mathbb{P}(S=1)=\mathbb{P}(S=-1)=1/2$.
Set $\widetilde{Y}_0=\sqrt{\overline{\sigma}_0^{(2)}}S$ and
$\widetilde{Y}_1=\sqrt{\overline{\sigma}_1^{(2)}}S$. Then
$\overline{\mu}^{(2)}=(\sqrt{\overline{\sigma}_1^{(2)}}-\sqrt{\overline{\sigma}_0^{(2)}})^2$.
If we set $\widetilde{Y}_1=-\sqrt{\overline{\sigma}_1^{(2)}}S$ instead, then
$\overline{\mu}^{(2)}=(\sqrt{\overline{\sigma}_1^{(2)}}+\sqrt{\overline{\sigma}_0^{(2)}})^2$.
By Lemma \ref{lem:coupling_to_scm}, these distributions are realized by some SCMs, so the bound \eqref{eq:bound_of_var} is sharp.
\end{proof}

\VarSignBound*

\begin{proof}
Let $f_\pm$ be as in Lemma \ref{lem:Var_updown_by_Cor}.

\underline{\emph{Bounds.}}
(1):
If $\overline{\sigma}_1^{(2)}<\overline{\sigma}_0^{(2)}$, the proof of Theorem \ref{thm:sensitivity_analysis_of_var} implies $\lambda<0$. 
Hence $\overline{\sigma}_1^{(2)} \geq \overline{\sigma}_0^{(2)}$.
If $\overline{\sigma}_1^{(2)}>\overline{\sigma}_0^{(2)}$, by Lemma \ref{lem:Var_updown_by_Cor} we have $f_+(1)\leq \sqrt{\overline{\mu}^{(2)}}\leq f_+(0)$, which is equivalent to \eqref{eq:bound_of_var_lambda_ge0}. The case $\overline{\sigma}_1^{(2)}=\overline{\sigma}_0^{(2)}$ is Theorem \ref{thm:sensitivity_analysis_of_var}(2).

(2):
If $\overline{\sigma}_1^{(2)}>\overline{\sigma}_0^{(2)}$, by Lemma \ref{lem:Var_updown_by_Cor}, we have $f_+(0)\leq \sqrt{\overline{\mu}^{(2)}}\leq f_+(-1)$, which is equivalent to \eqref{eq:bound_of_var_lambda_le0}.
If $\overline{\sigma}_1^{(2)}\leq \overline{\sigma}_0^{(2)}$, we see \eqref{eq:bound_of_var} since $\lambda>0$ does not hold in this case. 


\underline{\emph{Sharpness.}}
Take random variables $S,T$ with $\mathbb{P}(S=1)=\mathbb{P}(S=-1)=\mathbb{P}(T=1)=\mathbb{P}(T=-1)=1/2$ and assume that $S$ and $T$ are independent.
Set $\widetilde{Y}_0\coloneqq \sqrt{\overline{\sigma}_0^{(2)}}S$.

Set $\widetilde{Y}_1\coloneqq \sqrt{\overline{\sigma}_1^{(2)}}S$.
Then $\Var(\widetilde{Y}_x)=\overline{\sigma}_x^{(2)}$ and $\overline{\mu}^{(2)}=(\sqrt{\overline{\sigma}_1^{(2)}}-\sqrt{\overline{\sigma}_0^{(2)}})^2$.
If $\overline{\sigma}_0^{(2)}>0$ and $\overline{\sigma}_1^{(2)}\neq\overline{\sigma}_0^{(2)}$, then $\lambda=1$ when $\overline{\sigma}_1^{(2)}>\overline{\sigma}_0^{(2)}$ and $\lambda=-1$ when $\overline{\sigma}_1^{(2)}<\overline{\sigma}_0^{(2)}$.
If $\Var(Y_1-Y_0)\Var(Y_0)=0$, then $\lambda=0$ by definition.

If $\overline{\sigma}_1^{(2)}\geq \overline{\sigma}_0^{(2)}$, set $\widetilde{Y}_1\coloneqq \sqrt{\overline{\sigma}_0^{(2)}}S+\sqrt{\overline{\sigma}_1^{(2)}-\overline{\sigma}_0^{(2)}}T$.
Then $\Var(\widetilde{Y}_x)=\overline{\sigma}_x^{(2)}$, $\overline{\mu}^{(2)}=\overline{\sigma}_1^{(2)}-\overline{\sigma}_0^{(2)}$, and $\lambda=0$.

Set $\widetilde{Y}_1\coloneqq -\sqrt{\overline{\sigma}_1^{(2)}}S$.
Then $\Var(\widetilde{Y}_x)=\overline{\sigma}_x^{(2)}$, $\overline{\mu}^{(2)}=(\sqrt{\overline{\sigma}_1^{(2)}}+\sqrt{\overline{\sigma}_0^{(2)}})^2$, and $\lambda=-1$ if $\overline{\sigma}_0^{(2)}>0$, while $\lambda=0$ if $\overline{\sigma}_0^{(2)}=0$ by definition.

Under $\lambda\geq 0$, the first and second constructions attain both bounds of \eqref{eq:bound_of_var_lambda_ge0}.
Under $\lambda\leq 0$, the second and third constructions attain both bounds of \eqref{eq:bound_of_var_lambda_le0}.
If $\overline{\sigma}_1^{(2)}\leq \overline{\sigma}_0^{(2)}$, the first and third constructions attain both bounds of \eqref{eq:bound_of_var} and satisfy $\lambda\leq 0$.
Lemma \ref{lem:coupling_to_scm} yields an SCM attaining each bound, hence we have the sharpness.

\end{proof}

\SecondfromT*

\begin{proof}
We have $0\leq \overline{\mu}^{(2)}\leq \infty$ by definition.

\underline{\emph{Sharpness (lower bound).}}
Fix $\varepsilon>0$.
Choose a random variable $B$ with $\mathbb{E}[B]=0$ and $\mathbb{E}[B^3]=2\overline{\sigma}_0^{(3)}$.
Take a constant $M_0>0$ such that $M_0^2>\mathbb{E}[B^2]$ and $|\delta|<\sqrt{\varepsilon}$, where $\delta \coloneqq \dfrac{2(\overline{\sigma}_1^{(3)}-\overline{\sigma}_0^{(3)})}{3(M_0^2-\mathbb{E}[B^2])}$.
Then define $\widetilde{Y}_0$ by
\begin{equation}
\widetilde{Y}_0 \coloneqq
\begin{cases}
B & \text{with probability\ } 1/2,\\
M_0 & \text{with probability\ } 1/4,\\
-M_0 & \text{with probability\ } 1/4
\end{cases}
\end{equation}
and define $\widetilde{Y}_1$ by
\begin{equation}
\widetilde{Y}_1 \coloneqq
\begin{cases}
\widetilde{Y}_0-\delta & \text{if } \widetilde{Y}_0=B,\\
\widetilde{Y}_0+\delta & \text{if } \widetilde{Y}_0\in\{M_0,-M_0\}
\end{cases}.
\end{equation}
Then $\mathbb{E}[\widetilde{Y}_0]=\mathbb{E}[\widetilde{\Delta}]=0$ (recall $\widetilde{\Delta}=\widetilde{Y}_1-\widetilde{Y}_0$), hence $\mathbb{E}[\widetilde{Y}_1]=0$.
Moreover $\widetilde{\Delta}^2\equiv \delta^2$ and $\mathbb{E}[\widetilde{\Delta}^3]=0$, so
\begin{equation}
\mathbb{E}[\widetilde{Y}_1^3]=\mathbb{E}[\widetilde{Y}_0^3]+3\mathbb{E}[\widetilde{Y}_0^2\widetilde{\Delta}].
\end{equation}
Since $\mathbb{E}[\widetilde{Y}_0^3]=\overline{\sigma}_0^{(3)}$ and
$\mathbb{E}[\widetilde{Y}_0^2\widetilde{\Delta}]
=-\frac{1}{2} \mathbb{E}[B^2]\delta+\frac{1}{2} M_0^2\delta
=\frac{\delta}{2}(M_0^2-\mathbb{E}[B^2])$,
we have $\mathbb{E}[\widetilde{Y}_1^3]=\overline{\sigma}_1^{(3)}$.
We also see $\overline{\mu}^{(2)}=\mathbb{E}[\widetilde{\Delta}^2]=\delta^2<\varepsilon$.
By Lemma \ref{lem:coupling_to_scm}, these are realized by an SCM.

\underline{\emph{Sharpness (upper bound).}}
Fix $M>0$.
Let $B_0,B_1$ be random variables with $\mathbb{E}[B_x]=0,\,\mathbb{E}[B_x^3]=2\overline{\sigma}_x^{(3)}$ ($x\in\{0,1\}$).
Take a constant $A>0$ and define $(\widetilde{Y}_1,\widetilde{Y}_0)$ by
\begin{equation}
(\widetilde{Y}_1,\widetilde{Y}_0)\coloneqq
\begin{cases}
(B_1,B_0) & \text{with probability\ } 1/2,\\
(A,-A) & \text{with probability\ } 1/4,\\
(-A,A) & \text{with probability\ } 1/4.
\end{cases}
\end{equation}
Then $\mathbb{E}[\widetilde{Y}_1^3]=\overline{\sigma}_1^{(3)}$ and $\mathbb{E}[\widetilde{Y}_0^3]=\overline{\sigma}_0^{(3)}$.
We have $\widetilde{\Delta}=\widetilde{Y}_1-\widetilde{Y}_0=\pm 2A$ when $(\widetilde{Y}_1,\widetilde{Y}_0)=(\pm A,\mp A)$, and we have
$\overline{\mu}^{(2)}=\mathbb{E}[\widetilde{\Delta}^2]\geq \frac12 (2A)^2=2A^2$.
Taking $A>\sqrt{M/2}$ yields $\overline{\mu}^{(2)}>M$.
By Lemma \ref{lem:coupling_to_scm}, these are realized by an SCM.
\end{proof}

\SecondfromF*

\begin{proof}
\underline{\emph{Bounds.}}
By the Cauchy--Schwarz inequality, for each $x\in\{0,1\}$ we have $\mathbb{E}[\widetilde{Y}_x^2]\leq \sqrt{\mathbb{E}[\widetilde{Y}_x^4]}\sqrt{1^2}=\sqrt{\overline{\sigma}_x^{(4)}}$.
Hence $\|\widetilde{Y}_x\|_2\leq \sqrt[4]{\overline{\sigma}_x^{(4)}}$.
By Minkowski's inequality in $L^2$,
$\|\widetilde{Y}_1-\widetilde{Y}_0\|_2 \leq \|\widetilde{Y}_1\|_2+\|\widetilde{Y}_0\|_2 \leq \sqrt[4]{(\overline{\sigma}_1^{(4)})}+\sqrt[4]{\overline{\sigma}_0^{(4)}}$.
Hence 
$\overline{\mu}^{(2)} \leq (\sqrt[4]{\overline{\sigma}_1^{(4)}}+\sqrt[4]{\overline{\sigma}_0^{(4)}})^2$.
The lower bound $0\leq \overline{\mu}^{(2)}$ follows from the definition.

\underline{\emph{Sharpness (lower bound).}}
Fix $\varepsilon>0$.
Choose $p_x\in(0,1)$ ($x\in\{0,1\}$) with
$\sqrt{p_1\overline{\sigma}_1^{(4)}}+\sqrt{p_0\overline{\sigma}_0^{(4)}}< \varepsilon$.
Let $\widetilde{Y}_1$ and $\widetilde{Y}_0$ be independent random variables which satisfy
\begin{equation}
\mathbb{P}(\widetilde{Y}_x=\pm a_x)=p_x/2,
\qquad
\mathbb{P}(\widetilde{Y}_x=0)=1-p_x,
\qquad\text{ where } \quad 
a_x\coloneqq \sqrt[4]{\overline{\sigma}_x^{(4)}/p_x}
\quad (x\in\{0,1\}).
\end{equation}
Then $\mathbb{E}[\widetilde{Y}_x]=0$ and $\mathbb{E}[\widetilde{Y}_x^4]=\overline{\sigma}_x^{(4)}$ hold for $x=0,1$.
Moreover, by independence and $\mathbb{E}[\widetilde{Y}_x]=0$,
\begin{equation}
\overline{\mu}^{(2)}=\mathbb{E}[(\widetilde{Y}_1-\widetilde{Y}_0)^2]
=\mathbb{E}[\widetilde{Y}_1^2]+\mathbb{E}[\widetilde{Y}_0^2]
=\sqrt{p_1\overline{\sigma}_1^{(4)}}+\sqrt{p_0\overline{\sigma}_0^{(4)}}
< \varepsilon.
\end{equation}
By Lemma \ref{lem:coupling_to_scm}, these are realized by an SCM.

\underline{\emph{Sharpness (upper bound).}}
We show that the upper bound is attainable. Take a random variable $S$ with
$\mathbb{P}(S=1)=\mathbb{P}(S=-1)=1/2$, and set
\begin{equation}
\widetilde{Y}_1\coloneqq \sqrt[4]{\overline{\sigma}_1^{(4)}} S,
\qquad
\widetilde{Y}_0\coloneqq -\sqrt[4]{\overline{\sigma}_0^{(4)}} S.
\end{equation}
Then $\mathbb{E}[\widetilde{Y}_x^4]=\overline{\sigma}_x^{(4)}$ and $\overline{\mu}^{(2)}=(\sqrt[4]{\overline{\sigma}_1^{(4)}}+\sqrt[4]{\overline{\sigma}_0^{(4)}})^2$.
By Lemma \ref{lem:coupling_to_scm}, these are realized by an SCM.
\end{proof}

\SecondfromComb*


\begin{proof}
(1): The inequality follows from Theorem \ref{thm:bound_of_var}.
In the following, we show the sharpness.

If $\overline{\sigma}_1^{(2)}\overline{\sigma}_0^{(2)}=0$, then the lower and upper endpoints in \eqref{eq:mu2_from_23} coincide, and the sharpness is clear.

Fix $\varepsilon>0$ and assume $\overline{\sigma}_1^{(2)}\overline{\sigma}_0^{(2)}>0$.
By Lemma~\ref{lem:two_point_23}, choose a random variable $Z_0$ with $\mathbb{E}[Z_0]=0$, $\mathbb{E}[Z_0^2]=1$ and $\mathbb{E}[Z_0^3]=\overline{\sigma}_0^{(3)}/(\overline{\sigma}_0^{(2)})^{3/2}$.
Set $\widetilde{Y}_0\coloneqq \sqrt{\overline{\sigma}_0^{(2)}}Z_0$.
Pick $\delta\in(0,1)$ with $2\delta\sqrt{\overline{\sigma}_1^{(2)}\overline{\sigma}_0^{(2)}} < \varepsilon$ and set $\alpha\coloneqq \sqrt{\overline{\sigma}_1^{(2)}}(1-\delta)$.

\underline{\emph{Sharpness (lower bound).}}
Take a random variable $W$ independent of $Z_0$ such that
$\mathbb{E}[W]=0$, $\mathbb{E}[W^2]=\overline{\sigma}_1^{(2)}-\alpha^2$ and
$\mathbb{E}[W^3]=\overline{\sigma}_1^{(3)}-\alpha^3\mathbb{E}[Z_0^3]$
(Lemma~\ref{lem:two_point_23}).
Let $\widetilde{Y}_1\coloneqq \alpha Z_0+W$.
Then $\mathbb{E}[\widetilde{Y}_1^2]=\overline{\sigma}_1^{(2)}$ and $\mathbb{E}[\widetilde{Y}_1^3]=\overline{\sigma}_1^{(3)}$.
Moreover, $\mathbb{E}[\widetilde{Y}_1\widetilde{Y}_0]=\alpha\sqrt{\overline{\sigma}_0^{(2)}}$, so
\begin{equation}
\overline{\mu}^{(2)}
=\overline{\sigma}_1^{(2)}+\overline{\sigma}_0^{(2)}-2\alpha\sqrt{\overline{\sigma}_0^{(2)}}
=(\sqrt{\overline{\sigma}_1^{(2)}}-\sqrt{\overline{\sigma}_0^{(2)}})^2+2\delta\sqrt{\overline{\sigma}_1^{(2)}\overline{\sigma}_0^{(2)}}.
\end{equation}
By the choice of $\delta$, we obtain $\overline{\mu}^{(2)} < (\sqrt{\overline{\sigma}_1^{(2)}}-\sqrt{\overline{\sigma}_0^{(2)}})^2+\varepsilon$.
By Lemma \ref{lem:coupling_to_scm}, these are realized by an SCM.

\underline{\emph{Sharpness (upper bound).}}
Take a random variable $W'$ independent of $Z_0$ such that
$\mathbb{E}[W']=0$, $\mathbb{E}[W'^2]=\overline{\sigma}_1^{(2)}-\alpha^2$ and
$\mathbb{E}[W'^3]=\overline{\sigma}_1^{(3)}+\alpha^3\mathbb{E}[Z_0^3]$
(Lemma~\ref{lem:two_point_23}).
Let $\widetilde{Y}_1\coloneqq -\alpha Z_0+W'$.
Then $\mathbb{E}[\widetilde{Y}_1^2]=\overline{\sigma}_1^{(2)}$ and $\mathbb{E}[\widetilde{Y}_1^3]=\overline{\sigma}_1^{(3)}$.
Moreover, $\mathbb{E}[\widetilde{Y}_1\widetilde{Y}_0]=-\alpha\sqrt{\overline{\sigma}_0^{(2)}}$, so
\begin{equation}
\overline{\mu}^{(2)}
=\overline{\sigma}_1^{(2)}+\overline{\sigma}_0^{(2)}+2\alpha\sqrt{\overline{\sigma}_0^{(2)}}
=(\sqrt{\overline{\sigma}_1^{(2)}}+\sqrt{\overline{\sigma}_0^{(2)}})^2-2\delta\sqrt{\overline{\sigma}_1^{(2)}\overline{\sigma}_0^{(2)}}.
\end{equation}
By the choice of $\delta$, we obtain $\overline{\mu}^{(2)} > (\sqrt{\overline{\sigma}_1^{(2)}}+\sqrt{\overline{\sigma}_0^{(2)}})^2-\varepsilon$.
By Lemma \ref{lem:coupling_to_scm}, these are realized by an SCM.


(2): The inequality follows from Theorem~\ref{thm:bound_of_var}: note that we have 
$\left(\sqrt{\overline{\sigma}_1^{(2)}}+\sqrt{\overline{\sigma}_0^{(2)}}\right)^2 \leq \left(\sqrt[4]{\overline{\sigma}_1^{(4)}}+\sqrt[4]{\overline{\sigma}_0^{(4)}}\right)^2$
since $\overline{\sigma}_x^{(2)}\leq \sqrt{\overline{\sigma}_x^{(4)}}$ holds by the Cauchy-Schwarz inequality.

\underline{\emph{Sharpness.}}
If $\overline{\sigma}_1^{(2)}\overline{\sigma}_0^{(2)}=0$, then the lower and upper endpoints in \eqref{eq:mu2_from_24} coincide, and $\overline{\mu}^{(2)}$ is forced to be this common endpoint. 
Thus  \eqref{eq:mu2_from_24} is sharp.

Fix $\varepsilon>0$ and assume $\overline{\sigma}_0^{(2)}\overline{\sigma}_1^{(2)}>0$.
Choose $\eta\in(0,1)$ such that $2\eta\sqrt{\overline{\sigma}_1^{(2)}\overline{\sigma}_0^{(2)}} < \varepsilon$.
Let $S$ be a random variable with $\mathbb{P}(S=1)=\mathbb{P}(S=-1)=1/2$.
For $x\in\{0,1\}$, define
\begin{equation}
m_x\coloneqq \frac{\overline{\sigma}_x^{(4)}-(1-\eta)(\overline{\sigma}_x^{(2)})^2}{\eta}.
\end{equation}
Let $U_x$ be a random variable with $\mathbb{P}\Big(U_x=\pm \sqrt{m_x/\overline{\sigma}_x^{(2)}}\Big)=(\overline{\sigma}_x^{(2)})^2\big/{2m_x}$ and $\mathbb{P}(U_x=0)=1-\big((\overline{\sigma}_x^{(2)})^2\big/{m_x}\big)$.
Then $\mathbb{E}[U_x]=0$, $\mathbb{E}[U_x^2]=\overline{\sigma}_x^{(2)}$, and $\mathbb{E}[U_x^4]=m_x$.
Assume $U_1$ and $U_0$ are independent of each other and independent of $S$.
For $s\in\{1,-1\}$, define
\begin{equation}
(\widetilde{Y}_1,\widetilde{Y}_0)=
\begin{cases}
(\sqrt{\overline{\sigma}_1^{(2)}}sS,\sqrt{\overline{\sigma}_0^{(2)}}S) & \text{with probability } 1-\eta,\\
(U_1,U_0) & \text{with probability } \eta.
\end{cases}
\end{equation}
Then $\mathbb{E}[\widetilde{Y}_x^2]=\overline{\sigma}_x^{(2)}$ and $\mathbb{E}[\widetilde{Y}_x^4]=\overline{\sigma}_x^{(4)}$, and $\mathbb{E}[\widetilde{Y}_1\widetilde{Y}_0]=(1-\eta)s\sqrt{\overline{\sigma}_1^{(2)}\overline{\sigma}_0^{(2)}}$.
Hence
\begin{equation}
\overline{\mu}^{(2)}=\overline{\sigma}_1^{(2)}+\overline{\sigma}_0^{(2)}-2(1-\eta)s\sqrt{\overline{\sigma}_1^{(2)}\overline{\sigma}_0^{(2)}}.
\end{equation}
For $s=1$ this is at most $(\sqrt{\overline{\sigma}_1^{(2)}}-\sqrt{\overline{\sigma}_0^{(2)}})^2+\varepsilon$, and for $s=-1$ it is at least $(\sqrt{\overline{\sigma}_1^{(2)}}+\sqrt{\overline{\sigma}_0^{(2)}})^2-\varepsilon$.
By Lemma \ref{lem:coupling_to_scm}, these are realized by an SCM.
Thus we have the sharpness since we can take $\varepsilon$ arbitrarily small.

(3): This follows from Theorem \ref{thm:mu2_from_4}.

(4): This follows by intersecting Theorems~\ref{thm:bound_of_var}--\ref{thm:mu2_from_4}.
By the proof of (2), this intersection coincides with the bound in Theorem~\ref{thm:bound_of_var}.
\end{proof}

\ThirdfromS*

\begin{proof}
(1): If $\overline{\sigma}_0^{(2)}=\overline{\sigma}_1^{(2)}=0$, then $\widetilde{Y}_0=\widetilde{Y}_1=0$ almost surely, and hence $\overline{\mu}^{(3)}=0$.

(2): By swapping the roles of $0$ and $1$ if necessary, we assume $\overline{\sigma}_1^{(2)}>0$.

Fix $M\in\mathbb{R}$.
By Lemma~\ref{lem:two_point_23}, choose $\widetilde{Y}_1$ such that $\mathbb{E}[\widetilde{Y}_1]=0$, $\mathbb{E}[\widetilde{Y}_1^2]=\overline{\sigma}_1^{(2)}$ and $\mathbb{E}[\widetilde{Y}_1^3]=M$.
Take a random variable $S$ with $\mathbb{P}(S=1)=\mathbb{P}(S=-1)=1/2$, independent of $\widetilde{Y}_1$, and set $\widetilde{Y}_0\coloneqq \sqrt{\overline{\sigma}_0^{(2)}}S$.
Then $\mathbb{E}[\widetilde{Y}_0]=0$, $\mathbb{E}[\widetilde{Y}_0^2]=\overline{\sigma}_0^{(2)}$ and $\mathbb{E}[\widetilde{Y}_0^3]=0$.
By independence and $\mathbb{E}[\widetilde{Y}_1]=\mathbb{E}[\widetilde{Y}_0]=0$,
\begin{equation}
\overline{\mu}^{(3)}
=\mathbb{E}[(\widetilde{Y}_1-\widetilde{Y}_0)^3]
=\mathbb{E}[\widetilde{Y}_1^3]
=M.
\end{equation}
Since $M$ is arbitrary, $\overline{\mu}^{(3)}$ can take any value in $\mathbb{R}$.
By Lemma~\ref{lem:coupling_to_scm}, these $(\widetilde{Y}_1,\widetilde{Y}_0)$ can be realized by some SCM.
Then we have the claim.
\end{proof}

\ThirdfromT*

\begin{proof}
Fix $M>0$.
Using Lemma~\ref{lem:two_point_23}, choose $W_1,W_0$ such that
$\mathbb{E}[W_x]=0$, $\mathbb{E}[W_x^2]=1$ and $\mathbb{E}[W_x^3]=2\overline{\sigma}_x^{(3)}$ for each $x\in\{0,1\}$,
and independent of each other.

Fix $c\in\mathbb{R}$.
Let $J$ be a random variable such that
$\mathbb{P}(J=0)=1/2$, 
$\mathbb{P}(J=j)=1/6$ $(j=1,2,3)$,
and independent of $(W_1,W_0)$.
We define
\begin{equation}
(\widetilde{Y}_1,\widetilde{Y}_0)\coloneqq
\begin{cases}
(W_1,W_0) & \text{if } J=0,\\
(c,-c) & \text{if } J=1,\\
(-c,0) & \text{if } J=2,\\
(0,c) & \text{if } J=3.
\end{cases}
\end{equation}
Then for each $x\in\{0,1\}$ we have $\mathbb{E}[\widetilde{Y}_x^3]=\overline{\sigma}_x^{(3)}$.
Moreover, since $W_1$ and $W_0$ are independent and $\mathbb{E}[W_1]=\mathbb{E}[W_0]=0$,
\begin{equation}
\mathbb{E}[(W_1-W_0)^3]=\mathbb{E}[W_1^3]-\mathbb{E}[W_0^3]=2(\overline{\sigma}_1^{(3)}-\overline{\sigma}_0^{(3)}).
\end{equation}
Therefore,
\begin{equation}
\overline{\mu}^{(3)}
=\mathbb{E}[(\widetilde{Y}_1-\widetilde{Y}_0)^3]
=(\overline{\sigma}_1^{(3)}-\overline{\sigma}_0^{(3)})+c^3.
\end{equation}
Taking $c$ with $|c|$ large yields $\overline{\mu}^{(3)}>M$ and $\overline{\mu}^{(3)}<-M$.
By Lemma~\ref{lem:coupling_to_scm}, these are realized by some SCMs.
\end{proof}

\ThirdfromF*

\begin{proof}
By Lemma~\ref{lem:sharp_third_from_fourth_newextra} applied to $\widetilde{Y}_1-\widetilde{Y}_0$ and Minkowski's inequality,
\begin{equation}
\bigl|\overline{\mu}^{(3)}\bigr| 
= \bigl|\mathbb{E}[(\widetilde{Y}_1-\widetilde{Y}_0)^3]\bigr|
\leq \frac{\sqrt{2}}{\sqrt[4]{27}}\bigl\|\widetilde{Y}_1-\widetilde{Y}_0\bigr\|_4^3
\leq \frac{\sqrt{2}}{\sqrt[4]{27}}\bigl(\|\widetilde{Y}_1\|_4+\|\widetilde{Y}_0\|_4\bigr)^3
= \frac{\sqrt{2}}{\sqrt[4]{27}} \Bigl(\sqrt[4]{\overline{\sigma}_1^{(4)}}+\sqrt[4]{\overline{\sigma}_0^{(4)}}\Bigr)^3.
\end{equation}

\underline{\emph{Sharpness.}}
Set $Z_0$ as in Lemma~\ref{lem:sharp_third_from_fourth_newextra}, then taking
$\widetilde{Y}_1=\sqrt[4]{\overline{\sigma}_1^{(4)}}\,Z_0$ and
$\widetilde{Y}_0=-\sqrt[4]{\overline{\sigma}_0^{(4)}}\,Z_0$, we have 
$\overline{\mu}^{(3)} 
=\frac{\sqrt{2}}{\sqrt[4]{27}} \Bigl(\sqrt[4]{\overline{\sigma}_1^{(4)}}+\sqrt[4]{\overline{\sigma}_0^{(4)}}\Bigr)^3$.
Replacing $Z_0$ by $-Z_0$ attains the lower bound.
By Lemma~\ref{lem:coupling_to_scm}, these can be realized by an SCM.
\end{proof}

\ThirdfromComb*

\begin{proof}
(1)-(ii): If $\overline{\sigma}_x^{(2)}=0$ for some $x\in\{0,1\}$, then $\widetilde{Y}_x=0$ almost surely and thus $\overline{\sigma}_x^{(3)}=0$.
If $x=1$, then $\widetilde{\Delta}=-\widetilde{Y}_0$ and $\overline{\mu}^{(3)}=-\overline{\sigma}_0^{(3)}$.
If $x=0$, then $\widetilde{\Delta}=\widetilde{Y}_1$ and $\overline{\mu}^{(3)}=\overline{\sigma}_1^{(3)}$.

(1)-(i): Assume $\overline{\sigma}_0^{(2)}\overline{\sigma}_1^{(2)}>0$ and fix $M>0$.
Choose $\eta\in(0,\min\{\overline{\sigma}_0^{(2)},\overline{\sigma}_1^{(2)}\})$.

\underline{\emph{Sharpness (upper bound).}}
Set $a>0$ with $p\coloneqq 2\eta/(3a^2)<1/3$ and $a>{\big(M-(\overline{\sigma}_1^{(3)}-\overline{\sigma}_0^{(3)})\big)}\big/{3\eta}$.
Take independent random variables $Z_1,Z_0$ with
\begin{equation}
\mathbb{E}[Z_1]=\mathbb{E}[Z_0]=0, \qquad
\mathbb{E}[Z_x^2]=\frac{\overline{\sigma}_x^{(2)}-\eta}{1-3p}\quad(x\in \{0,1\}),\quad
\mathbb{E}[Z_1^3]=\frac{\overline{\sigma}_1^{(3)}-\eta a/2}{1-3p},\qquad
\mathbb{E}[Z_0^3]=\frac{\overline{\sigma}_0^{(3)}+\eta a/2}{1-3p},
\end{equation}
which is possible by Lemma~\ref{lem:two_point_23}.
Define $(\widetilde{Y}_1,\widetilde{Y}_0)$ by
\begin{equation}
(\widetilde{Y}_1,\widetilde{Y}_0)=
\begin{cases}
(Z_1,Z_0) & \text{with probability } 1-3p,\\
(a,-a) & \text{with probability } p,\\
(-a/2,a/2) & \text{with probability } 2p.
\end{cases}
\end{equation}
Then $\mathbb{E}[\widetilde{Y}_x^2]=\overline{\sigma}_x^{(2)}$ and $\mathbb{E}[\widetilde{Y}_x^3]=\overline{\sigma}_x^{(3)}$ for $x\in\{0,1\}$.
Moreover, since $Z_1,Z_0$ are independent and $\mathbb{E}[Z_1]=\mathbb{E}[Z_0]=0$,
\begin{equation}
\overline{\mu}^{(3)}
=\mathbb{E}[(\widetilde{Y}_1-\widetilde{Y}_0)^3]
=(1-3p)(\mathbb{E}[Z_1^3]-\mathbb{E}[Z_0^3])+p(2a)^3+2p(-a)^3
=\overline{\sigma}_1^{(3)}-\overline{\sigma}_0^{(3)}+3\eta a
>M.
\end{equation}
Lemma~\ref{lem:coupling_to_scm} implies the existence of an SCM realizing this.

\underline{\emph{Sharpness (lower bound).}}
Set $a>0$ with $p\coloneqq \eta/(3a^2)<2/3$ and $a>\big(M+(\overline{\sigma}_1^{(3)}-\overline{\sigma}_0^{(3)})\big)\big/6\eta$.
Take independent random variables $Z_1,Z_0$ with
\begin{equation}
\mathbb{E}[Z_1]=\mathbb{E}[Z_0]=0, \qquad
\mathbb{E}[Z_x^2]=\frac{\overline{\sigma}_x^{(2)}-\eta}{1-3p/2}\quad(x\in \{0,1\}),\quad
\mathbb{E}[Z_1^3]=\frac{\overline{\sigma}_1^{(3)}+\eta a}{1-3p/2},\qquad
\mathbb{E}[Z_0^3]=\frac{\overline{\sigma}_0^{(3)}-\eta a}{1-3p/2},
\end{equation}
which is possible by Lemma~\ref{lem:two_point_23}.
Define $(\widetilde{Y}_1,\widetilde{Y}_0)$ by
\begin{equation}
(\widetilde{Y}_1,\widetilde{Y}_0)=
\begin{cases}
(Z_1,Z_0) & \text{with probability } 1-3p/2,\\
(a,-a) & \text{with probability } p,\\
(-2a,2a) & \text{with probability } p/2.
\end{cases}
\end{equation}
Then $\mathbb{E}[\widetilde{Y}_x^2]=\overline{\sigma}_x^{(2)}$ and $\mathbb{E}[\widetilde{Y}_x^3]=\overline{\sigma}_x^{(3)}$ for $x\in\{0,1\}$.
Moreover, since $Z_1,Z_0$ are independent and $\mathbb{E}[Z_1]=\mathbb{E}[Z_0]=0$,
\begin{equation}
\overline{\mu}^{(3)}
=\mathbb{E}[(\widetilde{Y}_1-\widetilde{Y}_0)^3]
=(1-3p/2)(\mathbb{E}[Z_1^3]-\mathbb{E}[Z_0^3])+p(2a)^3+(p/2)(-4a)^3
=\overline{\sigma}_1^{(3)}-\overline{\sigma}_0^{(3)}-6\eta a
<-M.
\end{equation}
Lemma~\ref{lem:coupling_to_scm} implies the existence of an SCM realizing this.

(2) follows from Theorem~\ref{thm:mu3_from_4}.
\end{proof}

\FourthfromS*

\begin{proof}
(1): If $\overline{\sigma}_1^{(2)}=\overline{\sigma}_0^{(2)}=0$, we have $\widetilde{Y}_1=0$ almost surely and $\widetilde{Y}_0=0$ almost surely.
Therefore we have $\overline{\mu}^{(4)}=\mathbb{E}[(\widetilde{Y}_1-\widetilde{Y}_0)^4]=0$.

(2):
By the Cauchy--Schwarz inequality, $\overline{\mu}^{(4)}=\mathbb{E}[\widetilde{\Delta}^4]\geq (\mathbb{E}[\widetilde{\Delta}^2])^2$.
Let $\|Z\|_2 \coloneqq (\mathbb{E}[Z^2])^{1/2}$.
By the triangle inequality in $L^2$, $\|\widetilde{\Delta}\|_2 \geq |\|\widetilde{Y}_1\|_2-\|\widetilde{Y}_0\|_2|$.
Since $\|\widetilde{Y}_x\|_2=\sqrt{\overline{\sigma}_x^{(2)}}$, we have $\overline{\mu}^{(4)}\geq \Big(\sqrt{\overline{\sigma}_1^{(2)}}-\sqrt{\overline{\sigma}_0^{(2)}}\Big)^4$.

\underline{\emph{Sharpness (lower bound).}}
Take a random variable $S\in\{-1,1\}$ with $\mathbb{P}(S=1)=\mathbb{P}(S=-1)=1/2$ and set $\widetilde{Y}_x \coloneqq \sqrt{\overline{\sigma}_x^{(2)}}S$.
Then $\mathbb{E}[\widetilde{Y}_x]=0$, $\mathbb{E}[\widetilde{Y}_x^2]=\overline{\sigma}_x^{(2)}$ and $\widetilde{\Delta}=(\sqrt{\overline{\sigma}_1^{(2)}}-\sqrt{\overline{\sigma}_0^{(2)}})S$, so the lower bound is attained.

\underline{\emph{Sharpness (upper bound).}}
Fix $M>0$.
We assume $\overline{\sigma}_1^{(2)}>0$, by exchanging the roles of $0$ and $1$ if necessary.
Choose a constant $A>\sqrt{\overline{\sigma}_1^{(2)}}$ such that $(\overline{\sigma}_1^{(2)}/A^2)(A-\sqrt{\overline{\sigma}_0^{(2)}})^4>M$.

Let $I$ be a random variable with $\mathbb{P}(I=1)=\overline{\sigma}_1^{(2)}\big/{A^2}$ and $\mathbb{P}(I=0)=1-(\overline{\sigma}_1^{(2)}\big/A^2)$,
and let $S$ be a random variable with $\mathbb{P}(S=1)=\mathbb{P}(S=-1)=1/2$, independent of $I$.
Define $\widetilde{Y}_1 \coloneqq IAS$ and $\widetilde{Y}_0 \coloneqq \sqrt{\overline{\sigma}_0^{(2)}}S$.
Then $\mathbb{E}[\widetilde{Y}_x^2]=\overline{\sigma}_x^{(2)}$ ($x\in\{0,1\}$) and
$\overline{\mu}^{(4)}=\mathbb{E}[(\widetilde{Y}_1-\widetilde{Y}_0)^4]\geq \mathbb{P}(I=1)(A-\sqrt{\overline{\sigma}_0^{(2)}})^4>M$.
By Lemma~\ref{lem:coupling_to_scm}, these can be realized by an SCM.
\end{proof}

\FourthfromT*

\begin{proof}
The inequality is trivial.

\underline{\emph{Sharpness (lower bound).}}
Fix $\varepsilon>0$.
Choose a random variable $B$ with $\mathbb{E}[B]=0$ and $\mathbb{E}[B^3]=2\overline{\sigma}_0^{(3)}$.
Take a constant $M_0>0$ such that $M_0^2>\mathbb{E}[B^2]$ and $|\delta|<\sqrt[4]{\varepsilon}$, where $\delta \coloneqq \dfrac{2(\overline{\sigma}_1^{(3)}-\overline{\sigma}_0^{(3)})}{3(M_0^2-\mathbb{E}[B^2])}$.
Then define $\widetilde{Y}_0$ by
\begin{equation}
\widetilde{Y}_0 \coloneqq
\begin{cases}
B & \text{with probability\ } 1/2,\\
M_0 & \text{with probability\ } 1/4,\\
-M_0 & \text{with probability\ } 1/4
\end{cases}
\end{equation}
and define $\widetilde{Y}_1$ by
\begin{equation}
\widetilde{Y}_1 \coloneqq
\begin{cases}
\widetilde{Y}_0-\delta & \text{if } \widetilde{Y}_0=B,\\
\widetilde{Y}_0+\delta & \text{if } \widetilde{Y}_0\in\{M_0,-M_0\}
\end{cases}.
\end{equation}
Then $\mathbb{E}[\widetilde{Y}_0]=\mathbb{E}[\widetilde{\Delta}]=0$, hence $\mathbb{E}[\widetilde{Y}_1]=0$.
Moreover $\widetilde{\Delta}^2\equiv \delta^2$ and $\mathbb{E}[\widetilde{\Delta}^3]=0$, so
\begin{equation}
\mathbb{E}[\widetilde{Y}_1^3]=\mathbb{E}[\widetilde{Y}_0^3]+3\mathbb{E}[\widetilde{Y}_0^2\widetilde{\Delta}].
\end{equation}
Since $\mathbb{E}[\widetilde{Y}_0^3]=\overline{\sigma}_0^{(3)}$ and
$\mathbb{E}[\widetilde{Y}_0^2\widetilde{\Delta}]
=-\frac{1}{2} \mathbb{E}[B^2]\delta+\frac{1}{2} M_0^2\delta
=\frac{\delta}{2}(M_0^2-\mathbb{E}[B^2])$,
we have $\mathbb{E}[\widetilde{Y}_1^3]=\overline{\sigma}_1^{(3)}$.
We also see $\overline{\mu}^{(4)}=\mathbb{E}[\widetilde{\Delta}^4]=\delta^4<\varepsilon$.

\underline{\emph{Sharpness (upper bound).}}
Fix $M>0$.
Let $B_0,B_1$ be random variables with $\mathbb{E}[B_x]=0,\,\mathbb{E}[B_x^3]=2\overline{\sigma}_x^{(3)}$ ($x\in\{0,1\}$), independent of each other.
Take a constant $A>0$ and define $(\widetilde{Y}_1,\widetilde{Y}_0)$ by
\begin{equation}
(\widetilde{Y}_1,\widetilde{Y}_0)\coloneqq
\begin{cases}
(B_1,B_0) & \text{with probability\ } 1/2,\\
(A,-A) & \text{with probability\ } 1/4,\\
(-A,A) & \text{with probability\ } 1/4.
\end{cases}
\end{equation}
Then $\mathbb{E}[\widetilde{Y}_1^3]=\overline{\sigma}_1^{(3)}$ and $\mathbb{E}[\widetilde{Y}_0^3]=\overline{\sigma}_0^{(3)}$.
We have $\widetilde{\Delta}=\widetilde{Y}_1-\widetilde{Y}_0=\pm 2A$ when $(\widetilde{Y}_1,\widetilde{Y}_0)=(\pm A,\mp A)$, and we have
$\overline{\mu}^{(4)}=\mathbb{E}[\widetilde{\Delta}^4]\geq \frac{1}{2} (2A)^4=8A^4$.
We obtain $\overline{\mu}^{(4)}>M$ by taking $A>\sqrt[4]{M/8}$.

By Lemma~\ref{lem:coupling_to_scm}, these distributions can be realized by SCMs. Therefore we have the sharpness.
\end{proof}

\FourthfromF*

\begin{proof}
Recall that $\widetilde{\Delta}=\widetilde{Y}_1-\widetilde{Y}_0$.
Since $4$ is even, we have
\begin{equation}
\|\widetilde{Y}_x\|_4=\bigl(\mathbb{E}[|\widetilde{Y}_x|^4]\bigr)^{1/4}=\sqrt[4]{\overline{\sigma}_x^{(4)}}
\quad (x\in\{0,1\}),
\qquad
\|\widetilde{\Delta}\|_4^4=\mathbb{E}[|\widetilde{\Delta}|^4]=\overline{\mu}^{(4)}.
\end{equation}
By Minkowski's inequality in $L^4$-space,
\begin{equation}
\|\widetilde{\Delta}\|_4\leq \|\widetilde{Y}_1\|_4+\|\widetilde{Y}_0\|_4,
\qquad
\|\widetilde{Y}_1\|_4\leq \|\widetilde{\Delta}\|_4+\|\widetilde{Y}_0\|_4,
\quad
\text{and}
\quad
\|\widetilde{Y}_0\|_4\leq \|\widetilde{Y}_1\|_4+\|\widetilde{\Delta}\|_4
\end{equation}
hold.
Thus we have 
$\bigl|\|\widetilde{Y}_1\|_4-\|\widetilde{Y}_0\|_4\bigr|\le \|\widetilde{\Delta}\|_4\le \|\widetilde{Y}_1\|_4+\|\widetilde{Y}_0\|_4$.
Raising both sides to the fourth power yields \eqref{eq:mu4_from_4}.

\underline{\emph{Sharpness.}}
Take a random variable $S\in\{-1,1\}$ with $\mathbb{P}(S=1)=\mathbb{P}(S=-1)=1/2$.
Setting $\widetilde{Y}_1 \coloneqq \sqrt[4]{\overline{\sigma}_1^{(4)}}S$ and $\widetilde{Y}_0 \coloneqq -\sqrt[4]{\overline{\sigma}_0^{(4)}}S$ attains the upper bound.
Setting $\widetilde{Y}_1 \coloneqq \sqrt[4]{\overline{\sigma}_1^{(4)}}S$ and $\widetilde{Y}_0 \coloneqq \sqrt[4]{\overline{\sigma}_0^{(4)}}S$ attains the lower bound.
By Lemma~\ref{lem:coupling_to_scm}, these can be realized by an SCM.
\end{proof}

\FourthfromComb*

\begin{proof} 
(1)-(i) follows from Theorem \ref{thm:mu4_from_2}(1).

(1)-(ii): By the Cauchy--Schwarz inequality, $\mathbb{E}[\widetilde{\Delta}^4]\geq \mathbb{E}[\widetilde{\Delta}^2]^2$, hence $\overline{\mu}^{(4)}\geq (\overline{\mu}^{(2)})^2$.
Then Theorem~\ref{thm:bound_of_var} gives $\overline{\mu}^{(2)}\geq (\sqrt{\overline{\sigma}_1^{(2)}}-\sqrt{\overline{\sigma}_0^{(2)}})^2$.

\underline{\emph{Sharpness (lower bound).}}
Fix $\varepsilon>0$.
Choose $\eta\in(0,\frac{1}{4})$ and let $S$ be a random variable with $\mathbb{P}(S=1)=\mathbb{P}(S=-1)=1/2$.

If $\overline{\sigma}_0^{(3)}\ne 0$, 
let $U$ be a random variable $\mathbb{E}[U]=0$, $\mathbb{E}[U^2]=\eta^{-2/3}$, and $\mathbb{E}[U^3]=\overline{\sigma}_0^{(3)}/\eta$ (Lemma~\ref{lem:two_point_23}).
Otherwise set $U\coloneqq 0$.


Let $W$ be a random variable with
\begin{equation}
    \mathbb{P}(W=\eta^{-7/16})=\mathbb{P}(W=-\eta^{-7/16})=\mathbb{P}\left(W=\dfrac{\eta^{-7/16}}{2}\right)=\mathbb{P}\left(W=-\dfrac{\eta^{-7/16}}{2}\right)=\frac{1}{4}
\end{equation}
and define a random variable $H$ by
\begin{equation}
    H=\begin{cases}
        \dfrac{8(\overline{\sigma}_1^{(3)}-\overline{\sigma}_0^{(3)})}{9\eta a^2} & \text{ if } |W|=\eta^{-7/16}, \\
        -\dfrac{8(\overline{\sigma}_1^{(3)}-\overline{\sigma}_0^{(3)})}{9\eta a^2} & \text{ if } |W|=\dfrac{\eta^{-7/16}}{2},
    \end{cases}
\end{equation}
where $a\coloneqq \eta ^{-7/16}$.
Then we have $\mathbb{E}[W]=\mathbb{E}[H]=0$ and, since $\mathbb{E}[W^3]=\mathbb{E}[H^3]=\mathbb{E}[WH^2]=0$,
\begin{equation}
\mathbb{E}[(W+H)^3]=3\mathbb{E}[W^2H]
=3\left(\frac12 a^2\cdot \frac{8(\overline{\sigma}_1^{(3)}-\overline{\sigma}_0^{(3)})}{9\eta a^2}
-\frac12\left(\frac{a}{2}\right)^2\cdot \frac{8(\overline{\sigma}_1^{(3)}-\overline{\sigma}_0^{(3)})}{9\eta a^2}\right)
=\frac{\overline{\sigma}_1^{(3)}-\overline{\sigma}_0^{(3)}}{\eta}.
\end{equation}

Define
\begin{equation}
c_0^2\coloneqq \frac{\overline{\sigma}_0^{(2)}-\eta\mathbb{E}[U^2]-\eta\mathbb{E}[W^2]}{1-2\eta},\qquad
c_1^2\coloneqq \frac{\overline{\sigma}_1^{(2)}-\eta\mathbb{E}[U^2]-\eta\mathbb{E}[(W+H)^2]}{1-2\eta}.
\end{equation}
For $\eta$ small enough, $c_0^2>0$ and $c_1^2>0$.
Assume $(U,W,S)$ are independent and define $(\widetilde{Y}_1,\widetilde{Y}_0)$ by
\begin{equation}
(\widetilde{Y}_1,\widetilde{Y}_0)=
\begin{cases}
(U,U) & \text{with probability } \eta,\\
(W+H,W) & \text{with probability } \eta,\\
(c_1 S,c_0 S) & \text{with probability } 1-2\eta.
\end{cases}
\end{equation}
Then $\mathbb{E}[\widetilde{Y}_x^2]=\overline{\sigma}_x^{(2)}$ for $x\in\{0,1\}$ by the definition of $c_0^2,c_1^2$.
Moreover, $\mathbb{E}[\widetilde{Y}_0^3]=\eta\mathbb{E}[U^3]=\overline{\sigma}_0^{(3)}$ and $\mathbb{E}[\widetilde{Y}_1^3]=\eta\mathbb{E}[U^3]+\eta\mathbb{E}[(W+H)^3]=\overline{\sigma}_1^{(3)}$.

In addition, we have
\begin{equation}
\overline{\mu}^{(4)}=\mathbb{E}[\widetilde{\Delta}^4]=(1-2\eta)(c_1-c_0)^4+\eta\mathbb{E}[H^4].
\end{equation}

If $\overline{\sigma}_0^{(3)}\ne0$, then $\eta\mathbb{E}[U^2]=\eta^{1/3}$, while if $\overline{\sigma}_0^{(3)}=0$, then $\eta\mathbb{E}[U^2]=0$. Moreover, we have
\begin{equation}
\eta\mathbb{E}[W^2]=\frac{5}{8}\eta^{1/8},\qquad
\eta\mathbb{E}[H^2]=\frac{64}{81}(\overline{\sigma}_1^{(3)}-\overline{\sigma}_0^{(3)})^2\eta^{3/4},\qquad  
\eta\mathbb{E}[H^4]=\frac{4096}{6561}(\overline{\sigma}_1^{(3)}-\overline{\sigma}_0^{(3)})^4\eta^{1/2}.
\end{equation}
Since $\eta\mathbb{E}[(W+H)^2]\leq 2\eta\mathbb{E}[W^2]+2\eta\mathbb{E}[H^2]$, we have $c_0^2\to\overline{\sigma}_0^{(2)}$, $c_1^2\to\overline{\sigma}_1^{(2)}$, and $\overline{\mu}^{(4)}\to(\sqrt{\overline{\sigma}_1^{(2)}}-\sqrt{\overline{\sigma}_0^{(2)}})^4$ as $\eta\to0$.
Therefore, by choosing $\eta$ small enough, we have $\overline{\mu}^{(4)}<(\sqrt{\overline{\sigma}_1^{(2)}}-\sqrt{\overline{\sigma}_0^{(2)}})^4+\varepsilon$.

\underline{\emph{Sharpness (upper bound).}}
Fix $M>0$.
Choose $b>0$ such that $b>M/16$ and $b^{-1}<\min\{\overline{\sigma}_0^{(2)},\overline{\sigma}_1^{(2)}\}$, and set $p\coloneqq b^{-3}$.
Let $Z_x$ be a random variable with
\begin{equation}
\mathbb{E}[Z_x]=0,\qquad 
\mathbb{E}[Z_x^2]=\frac{\overline{\sigma}_x^{(2)}-p b^2}{1-p},\qquad
\mathbb{E}[Z_x^3]=\frac{\overline{\sigma}_x^{(3)}}{1-p},
\end{equation}
for each $x\in\{0,1\}$ (Lemma~\ref{lem:two_point_23}), and assume $Z_1,Z_0$ are independent.
Define $(\widetilde{Y}_1,\widetilde{Y}_0)$ by
\begin{equation}
(\widetilde{Y}_1,\widetilde{Y}_0)=
\begin{cases}
(Z_1,Z_0) & \text{with probability } 1-p,\\
(b,-b) & \text{with probability } p/2,\\
(-b,b) & \text{with probability } p/2.
\end{cases}
\end{equation}
Then $\mathbb{E}[\widetilde{Y}_x^2]=\overline{\sigma}_x^{(2)}$ and $\mathbb{E}[\widetilde{Y}_x^3]=\overline{\sigma}_x^{(3)}$ for $x\in\{0,1\}$.
Moreover, we obtain $\overline{\mu}^{(4)}\geq p(2b)^4=16b>M$.

(1)-(iii): Since $\overline{\sigma}_0^{(2)}=0$, $\widetilde{Y}_0=0$ almost surely and $\overline{\sigma}_0^{(3)}=0$. 
Since $\overline{\sigma}_1^{(3)}=\mathbb{E}[\widetilde{Y}_1(\widetilde{Y}_1^2-\overline{\sigma}_1^{(2)})]$, we have 
\begin{equation}
(\overline{\sigma}_1^{(3)})^2\leq \overline{\sigma}_1^{(2)}(\overline{\mu}^{(4)}-(\overline{\sigma}_1^{(2)})^2)
\end{equation}
by the Cauchy--Schwarz inequality.
Thus $\overline{\mu}^{(4)}\geq (\overline{\sigma}_1^{(2)})^2+(\overline{\sigma}_1^{(3)})^2/\overline{\sigma}_1^{(2)}$. 

\underline{\emph{Sharpness (lower bound).}}
Let $Z$ be a random variable with $\mathbb{E}[Z]=0$, $\mathbb{E}[Z^2]=\overline{\sigma}_1^{(2)}$ and $\mathbb{E}[Z^3]=\overline{\sigma}_1^{(3)}$ (Lemma \ref{lem:two_point_23}). 
Then $Z^2=(\overline{\sigma}_1^{(3)}/\overline{\sigma}_1^{(2)})Z+\overline{\sigma}_1^{(2)}$ almost surely, and hence $\mathbb{E}[Z^4]=(\overline{\sigma}_1^{(2)})^2+(\overline{\sigma}_1^{(3)})^2/\overline{\sigma}_1^{(2)}$. Taking $\widetilde{Y}_1=Z$ and $\widetilde{Y}_0=0$ attains the lower bound.

\underline{\emph{Sharpness (upper bound).}}
Fix $M>0$. Choose $b>0$ such that $b>M$ and $b^{-1}<\overline{\sigma}_1^{(2)}$, and set $p\coloneqq b^{-3}$. 
Take a random variable $Z$ with 
\begin{equation}
    \mathbb{E}[Z]=0,\qquad \mathbb{E}[Z^2]=\frac{\overline{\sigma}_1^{(2)}-pb^2}{1-p}, \qquad \mathbb{E}[Z^3]=\frac{\overline{\sigma}_1^{(3)}}{1-p}
\end{equation}
(Lemma \ref{lem:two_point_23}). 
Define $\widetilde{Y}_0=0$ almost surely and let 
\begin{equation}
\widetilde{Y}_1=\begin{cases}
    Z & \text{with probability } 1-p, \\
    b & \text{with probability } p/2,\\
    -b & \text{with probability } p/2.
\end{cases}
\end{equation}
Then $\mathbb{E}[\widetilde{Y}_1^2]=\overline{\sigma}_1^{(2)}$, $\mathbb{E}[\widetilde{Y}_1^3]=\overline{\sigma}_1^{(3)}$, and $\overline{\mu}^{(4)}\geq pb^4=b>M$.

(1)-(iv): The proof is identical after interchanging the roles of $1$ and $0$. 

By Lemma~\ref{lem:coupling_to_scm}, these are realized by some SCMs. Thus the bounds in (1) are sharp.

(2), (3), and (4) follow from Theorems~\ref{thm:mu4_from_2}--\ref{thm:mu4_from_4}.
Lemma~\ref{lem:coupling_to_scm} yields SCMs.
\end{proof}

\Minkowskibound*

\begin{proof}
Recall that $\widetilde{\Delta}=\widetilde{Y}_1-\widetilde{Y}_0$.
Since $m$ is even, we have
\begin{equation}
\|\widetilde{Y}_x\|_m=\bigl(\mathbb{E}[|\widetilde{Y}_x|^m]\bigr)^{1/m}=\sqrt[m]{\overline{\sigma}_x^{(m)}}
\quad (x\in\{0,1\}),
\qquad
\|\widetilde{\Delta}\|_m^m=\mathbb{E}[|\widetilde{\Delta}|^m]=\overline{\mu}^{(m)}.
\end{equation}
By Minkowski's inequality in $L^m$-space,
\begin{equation}
\|\widetilde{\Delta}\|_m\leq \|\widetilde{Y}_1\|_m+\|\widetilde{Y}_0\|_m,
\qquad
\|\widetilde{Y}_1\|_m\leq \|\widetilde{\Delta}\|_m+\|\widetilde{Y}_0\|_m,
\quad
\text{and}
\quad
\|\widetilde{Y}_0\|_m\leq \|\widetilde{Y}_1\|_m+\|\widetilde{\Delta}\|_m
\end{equation}
hold.
Thus we have 
$\bigl|\|\widetilde{Y}_1\|_m-\|\widetilde{Y}_0\|_m\bigr|\le \|\widetilde{\Delta}\|_m\le \|\widetilde{Y}_1\|_m+\|\widetilde{Y}_0\|_m$.
Raising both sides to the $m$-th power yields \eqref{eq:minkowski_bound}.

\underline{\emph{Sharpness.}}
Take a random variable $S\in\{-1,1\}$ with $\mathbb{P}(S=1)=\mathbb{P}(S=-1)=1/2$.
Setting $\widetilde{Y}_1 \coloneqq \sqrt[m]{\overline{\sigma}_1^{(m)}}S$ and $\widetilde{Y}_0 \coloneqq -\sqrt[m]{\overline{\sigma}_0^{(m)}}S$ attains the upper bound.
Setting $\widetilde{Y}_1 \coloneqq \sqrt[m]{\overline{\sigma}_1^{(m)}}S$ and $\widetilde{Y}_0 \coloneqq \sqrt[m]{\overline{\sigma}_0^{(m)}}S$ attains the lower bound.
By Lemma~\ref{lem:coupling_to_scm}, these can be realized by an SCM.
\end{proof}

\EvenfromSingle*

\begin{proof}
(1): The inequality $0 \leq \overline{\mu}^{(m)} \leq \infty$ follows from the definition.

\underline{\emph{Sharpness (lower bound).}}
Fix $\varepsilon>0$.
Choose a random variable $B$ with $\mathbb{E}[B]=0$ and $\mathbb{E}[B^k]=2\overline{\sigma}_0^{(k)}$ (Lemma \ref{lem:exist_odd_moment}).
For $a,\delta\in \mathbb{R}$, define a random variable $W_a$ by
\begin{equation}
W_a \coloneqq
\begin{cases}
B & \text{with probability\ } 1/2,\\
a & \text{with probability\ } 1/4,\\
-a & \text{with probability\ } 1/4
\end{cases}
\end{equation}
and define $V_{a,\delta}$ by
\begin{equation}
V_{a,\delta} \coloneqq
\begin{cases}
W_a-\delta & \text{if } W_a=B,\\
W_a+\delta & \text{if } W_a\in\{a,-a\}
\end{cases}.
\end{equation}
Then $\mathbb{E}[W_a]=\mathbb{E}[V_{a,\delta}]=0$ and, since $k$ is odd, $\mathbb{E}[W_a^k]=\overline{\sigma}_0^{(k)}$.
Moreover $V_{a,\delta}-W_a=\pm\delta$, hence $\mathbb{E}[(V_{a,\delta}-W_a)^m]=|\delta|^m$.
Let $F_a(\delta)\coloneqq \mathbb{E}[V_{a,\delta}^k]$.
Then
\begin{equation}
F_a(\delta)=\frac{1}{2}\mathbb{E}[(B-\delta)^k]+\frac{1}{4}(a+\delta)^k+\frac{1}{4}(-a+\delta)^k
=\frac{1}{2}\mathbb{E}[(B-\delta)^k]+\frac{1}{4}(a+\delta)^k-\frac{1}{4}(a-\delta)^k,
\end{equation}
so $F_a(\delta)\to\infty$ and $F_a(-\delta)\to-\infty$ as $a\to\infty$.

Hence, for some $a_0>0$ we have $F_{a_0}(-\sqrt[m]{\varepsilon})\leq \overline{\sigma}_1^{(k)}\leq F_{a_0}(\sqrt[m]{\varepsilon})$.
Since $F_{a_0}$ is continuous, there exists $\delta_0\in[-\sqrt[m]{\varepsilon},\sqrt[m]{\varepsilon}]$ such that $F_{a_0}(\delta_0)=\overline{\sigma}_1^{(k)}$.
With $(\widetilde{Y}_1,\widetilde{Y}_0) \coloneqq (V_{a_0,\delta_0},W_{a_0})$, we have $\mathbb{E}[\widetilde{Y}_x^k]=\overline{\sigma}_x^{(k)}$ and $\overline{\mu}^{(m)}=|\delta_0|^m\leq \varepsilon$.

\underline{\emph{Sharpness (upper bound).}}
Fix $M>0$.
Let $B_0,B_1$ be random variables with $\mathbb{E}[B_x]=0,\,\mathbb{E}[B_x^k]=2\overline{\sigma}_x^{(k)}$ ($x\in\{0,1\}$) (Lemma \ref{lem:exist_odd_moment}).
Take a constant $A>0$ with $\frac{1}{2}(2A)^m>M$, and define $(\widetilde{Y}_1,\widetilde{Y}_0)$ by
\begin{equation}
(\widetilde{Y}_1,\widetilde{Y}_0)\coloneqq
\begin{cases}
(B_1,B_0) & \text{with probability\ } 1/2,\\
(A,-A) & \text{with probability\ } 1/4,\\
(-A,A) & \text{with probability\ } 1/4.
\end{cases}
\end{equation}
Then $\mathbb{E}[\widetilde{Y}_1^k]=\overline{\sigma}_1^{(k)}$ and $\mathbb{E}[\widetilde{Y}_0^k]=\overline{\sigma}_0^{(k)}$.
We have $\widetilde{\Delta}=\widetilde{Y}_1-\widetilde{Y}_0=\pm 2A$ when $(\widetilde{Y}_1,\widetilde{Y}_0)=(\pm A,\mp A)$, and we have
$\overline{\mu}^{(m)}=\mathbb{E}[\widetilde{\Delta}^m]\geq \frac{1}{2} (2A)^m>M$.
Taking $A$ large yields $\overline{\mu}^{(m)}>M$.

By Lemma~\ref{lem:coupling_to_scm}, these distributions can be realized by SCMs.

(2)-(i): If $\overline{\sigma}_1^{(k)}=\overline{\sigma}_0^{(k)}=0$, we have $\widetilde{Y}_1=0$ almost surely and $\widetilde{Y}_0=0$ almost surely.
Hence $\overline{\mu}^{(m)}=\mathbb{E}[(\widetilde{Y}_1-\widetilde{Y}_0)^m]=0$.

(2)-(ii): 
Since $m>k$, $\|\widetilde{\Delta}\|_m\geq \|\widetilde{\Delta}\|_k$.
By the triangle inequality in $L^k$-space,
$\|\widetilde{\Delta}\|_k \geq |\|\widetilde{Y}_1\|_k-\|\widetilde{Y}_0\|_k|=|\sqrt[k]{\overline{\sigma}_1^{(k)}}-\sqrt[k]{\overline{\sigma}_0^{(k)}}|$.
Raising to the $m$-th power yields the lower bound.

\underline{\emph{Sharpness (lower bound).}}
Take a random variable $S$ with $\mathbb{P}(S=1)=\mathbb{P}(S=-1)=1/2$ and set $\widetilde{Y}_x \coloneqq \sqrt[k]{\overline{\sigma}_x^{(k)}}S$.
Then $\widetilde{\Delta}=(\sqrt[k]{\overline{\sigma}_1^{(k)}}-\sqrt[k]{\overline{\sigma}_0^{(k)}})S$, so the lower bound is attained.

\underline{\emph{Sharpness (upper bound).}}
If $\overline{\sigma}_1^{(k)}=\overline{\sigma}_0^{(k)}=0$, then $\overline{\mu}^{(m)}=0$.
Assume $\overline{\sigma}_1^{(k)}+\overline{\sigma}_0^{(k)}>0$, and fix $M>0$.
We may assume $\overline{\sigma}_1^{(k)}>0$, by exchanging the roles of $0$ and $1$ if necessary.
Choose a constant $A>\sqrt[k]{\overline{\sigma}_1^{(k)}}$ such that $(\overline{\sigma}_1^{(k)}/A^k)(A-\sqrt[k]{\overline{\sigma}_0^{(k)}})^m>M$.

Let $I$ be a random variable with $\mathbb{P}(I=1)=\overline{\sigma}_1^{(k)}/A^k$, and
$S$ be a random variable with $\mathbb{P}(S=1)=\mathbb{P}(S=-1)=1/2$, independent of $I$.
Define $\widetilde{Y}_1 \coloneqq IAS$ and $\widetilde{Y}_0 \coloneqq \sqrt[k]{\overline{\sigma}_0^{(k)}}S$.
Then $\mathbb{E}[\widetilde{Y}_x^k]=\overline{\sigma}_x^{(k)}$ ($x\in\{0,1\}$) and
$\overline{\mu}^{(m)}=\mathbb{E}[\widetilde{\Delta}^m]\geq \mathbb{P}(I=1)(A-\sqrt[k]{\overline{\sigma}_0^{(k)}})^m>M$.

By Lemma~\ref{lem:coupling_to_scm}, these can be realized by SCMs.
Therefore we have the sharpness.
\end{proof}

\OddUnboundedSingle*

\begin{proof}
(1): If $\overline{\sigma}_1^{(k)}=\overline{\sigma}_0^{(k)}=0$ for an even $k$, we have $\widetilde{Y}_1=0$ almost surely and $\widetilde{Y}_0=0$ almost surely.
Hence $\overline{\mu}^{(m)}=\mathbb{E}[(\widetilde{Y}_1-\widetilde{Y}_0)^m]=0$.

(2): Fix $M>0$.

\underline{\emph{The case $k$ is odd.}}
For each $x\in\{0,1\}$, take a random variable $B_x$ with $\mathbb{E}[B_x]=0$ and $\mathbb{E}[B_x^k]=2\overline{\sigma}_x^{(k)}$ (Lemma \ref{lem:exist_odd_moment}).
For $A>0$, define $(U_A,V_A)$ 
by
\begin{equation}
    (U_A,V_A)=\begin{cases}
        (A,-A) & \text{ with probability } 1/3, \\
        (-A,0) & \text{ with probability } 1/3, \\
        (0,A) & \text{ with probability } 1/3.
    \end{cases}
\end{equation}
Then $\mathbb{E}[U_A^k]=\mathbb{E}[V_A^k]=0$ and $\mathbb{E}[(U_A-V_A)^m]=\frac{2^m-2}{3}A^m$.
Define $(\widetilde{Y}_1,\widetilde{Y}_0)$ 
by
\begin{equation}
    (\widetilde{Y}_1,\widetilde{Y}_0)=\begin{cases}
        (B_1,B_0) & \text{ with probability } 1/2, \\
        (U_A,V_A) & \text{ with probability } 1/2.
    \end{cases}
\end{equation}
Then $\mathbb{E}[\widetilde{Y}_x^k]=\overline{\sigma}_x^{(k)}$ and
\begin{equation}
\overline{\mu}^{(m)}=\mathbb{E}[\widetilde{\Delta}^m]=\frac12\mathbb{E}[(B_1-B_0)^m]+\frac{2^m-2}{6}A^m.
\end{equation}
Taking $A$ large yields $\overline{\mu}^{(m)}>M$.
Replacing $(U_A,V_A)$ by $(V_A,U_A)$ yields $\overline{\mu}^{(m)}<-M$.

\underline{\emph{The case $k$ is even.}}
We assume $\overline{\sigma}_1^{(k)}>0$, by exchanging the roles of $0$ and $1$ if necessary.
Take a random variable $S$ with $\mathbb{P}(S=1)=\mathbb{P}(S=-1)=1/2$ and put $W_0 \coloneqq \sqrt[k]{2\overline{\sigma}_0^{(k)}}S$.
Choose a constant $A$ such that $p_A \coloneqq \dfrac{2\overline{\sigma}_1^{(k)}}{A^k(1/2+1/2^k)}\leq \dfrac{2}{3}$.
Let $Z_A$ be a random variable with
\begin{equation}
 \mathbb{P}(Z_A=A)=\frac{p_A}{2},\quad \mathbb{P}\Big(Z_A=-\frac{A}{2}\Big)=p_A \quad \text{ and } \quad \mathbb{P}(Z_A=0)=1-\frac{3p_A}{2}.
\end{equation}
Then $\mathbb{E}[Z_A]=0$, $\mathbb{E}[Z_A^k]=2\overline{\sigma}_1^{(k)}$, and
$\mathbb{E}[Z_A^m]=2\overline{\sigma}_1^{(k)}\dfrac{1/2-1/2^m}{1/2+1/2^k}A^{m-k}$.
Define $(\widetilde{Y}_1,\widetilde{Y}_0)$ by
\begin{equation}
    (\widetilde{Y}_1,\widetilde{Y}_0)=\begin{cases}
        (0,W_0) & \text{with probability } 1/2, \\
        (Z_A,0) & \text{with probability } 1/2.
    \end{cases}
\end{equation}
Then $\mathbb{E}[\widetilde{Y}_x^k]=\overline{\sigma}_x^{(k)}$ and $\overline{\mu}^{(m)}=\mathbb{E}[Z_A^m]\big/2$.
Taking $A$ large yields $\overline{\mu}^{(m)}>M$.
Replacing $Z_A$ by $-Z_A$ yields $\overline{\mu}^{(m)}<-M$.
By Lemma~\ref{lem:coupling_to_scm}, these can be realized by SCMs.
\end{proof}

\section{Numerical Experiments}
\label{sec_app:numerical_experiments}

\begin{sidewaystable}
\centering
\caption{Estimates of bounds of central moments of ICE under SCM (A). MCI denotes the estimates obtained by \citep{Kawakami2025_moments}. 
LB denotes the lower bound, and UB denotes the upper bound.
}
\label{tab:numericalexp_1}
\scalebox{1}{
\begin{tabular}{c|c|c|cc|cc|c}
\hline
\multirow{2}{*}{Target $m$} & \multirow{2}{*}{Theorem} & \multirow{2}{*}{\makecell{Available $(\overline{\sigma}^{(k)}_1,\overline{\sigma}^{(k)}_0)$}} & \multicolumn{2}{c|}{$N=20$}  & \multicolumn{2}{c|}{$N=1000$} & \multirow{2}{*}{True}\\
\cline{4-7}
 &  &  & \multicolumn{1}{c}{LB} & \multicolumn{1}{c|}{UB}  & \multicolumn{1}{c}{LB} & \multicolumn{1}{c|}{UB} & \\
\hline
\hline
\multirow{4}{*}{2} & \ref{thm:bound_of_var} & 2 & 0.092 [0.000,0.456] & 1.128 [0.317,2.088] & 0.005 [0.001,0.013] & 1.492 [1.386,1.600] & 0.083 \\
 & \ref{thm:mu2_from_3} & 3 & 0.000 [0.000,0.000] & $\infty$  & 0.000 [0.000,0.000] & $\infty$  & 0.083 \\
 & \ref{thm:mu2_from_4} & 4 & 0.000 [0.000,0.000] & 1.603 [0.523,2.737] & 0.000 [0.000,0.000] & 2.157 [2.051,2.270] & 0.083 \\
 & MCI & - & 0.004 [0.000,0.044] & 1.020 [0.296,1.904] & 0.000 [0.000,0.000] & 1.410 [1.243,1.576] & 0.083 \\
\hline
\multirow{4}{*}{3} & \ref{thm:mu3_from_2} & 2 & $-\infty$  & $\infty$  & $-\infty$  & $\infty$  & 0.000 \\
 & \ref{thm:mu3_from_3} & 3 & $-\infty$  & $\infty$  & $-\infty$  & $\infty$  & 0.000 \\
 & \ref{thm:mu3_from_4} & 4 & -1.328 [-2.810,-0.235] & 1.328 [0.235,2.810] & -1.966 [-2.122,-1.822] & 1.966 [1.822,2.122] & 0.000 \\
 & MCI & - & -0.981 [-2.157,-0.201] & 0.450 [0.021,1.338] & -1.613 [-1.840,-1.395] & 0.629 [0.475,0.814] & 0.000 \\
\hline
\multirow{4}{*}{4} & \ref{thm:mu4_from_2} & 2 & 0.023 [0.000,0.208] & $\infty$  & 0.000 [0.000,0.000] & $\infty$  & 0.013 \\
 & \ref{thm:mu4_from_3} & 3 & 0.000 [0.000,0.000] & $\infty$  & 0.000 [0.000,0.000] & $\infty$  & 0.013 \\
 & \ref{thm:mu4_from_4} & 4 & 0.058 [0.000,0.455] & 2.937 [0.274,7.494] & 0.000 [0.000,0.001] & 4.657 [4.207,5.152] & 0.013 \\
 & MCI & - & 0.000 [0.000,0.000] & 2.240 [0.233,5.663] & 0.000 [0.000,0.000] & 3.980 [3.355,4.728] & 0.013 \\
\hline
\end{tabular}
}
\vspace{1cm}
\caption{Estimates of bounds of central moments of ICE under SCM (B). MCI denotes the estimates obtained by \citep{Kawakami2025_moments}. 
LB denotes the lower bound, and UB denotes the upper bound.}
\label{tab:numericalexp_2}
\scalebox{1}{
\begin{tabular}{c|c|c|cc|cc|c}
\hline
\multirow{2}{*}{Target $m$} & \multirow{2}{*}{Theorem} & \multirow{2}{*}{\makecell{Available  $(\overline{\sigma}^{(k)}_1,\overline{\sigma}^{(k)}_0)$}} & \multicolumn{2}{c|}{$N=20$} & \multicolumn{2}{c|}{$N=1000$} & \multirow{2}{*}{True}\\
\cline{4-7}
 &  & & \multicolumn{1}{c}{LB} & \multicolumn{1}{c|}{UB} & \multicolumn{1}{c}{LB} & \multicolumn{1}{c|}{UB} & \\
\hline
\hline
\multirow{4}{*}{2} & \ref{thm:bound_of_var} & 2 & 0.858 [0.604,0.978] & 1.021 [0.746,1.142] & 0.887 [0.879,0.895] & 1.117 [1.107,1.125] & 0.903 \\
 & \ref{thm:mu2_from_3} & 3 & 0.000 [0.000,0.000] & $\infty$  & 0.000 [0.000,0.000] & $\infty$  & 0.903 \\
 & \ref{thm:mu2_from_4} & 4 & 0.000 [0.000,0.000] & 1.144 [1.024,1.286] & 0.000 [0.000,0.000] & 1.139 [1.131,1.148] & 0.903 \\
 & MCI & - & 0.086 [0.000,0.147] & 0.843 [0.412,1.117] & 0.021 [0.001,0.055] & 1.071 [0.990,1.140] & 0.903 \\
\hline
\multirow{4}{*}{3} & \ref{thm:mu3_from_2} & 2 & $-\infty$  & $\infty$  & $-\infty$  & $\infty$  & 0.000 \\
 & \ref{thm:mu3_from_3}  & 3 & $-\infty$  & $\infty$  & $-\infty$  & $\infty$  & 0.000 \\
 & \ref{thm:mu3_from_4} & 4 & -0.760 [-0.904,-0.643] & 0.760 [0.643,0.904] & -0.754 [-0.763,-0.746] & 0.754 [0.746,0.763] & 0.000 \\
 & MCI & - & -0.403 [-0.606,-0.101] & 0.405 [0.123,0.601] & -0.555 [-0.650,-0.458] & 0.556 [0.465,0.646] & 0.000 \\
\hline
\multirow{4}{*}{4} & \ref{thm:mu4_from_2} & 2 & 0.745 [0.365,0.956] & $\infty$  & 0.787 [0.772,0.801] & $\infty$  & 0.819 \\
 & \ref{thm:mu4_from_3} & 3 & 0.000 [0.000,0.000] & $\infty$  & 0.000 [0.000,0.000] & $\infty$  & 0.819 \\
 & \ref{thm:mu4_from_4} & 4 & 0.904 [0.730,1.153] & 1.313 [1.048,1.653] & 0.761 [0.750,0.774] & 1.298 [1.280,1.318] & 0.819 \\
 & MCI & - & 0.000 [0.000,0.001] & 0.792 [0.286,1.248] & 0.000 [0.000,0.000] & 1.158 [0.960,1.351] & 0.819 \\
\hline
\end{tabular}
}
\end{sidewaystable}

\begin{table}[!tb]
\centering
\caption{Estimates of central moments of ICE under IED condition for SCM (C) based on Theorem \ref{thm:identification_under_independence}. 
}
\label{tab:numericalexp_3}
\scalebox{1}{
\begin{tabular}{c|c|c|c}
\hline
Target & $N=20$ & $N=1000$ & True \\
\hline
\hline
2 & 0.412 [0.000,0.903] & 0.335 [0.266,0.406] & 0.333 \\
3 & -0.010 [-0.500,0.455] & 0.000 [-0.072,0.072] & 0.000 \\
4 & 0.461 [0.000,1.596] & 0.205 [0.141,0.278] & 0.200 \\
\hline
\end{tabular}
}
\end{table}

In this appendix, we conduct numerical experiments to illustrate the properties of our results in a finite-sample setting.\\

{\bf Baseline.}
We compare our bounds with those given by the Monte Carlo integration-based estimator (MCI) \citep{Kawakami2025_moments}.
Note that only the upper bounds of MCI are sharp when $m$ is even.
\citep{Fan2010, Post2023, Post2025} do not provide results for moments of causal effects; thus, there is no baseline other than \citep{Kawakami2025_moments}.
Our estimator uses only the marginal central moments of each PO, whereas MCI requires the full marginal distribution of each PO.

{\bf Setting.}
We work with the following SCMs:
\begin{gather}
\text{SCM(A)}:\begin{cases}
    Y\coloneqq -(X+1)U\mathbbm{1}(XU\geq 0), \\
    X\sim \mathrm{Bern}(0.8),\,U\sim \mathrm{Unif}(-1,1)
\end{cases} \\
\text{SCM(B)}:\begin{cases}
    Y\coloneqq X\mathrm{sign}(U)+0.1(1-X)U, \\
    X\sim \mathrm{Bern}(0.8),\,U\sim \mathrm{Unif}(-1,1)
\end{cases} \\
\text{SCM(C)}:\begin{cases}
    Y=U_1+XU_2, \\
    U_1\sim \mathrm{Unif}(-1,1),\,U_2\sim \mathrm{Unif}(-1,1),\\
    X\sim \mathrm{Bern}(0.8),\,U_1\indep U_2
\end{cases}
\end{gather}
where $\mathrm{Bern}(p)$ is a Bernoulli distribution with probability $p$, and $\mathrm{Unif}(-1,1)$ is a uniform distribution over $[-1,1]$.
SCM (A) is the setting used in the experiments of \citet{Kawakami2025_moments}.
SCM(A) and (B) do not satisfy the IED condition, and the SCM(C) satisfies the IED condition.
{We estimate the bounds by plugging in the empirical central moments of each potential outcome. For even-order moments, if the plug-in estimate is negative, we set it to zero.}
We conduct 1,000 simulation runs for sample sizes $N=20, 1000$.
We report the means of the estimates and their 95\% confidence intervals (CIs).

{\bf Results.}
From Table~\ref{tab:numericalexp_1}, in setting (A), when the sample size is relatively large ($N = 1000$), the tightest bounds for the second and fourth central moments of ICE obtained by our method—using $(\overline{\sigma}_1^{(2)}, \overline{\sigma}_0^{(2)})$ and $(\overline{\sigma}_1^{(4)}, \overline{\sigma}_0^{(4)})$, respectively—are similar to those of MCI, with comparable CIs.
In contrast, the tightest bound for the third central moment of ICE using $(\overline{\sigma}_1^{(4)}, \overline{\sigma}_0^{(4)})$ is slightly wider than those of MCI.
When the sample size is relatively small ($N = 20$), the estimated bounds from our method have relatively wide CIs compared to those of MCI.

In setting (B), from Table~\ref{tab:numericalexp_2}, the tightest bounds for the second and fourth central moments of ICE obtained by our method—using $(\overline{\sigma}_1^{(2)}, \overline{\sigma}_0^{(2)})$ and $(\overline{\sigma}_1^{(4)}, \overline{\sigma}_0^{(4)})$, respectively—are narrower than those of MCI.

In setting (C), as shown in Table~\ref{tab:numericalexp_3}, when the sample size is relatively large ($N = 1000$), we obtain reliable estimates of the second, third, and fourth central moments of ICE, as indicated by the narrow CIs.
In contrast, when the sample size is relatively small ($N = 20$), the estimates of the second, third, and fourth central moments of ICE are unstable, with wide CIs.

\begin{remark*}[Confidence intervals for plug-in estimates.]
The empirical case studies report plug-in estimates because the original articles provide only point summaries. This appendix reports confidence intervals in simulated finite-sample settings. If confidence intervals for marginal moments are available in empirical applications, they can be propagated through the formulas in this paper. For example, under IED, $\overline{\mu}^{(2)}=\overline{\sigma}^{(2)}_1-\overline{\sigma}^{(2)}_0$. If $L^1\leq\overline{\sigma}^{(2)}_1\leq U^1$ and $L^0\leq\overline{\sigma}^{(2)}_0\leq U^0$ hold simultaneously, then
\begin{equation}
L^1-U^0\leq \overline{\mu}^{(2)}\leq U^1-L^0.
\end{equation}
If the two marginal intervals each have coverage at least $\alpha$, then the simultaneous coverage is at least $2\alpha-1$, since $\mathbb{P}(A\cap B)\geq 1-\mathbb{P}(A^c)-\mathbb{P}(B^c)$. The same argument applies to the bounds derived in this paper.
\end{remark*}

\end{document}